%% file: rr-main.tex
\documentclass[hidelinks,12pt]{article}
\RequirePackage{amsthm,amsmath,amsfonts,amssymb,centernot,float,import,makeidx,subfiles,threeparttable}
\usepackage{xr-hyper}
\usepackage{lscape}
\usepackage[pdftex]{graphicx}
\usepackage{hyperref}
\usepackage{longtable}
\usepackage{multirow}
\usepackage{breqn}
\usepackage{array}
\usepackage{booktabs}
\usepackage{caption}
\usepackage{dsfont}
\usepackage[nokeyprefix]{refstyle}
\usepackage{varioref}
\usepackage[english]{babel}
\usepackage{tikz}
\usepackage{float}
\usetikzlibrary{positioning, calc, shapes.geometric, shapes, shapes.multipart, arrows.meta, arrows, decorations.markings, external, trees}
\usepackage{cleveref}

\newlength\myindent
\tikzstyle{Arrow} = [
	thick, 
	decoration={
		markings,
		mark=at position 1 with {
			\arrow[thick]{latex}
			}
		}, 
	shorten >= 3pt, preaction = {decorate}
	]
\theoremstyle{plain}

\newtheorem{proposition}{Proposition}
\newtheorem{assumption}{Assumption}

\theoremstyle{remark}
\newtheorem{remark}{Remark}

\newcommand{\blind}{0}

\begin{document}

\def\spacingset#1{\renewcommand{\baselinestretch}%
{#1}\small\normalsize} \spacingset{1}

%%%%%%%%%%%%%%%%%%%%%%%%%%%%%%%%%%%%%%%%%%%%%%%%%%%%%%%%%%%%%%%%%%%%%%%%%%%%%%

\if0\blind
{
  \title{\bf The effect of COVID-19 vaccinations on self-reported depression and anxiety during February 2021}
\author{Max Rubinstein\thanks{
    The authors gratefully acknowledge invaluable advice from discussions with Vinni Bhatia, Nate Breg, David Choi, Riccardo Fogliato, Joel Greenhouse, Edward Kennedy, Samantha Patel, Alex Reinhart, and Robin Mejia.} \\ Heinz College and Department of Statistics \& Data Science \\Carnegie Mellon University\and
  Amelia Haviland \\ Heinz College \\Carnegie Mellon University \and
  Joshua Breslau\\ RAND Corporation}
  \date{\vspace{-5ex}}
  \maketitle
} \fi

\if1\blind
{
  \bigskip
  \bigskip
  \bigskip
  \begin{center}
    {\LARGE\bf Title}
\end{center}
  \medskip
} \fi

\bigskip
\begin{abstract}
Using the COVID-19 Trends and Impact survey, we find that COVID-19 vaccinations reduced the prevalence of self-reported feelings of depression and anxiety, isolation, and worries about health among vaccine-accepting respondents in February 2021 by 3.7, 3.3, and 4.3 percentage points, respectively, with particularly large reductions among respondents aged 18 and 24 years old. We show that interventions targeting social isolation account for 39.1\% of the total effect of COVID-19 vaccinations on depression, while interventions targeting worries about health can account for 8.3\%. This suggests that social isolation is a stronger mediator of the effect of COVID-19 vaccinations on depression than worries about health. We caution that these causal interpretations rely on strong assumptions.
\end{abstract}

\noindent%
{\it Keywords:}  Mental health, COVID-19, vaccinations
\vfill

\newpage
\spacingset{1.45} % DON'T change the spacing!

\input{paper}

\bibliographystyle{plain}
\bibliography{research}  
\newpage

\input{appendices/01-sample-chars.tex}

\clearpage

\input{appendices/02-estimands.tex}

\clearpage

\input{appendices/03-identification.tex}

\clearpage

\input{appendices/04-estimation.tex}

\clearpage

\input{appendices/05-results.tex}

\clearpage

\input{appendices/06-sensitivity.tex}

\clearpage

\input{appendices/07-robustness.tex}

\end{document}

%% file: paper.tex
\section{Introduction}

The COVID-19 pandemic and the associated social isolation, economic hardship associated with job loss, and social uncertainty led to an increase psychological distress (\cite{pfefferbaum2020mental}). For example, \cite{breslau2021longitudinal} showed that the same number of people experienced serious psychological distress in May 2021 as the entire year prior to the pandemic. \cite{santomauro2021global} estimated that the pandemic caused a 27.6\% increase in cases of major depressive disorder and a 25.6\% increase in cases of anxiety disorders globally. Younger adults have consistently been found to be among the most adversely subgroups (\cite{santomauro2021global}, \cite{sojli2021covid}). Secondary consequences of these mental health impacts, such as suicide and drug overdoses, are also of concern (see, e.g., \cite{kawohl2020covid}, \cite{glober2020impact}).

If COVID-19 pandemic worsened mental health, then perhaps receiving a COVID-19 vaccination improved it. This might happen, for example, by decreasing worries about health or social isolation. We investigate this hypothesis by estimating the effects of COVID-19 vaccinations on the rates of depression and anxiety, feelings of social isolation, and worries about health in February 2021. Our analysis is based on the COVID-19 Trends and Impacts Survey (CTIS) (\cite{salomon2021us}), a cross-sectional survey designed by the Delphi group at Carnegie Mellon University and administered in collaboration with Facebook. We quantify the average effect of vaccinations on these outcomes among a subset of vaccine-accepting Facebook users during February 2021. By vaccine-accepting we mean CTIS respondents who indicated that they either would receive a COVID-19 vaccine if offered one that day, or who indicated that they had already received a COVID-19 vaccine. By depression and anxiety we mean self-reported feelings, not a clinical diagnosis. A key assumption for the causal interpretation of this analysis is that conditional on vaccine-acceptance and observed covariates, vaccinated and unvaccinated respondents differed on average only in vaccination status.

We also examine effect heterogeneity and mediation. For our heterogeneity analysis, we consider both pre-specified and data-driven subgroups, using an approach similar to \cite{athey2016recursive}. For our mediation analysis, we consider a model where COVID-19 vaccinations affect depression and anxiety through feelings of social isolation, worries about health, and a direct effect capturing all other channels. We allow the causal ordering of the mediators to be unknown; however, we assume that they are determined prior to the outcome, and that no other effects of vaccinations affect these mediators. Using the interventional effect decomposition proposed by \cite{vansteelandt2017interventional}, we estimate the proportion of the total effect on depression via each mediator. While interventional indirect effects do not necessarily reflect mechanisms without stronger assumptions (\cite{miles2022causal}), this is the first study we are aware of that attempts to estimate the effects of COVID-19 vaccinations on depression via intermediate pathways.

Section~\ref{sec:literature} presents a brief review of the related literature and our contributions. Section~\ref{sec:data} describes our data and provides summary statistics. Section~\ref{sec:methods} outlines our identification, estimation, and inferential strategy. Section~\ref{sec:results} presents our results. Section~\ref{sec:limitations} considers the sensitivity of our results to our causal assumptions and their robustness to analytic choices, and Section~\ref{sec:discussion} provides discussion and conclusion. Additional materials are available in the Appendices.

\section{Related Literature}\label{sec:literature}

Two concurrent studies examined the effect of COVID-19 vaccinations on mental health: \cite{perez2021covid} and \cite{agrawal2021impact}. \cite{perez2021covid} use longitudinal data from the Understanding America Survey and fixed-effects models to estimate that vaccines caused a 4\% reduction of the probability of being at least mildly depressed and a 15\% reduction in the probability of being severely depressed. When testing for effect heterogeneity, they were unable to rule out that the effects were the same across sex, race, and education level. \cite{agrawal2021impact} use repeated cross-sectional survey data from Census Pulse. Using heterogeneity in state-level eligibility requirements as an instrumental variable, they use two-stage least squares to estimate that COVID-19 vaccination reduces anxiety and depression symptoms by nearly 30\%. They find larger effect sizes among individuals with lower education levels, who rent their housing, who are unable to telework, and who live with children. 

Our identification strategy instead hinges on an assumption that sample selection and COVID-19 vaccination are independent of potential outcomes given vaccine-acceptance and observed covariates. This assumption precludes confounding at the population-level and biases induced by conditioning on our sample. We also exploit auxiliary data to identify potential high impact subgroups hence mitigating the risk of false discoveries. Finally, we estimate interventional indirect effects to better understand the mechanisms through which COVID-19 vaccinations affect depression and anxiety. By contrast, the previous studies only identify average effects or local average effects, do not use sample-splitting for their heterogeneity analyses, and do not address mediation beyond speculation. A key contribution of our paper is therefore our mediation analysis. This analysis hinges on the idea that worries about health and social isolation are the two primary pathways through which vaccinations would change depression and anxiety. These are commonly cited as two major ways that the pandemic has adversely impacted mental health (\cite{robb2020associations}, \cite{gonccalves2020preliminary}, \cite{corbett2020health}).

\section{Data}\label{sec:data}

Our primary dataset is the COVID-19 Trends and Impact Survey (CTIS), created by the Delphi group at Carnegie Mellon University in collaboration with Facebook. From April 2020 through June 2022, Facebook offered the CTIS survey using a stratified random sample (strata defined by geographic region) of approximately 100 million US residents from the Facebook Active User Base who use one of the supported languages: English, Spanish, French, Brazilian Portuguese, Vietnamese, and simplified Chinese. A link to participate in the survey was displayed at the top of users' News Feed. Users received this offer at most once per month depending on their geographic location. The goal was to obtain approximately one million monthly responses. The CTIS evolved over time, but all waves of the survey instrument include questions about COVID symptoms, health behaviors, mental health, and demographics. A complete list of the questions over time is available \href{https://cmu-delphi.github.io/delphi-epidata/symptom-survey/coding.html}{online}. \cite{salomon2021us} provides more details about the survey.  

We analyze the period from February 1-28, 2021\footnote{This period encompasses both CTIS waves 7 and 8. Wave 8 was deployed on February 8, 2021, and incorporated changes to some of the questions we use as covariates. These include the addition of categories to the chronic health conditions and occupation questions, and changing a question on the previous receipt of flu vaccine to be defined from July 2020 instead of June 2020.} We choose this timeframe to increase the plausibility of our identifying assumptions, which we discuss further in Section~\ref{sec:timeframe}. Our primary outcome ($Y$) is whether a respondent reported feeling depressed or anxious ``most'' or ``all'' of the time in the previous five days. We hereafter refer to this variable as ``depressed.'' Our treatment variable ($A$), is an indicator of whether the respondent ever received a COVID-19 vaccine.\footnote{We recode the response ``I don't know'' to be grouped with people who indicated they did not already receive a COVID-19 vaccine.} Importantly, the CTIS does not tell us \textit{when} a respondent was vaccinated. We also consider responses to how often a respondent felt isolated from others in the previous five days ($M_1$), and how worried respondents reported feeling about themselves or an immediate family member becoming seriously ill from COVID-19 ($M_2$), each containing four levels.\footnote{The CTIS questions specifically are ``In the past 5 days, how often have you felt isolated from others?" and ``How worried are you that you or someone in your immediate family might become seriously ill from COVID-19 (coronavirus disease)?''} We hereafter refer to these variables as ``isolated'' and ``worries about health,'' respectively. We also create the variable $M$ as the sixteen level joint variable $M_1 \times M_2$. Finally, we dichotomize $M_1$ and $M_2$ when analyzing them as outcomes. Specifically, we create indicators for the responses ``most'' or ``all'' of the time for isolation, and ``very worried'' for worries about health. Whenever referring to our outcomes analyses, for simplicity we refer to any outcomes as $Y$.

We consider several covariates recorded in the CTIS pertaining to demographic, household, employment, and health status. For demographics we consider each respondent's age category, race, ethnicity, gender, and highest education attained. For household variables, we construct indicators of whether a respondent lives alone; lives with a child attending school in-person full-time, part-time, or not in school; or lives with an elderly (over 65) individual. For employment, we record whether each respondent worked for pay in the previous 30 days, and whether their work was at home or outside. Among respondents who work, we record the respondent's occupation type, including healthcare, education, service-industry workers, protective services, other employment sector. Finally, we record information regarding each respondent's health and health behaviors, including whether they have previously been tested for and/or diagnosed with COVID-19, how many ``high-risk'' health conditions they have ever been diagnosed with, (0, 1, 2, or three or more), and whether they received a flu shot since April 2020 (included as a proxy for health behaviors). We also record whether the respondent took the survey in English or Spanish. For all of the CTIS survey-measured covariates, we recode missingness as a separate category to account for item non-response. Lastly, we augment the dataset with indicators of each respondent's state and week of response. The state indicators can control for fixed state-level policies during the study period that may be associated with both vaccine access and depression, including eligibility requirements, while the week indicators control for time trends associated with both treatment and outcomes. 

We also join several county-level variables using each respondents' FIPS code. Using data from the 2019 American Community Survey, we add the county-level population, population density, percent living in poverty, percent uninsured, and Gini index. Using the National Center for Health Statistics Urban-Rural classification scheme, we classify the type of county each respondent lives in (Large Central Metro, Large Fringe Metro, Medium Metro, Small Metro, Micropolitan, and Non-core regions). These variables help control for differential COVID-19 vaccine access. We also include the county the total number of deaths per capita from COVID-19 in each county the final two weeks in January, and divide this variable into sextiles (data obtained from Johns Hopkins University). Finally, we include Biden's county vote share at the county-level (Biden minus Trump vote totals divided by total votes). We denote all demographic and county-level covariates as $X$.

Our initial dataset from February includes 1,232,398 responses. We exclude 2,530 respondents whose Facebook user language was not English or Spanish, 3,690 responses from Alaska -- Alaska did not report 2020 voting results at the county-level -- and 39,526 responses that did not provide FIPS code information -- preventing us from merging county-level covariates. We then subset our data to only include respondents who say that they ``probably'' or ``definitely'' would get the COVID-19 vaccine if offered to them today or who reported already being vaccinated ($V = 1$). This left 881,315 responses, excluding 229,894 vaccine-hesitant respondents and 75,443 respondents who did not indicate whether or not they intended to receive a COVID-19 vaccine. Finally, 758,594 respondents on our primary sample, or 86.1 percent, provided answers to all questions pertaining to vaccination and all three outcomes ($R = 1$). This defines our analytic sample. Additional details are available in Appendix~\ref{app:samplechars}.

\subsection{Summary statistics}\label{sec:sumstats}

\Cref{fig:vaccinated} displays the weekly prevalence of depression among all CTIS respondents. The plot displays a spike in May 2020 followed the police murder of George Floyd and the associated protests, another spike in early November 2020 during the presidential election, and generally decreasing depression afterwards. One exception occurred during the week of January 6th, when Donald Trump instigated an insurrection on the U.S. Capitol following his loss in the 2020 election.

\begin{figure}[H]
\begin{center}
    \includegraphics[scale = 0.5]{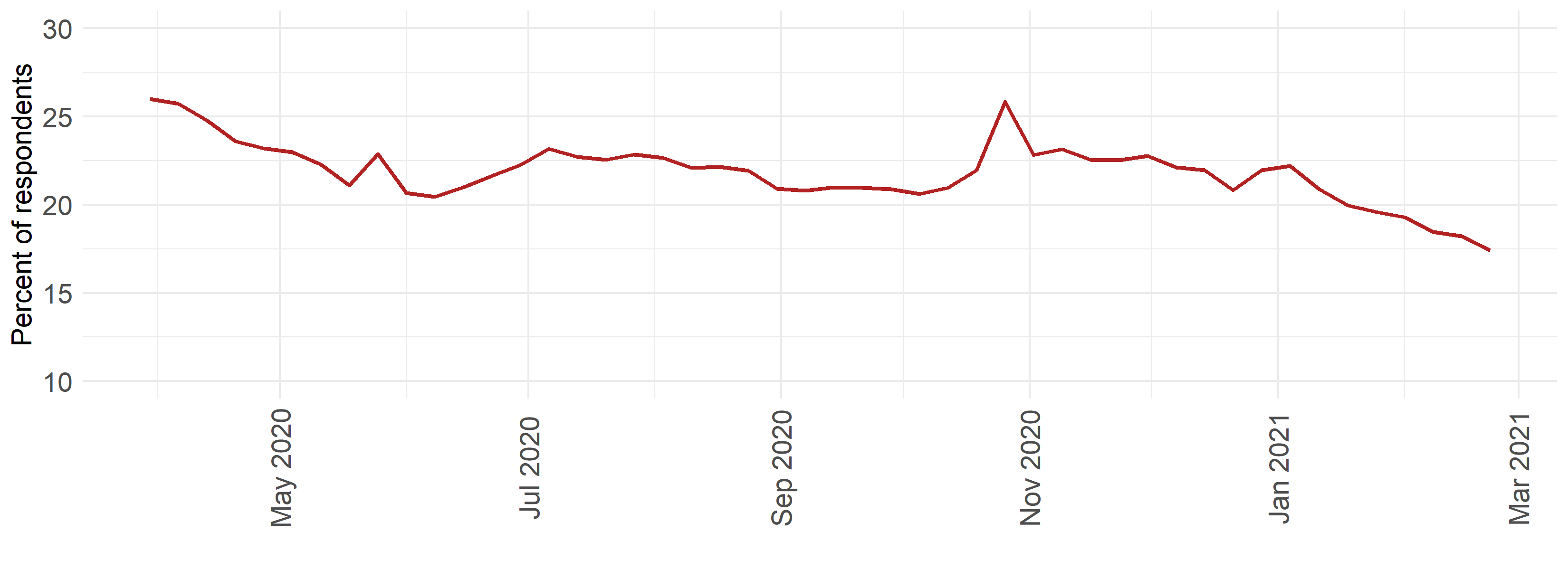}
    \caption{Trends in depression and anxiety, April 2020 - February 2021}
    \label{fig:vaccinated}
\end{center}
\end{figure}

The prevalence of depression varies substantially by age group. The left-hand panel of Figure~\ref{fig:alltrends} shows the weekly prevalence from October 2020 through February 2021, with rates increasing monotonically across seven age groups. This may depict a real negative correlation between depression and age; however, an alternative interpretation could be that younger individuals have a different vocabulary around mental health than older individuals and respond to these questions differently. While possible, Figure~\ref{fig:alltrends} displays similar patterns for isolation and worries about health. Perhaps paradoxically, the percentage of respondents who report being very worried about their health is lowest among the older respondents and generally highest among 18-34 year old respondents.\footnote{Similar patterns were observed in other surveys: for example, an \href{https://www.economist.com/graphic-detail/2020/09/01/why-young-americans-are-more-worried-than-older-ones-about-covid-19}{Economist/YouGov poll} showed similar age patterns with respect to worries about contracting COVID-19. \cite{clair2021effects} also found higher levels of social isolation among younger individuals.}

\begin{figure}[H]
\begin{center}
    \includegraphics[scale=0.5]{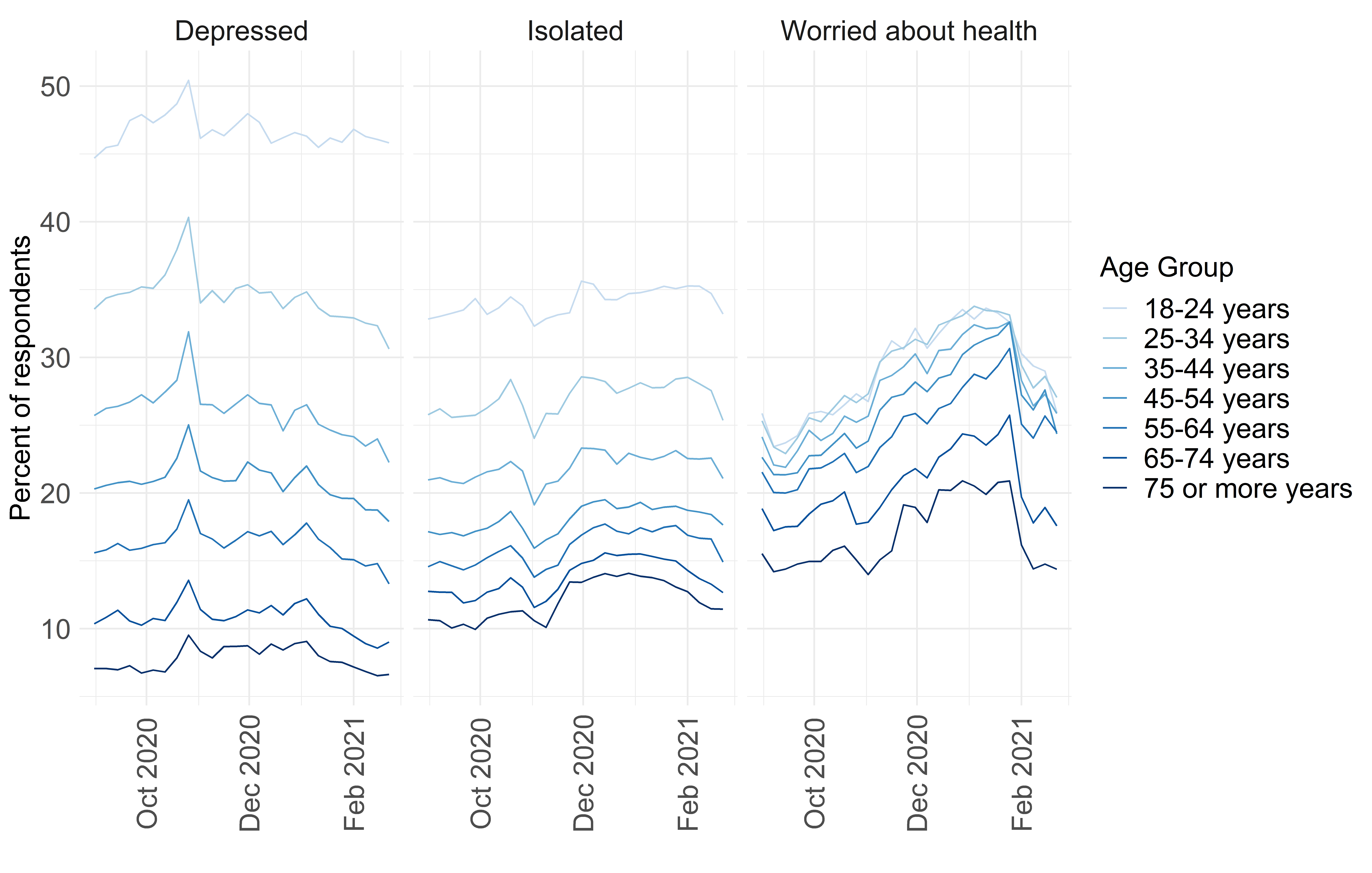}
    \caption{Mental health trends, October 2020 - February 2021}
    \label{fig:alltrends}
\end{center}
\end{figure}

Each outcome declined in prevalence throughout February 2021, with particularly large declines in worries about health. This may reflect seasonal trends and the fact that the number of new COVID-19 infections was falling substantially.\footnote{See, e.g., https://www.nytimes.com/interactive/2021/us/covid-cases.html} However, this timeframe also coincided with substantial numbers of people receiving their COVID-19 vaccinations. Table~\ref{tab:schars} below illustrates number of individuals and the percentage vaccinated within each age category among all CTIS respondents in February 2021.\footnote{As has been observed elsewhere, the CTIS over-represents vaccinated individuals compared to the U.S. population (\cite{salomon2021us}). This does not necessarily threaten the internal validity of our analysis (see Section~\ref{ssec:identification}), though may threaten our ability to generalize our results.} 

\input{tables/sample-chars-paper.tex}

Our study examines whether COVID-19 vaccinations in part account for these improving outcomes during February 2021. We begin by examining the differences in the prevalence of each outcome in vaccinated versus unvaccinated respondents that month, taking these averages within separate subgroups. Figure~\ref{fig:htesunadjusted} displays these results. 

\begin{figure}[H]
\begin{center}
    \includegraphics[scale=0.5]{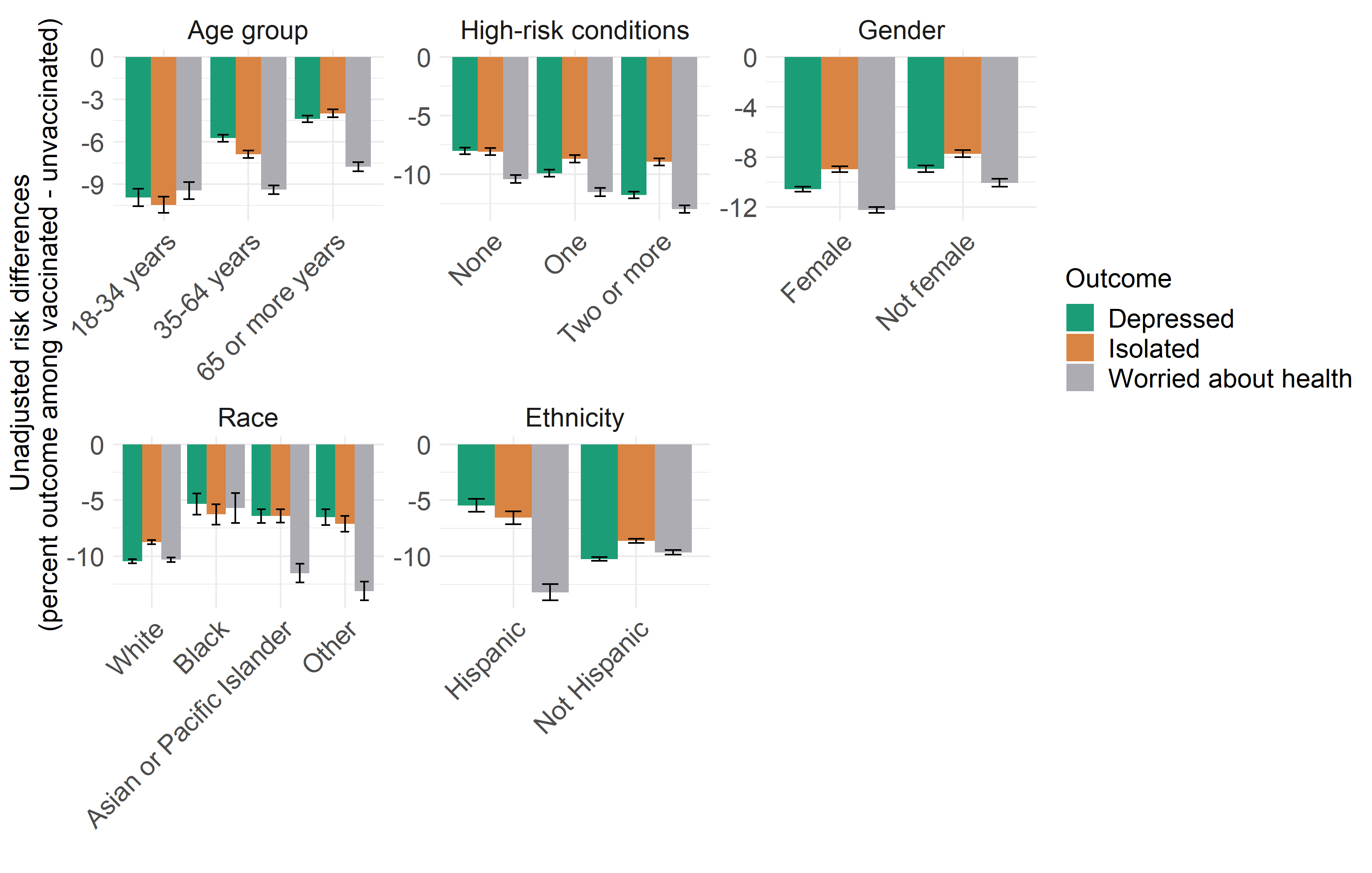}
    \caption{Associations between mental health outcomes and vaccination, February 2021}
    \label{fig:htesunadjusted}
\end{center}
\end{figure}

Given the high base-rates of depression and isolation among the youngest respondents, perhaps not surprisingly we observe the largest declines in the prevalence of each outcome among younger respondents relative to older ones. We also see larger declines in rates of depression and isolation among White relative to Black respondents, and non-Hispanic relative to Hispanic respondents. When examining worries about health by each demographic category, the declines appear lower among the the oldest age groups, higher among people with three or more high-risk conditions, and higher among Hispanic versus non-Hispanic respondents. Broadly speaking, patterns in depression and isolation appear more similar to each other than to patterns in worries about health. 

Of course, COVID-19 vaccinations were not randomly assigned, and we cannot make causal claims from these associations. For example, individuals who were most at-risk and who had the means were most likely to receive the vaccination. If these same individuals were also more likely to be depressed or anxious without the vaccine, then these associations would overstate the absolute magnitude of any treatment effects. We therefore adjust these estimates to control for the demographic, household, employment, health, and county-characteristics described previously, as well as state and week indicators, and outline the assumptions necessary for these adjusted associations to reflect causal quantities. Additional summary statistics are available in Appendix~\ref{app:samplechars}.

\section{Methods}\label{sec:methods}

\subsection{Preliminaries}

Recall that $A$ indicates prior receipt of COVID-19 vaccination, $M_1$ isolation, $M_2$ worries about health, and $Y$ depression. $X$ indicates the covariate vector described in Section~\ref{sec:data}, and $V$ indicates vaccine-acceptance. We define $S$ as an indicator of responding to the survey, and $(R_Y, R_A, R_{M_1}, R_{M_2})$ as indicators of observing $(Y, A, M_1, M_2)$, respectively. We observe $n$ independent samples from a pool of active Facebook users residing in the United States of the random variable $O = (YR_Y, AR_A, M_1R_{M_1}, M_2R_{M_2}, R_Y, R_A, R_{M_1}, R_{M_2}, X, V, S = 1)$. We define two important functions of this distribution. First, we let $R = R_YR_{M_1}R_{M_2}R_A$. This is an indicator of whether we observe all outcomes, mediators, and treatment assignment information (hereafter ``complete-cases'') for a respondent. Second, define $Z = RSV$. In other words, $Z$ is an indicator of inclusion in our analytic sample.

\subsection{Estimands}\label{ssec:estimand}

We next define causal estimands with respect to our outcomes and mediation analyses using potential outcomes notation. Before defining these quantities, we make the following assumptions. First, we invoke the SUTVA. This implies that there is only one version of treatment and rules out interference between individuals: specifically, each individual's potential outcomes only depend on their own vaccination status. Second, we assume no carryover effects. This implies that an individual's potential outcome having been vaccinated one week prior to responding to the CTIS is identical to their potential outcome were they vaccinated one month prior. Finally, we assume no anticipatory effects, implying that future vaccination status does not affect one's potential outcomes at the time an individual responded to the CTIS.

Under these assumptions, we define the potential outcomes for each individual as $Y_i^{a_i}$: that is, an individual's outcome when setting their vaccination status $A_i$ to level $a_i$. This quantity does not depend on anyone else's vaccination status (SUTVA), their past vaccination status (no carryover effects), nor their future vaccination status (no anticipatory effects). Finally, we define the function $\mu_a(X_i) = \mathbb{E}[Y_i^a \mid X_i, Z_i = 1]$, representing the expected potential outcome under treatment $A_i = a$ for an individual given their covariates, vaccine-acceptance, and being in our analytic sample.

\subsubsection{Outcomes}\label{sssec:estimand:outcomes}

We use $Y$ to denote any outcome considered in this section to ease notation, recalling that $M_1$ and $M_2$ are dichotomized for this analysis. Assume that there are $N$ adults living in the United States. A natural estimand is the average treatment effect:

\begin{align}\label{eqn:estimand:ATE}
    \psi^{ate} = N^{-1}\sum_{i=1}^N[Y_i^1 - Y_i^0]
\end{align}

This reflects the mean difference in the prevalence of each outcome if everyone were vaccinated contrasted against the case where no one was vaccinated, averaging over the empirical covariate distribution of adults residing in the United States.

This causal quantity is difficult to target for at least four reasons. First, the U.S. adult population consists of two types of individuals: the vaccine-hesitant and the vaccine-accepting. The vaccine-hesitant constitute a group of ``never-takers,'' a group would never choose be vaccinated (at least during this timeframe). We cannot estimate the treatment effect among this subgroup without strong modeling assumptions. Second, the sample is taken over the population of Facebook users who reside in the United States. This population differs from the U.S. adult population, the arguably more interesting population. Third, the average over the population of Facebook users is also difficult to credibly estimate, as survey response rates in February 2021 were approximately one percent. Unless we assume that survey non-response was completely random, we cannot estimate this quantity without further assuming, for example, that the selection process was based on pre-treatment characteristics and then modeling this process. Finally, among the survey respondents, we only observe the vaccination status and outcomes among those who respond to these items. Similar to survey non-response, generalizing the causal effects to include these individuals requires additional assumptions and modeling.

We make several analytic choices to address these challenges. First, we condition our analysis among the vaccine-accepting population. Second, our primary estimands are conditional on our analytic sample. Specifically, we define our target estimand as:\footnote{This expression slightly simplifies of our true targeted estimand. As we explain in Section~\ref{sec:data}, we also remove those with missing FIPS code information and those who live in Alaska.}

\begin{align}
\label{eqn:estimand:wate}\psi &= \frac{1}{\sum_{i = 1}^NZ_i}\sum_{i: Z_i = 1} \mathbb{E}[Y_i^1 - Y_i^0 \mid X_i, Z_i = 1] \\
\label{eqn:estimand:wate1}&= \mathbb{P}_n[\mu_1(X_i) - \mu_0(X_i)]
\end{align}

This estimand is the average expected difference in the prevalence of the outcomes given the empirical covariate distribution of the vaccine-accepting complete-cases in our analytic sample. We abbreviate this empirical average using $\mathbb{P}_n$ moving forward. We also examine all effects on both the risk-ratio scale ($\mathbb{P}_n[\mu_1(X)]/\mathbb{P}_n[\mu_0(X)]$) and averaged within distinct subgroups defined by the covariates. In Appendix~\ref{app:results-incremental} we consider estimands that do not deterministically set vaccination status, but rather shift each individual's probability of vaccination (\cite{kennedy2019nonparametric}).

\subsubsection{Interventional effects}

We next decompose the effect of vaccinations on depression into four separate channels: a direct effect, an indirect effect through social isolation, an indirect effect through worries about health, and an indirect effect reflecting the dependence of social isolation and worries about health on each other, following \cite{vansteelandt2017interventional}:

\begin{align*}
    \nonumber\psi_{IDE, 1} &= \mathbb{P}_n\left(\sum_{m_1}\sum_{m_2}(\mathbb{E}[Y^{1m_1m_2} \mid X, Z = 1] - \mathbb{E}[Y^{0m_1m_2} \mid X, Z = 1]) \right.\\
    \nonumber &\left.\times P(M_1^1 = m_1, M_2^1 = m_2 \mid X, Z = 1)\right) \\
    \nonumber\psi_{M_1, 0} &= \mathbb{P}_n\left(\sum_{m_1}\sum_{m_2}\mathbb{E}[Y^{0m_1m_2} \mid X, Z = 1](P(M_1^1 = m_1 \mid X, Z = 1) \right.\\
    \nonumber&\left.- P(M_1^0 = m_1 \mid X, Z = 1)]P(M_2^0 = m_2 \mid X, Z = 1)\right) \\
    \nonumber\psi_{M_2, 0} &= \mathbb{P}_n\left(\sum_{m_1}\sum_{m_2}\mathbb{E}[Y^{0m_1m_2} \mid X, Z = 1](P(M_2^1 = m_1 \mid X, Z = 1) \right.\\
    \nonumber&\left.- P(M_2^0 = m_2 \mid X, Z = 1))P(M_1^1 = m_2 \mid X, Z = 1)\right) \\
    \psi_{Cov, 0} &= \psi - \psi_{IDE, 1} - \psi_{M_1, 0} - \psi_{M_2, 0}
\end{align*}

As explained by \cite{vansteelandt2017interventional}, the interventional direct effect $(\psi_{IDE, 1})$ represents the contrast in mean potential outcomes when everyone is vaccinated versus no one is vaccinated, while drawing the mediators for each individual randomly from the joint distribution of the mediators conditional on $A = 0$ given their covariates $X$. By contrast the interventional indirect effect of $M_1$ $(\psi_{M_1, 0})$ represents the average contrast between subject-specific draws from the counterfactual distribution of feelings of social isolation at $A = 0$ versus $A = 1$ given their covariates, while drawing worries about health from the counterfactual distribution under $A = 0$, given their covariates and fixing vaccination status at level $A = 1$. The interventional indirect effect through $M_2$ $\left(\psi_{M_2, 0}\right)$ is defined analogously. The covariant effect $(\psi_{Cov, 0})$ is the difference between the total effect and all three of these effects. This effect captures the effect of the dependence of the mediators on each other, and would be equal to zero, for example, if these mediators were conditionally independent. Finally, a similar decomposition holds switching 0 for 1 in each estimand and reversing the signs. We consider both decompositions in our analysis. 

\cite{vansteelandt2017interventional} describes $\psi_{M_1, 0}$ as capturing the pathways $A \to M_1 \to Y$ and $A \to M_2 \to M_1 \to Y$, with $\psi_{M_2, 0}$ defined analogously. This implies that $\psi_{M_1, 0}$ captures the effect of vaccinations on depression via directly changing social isolation which then affects depression; and via directly changing worries about health, which then affects social isolation which affects depression. By contrast, $\psi_{M_2, 0}$ would capture the effect of vaccinations on depression via direct changes in worries about health which affects depression, and via changes in social isolation which affect worries about health which affect depression. We refer to \cite{vansteelandt2017interventional} for more details on this decomposition and the interpretation of these estimands.\footnote{We caution that recent results from \cite{miles2022causal} casts doubt on mechanistic interpretations of interventional effects.}

\subsubsection{Other estimands}

In Appendix~\ref{app:estimands} we describe three other estimands: the effect of the observed distribution of vaccination, the effect of incremental interventions on the propensity score, and the effect of the pandemic on mental health. The first two estimands may correspond to more realistic interventions, requiring weaker identifying assumptions than those noted below. The third estimand may be of interest if we wish to use our effect estimates to learn about the impact of the COVID-19 pandemic on mental health. We show in Appendix~\ref{app:identification} that under some assumptions, we may interpret our average effect estimates on either the risk-differences or risk-ratio scale as a lower bound on this quantity using our existing estimates of the effect of COVID-19 vaccinations on mental health.

\subsection{Identification}\label{ssec:identification}

We assume the following causal structure holds among vaccine-accepting individuals.

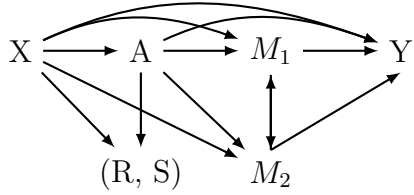
\begin{figure}[H]
\begin{center}
\begin{tikzpicture}
[
array/.style={rectangle split, 
	rectangle split parts = 3, 
	rectangle split horizontal, 
    minimum height = 2em
    }
]
 \node (3) {X};
 \node [right =of 3] (4) {A};
 \node [below =of 4] (8) {(R, S)};
 \node [right =of 4] (5) {$M_1$};
 \node [below =of 5] (6) {$M_2$};
 \node [right =of 5] (7) {Y};
 \draw[Arrow] (3.east) -- (4.west);
 \draw[Arrow] (4.east) -- (5.west);
 \draw[Arrow] (6.north) -- (7.south);
 \draw[Arrow] (3) to [out = 25, in = 160] (7);
 \draw[Arrow] (3) to [out = 25, in = 160] (5);
 \draw[Arrow] (4) to (8);
 \draw[Arrow] (3) to (8);
 \draw[Arrow] (5) to (6);
 \draw[Arrow] (6) to (5);
 \draw[Arrow] (5) to (7);
 \draw[Arrow] (3) to (6);
 \draw[Arrow]  (4) to [out = 25, in = 160] (7);
 \draw[Arrow] (4) to (6);
\end{tikzpicture}
\caption{Assumed causal structure}
\label{fig:dag}
\end{center}
\end{figure}

Figure~\ref{fig:dag} clarifies assumptions about causal ordering of the variables. Two particular features are worth noting: first, we assume no reverse causation between depression and the mediators. Second, we allow the mediators to depend on each other and require no assumptions about their causal ordering. Figure~\ref{fig:dag} also encodes the assumed relationships between the variables necessary to identify our causal estimands in terms of our observed data. We instead motivate these conditions using potential outcomes, and assume that for all values of $(a, m_1, m_2)$:\footnote{The conditions in the DAG are necessary, but not sufficient for our identification result (see Appendix~\ref{app:identification}). For example, the DAG implies that $(R, S) \perp (Y, M_1, M_2) \mid (X, A, V = 1)$, while we instead invoke Assumption~\ref{assumption:rss}. Identification under alternative conditions is possible.} 

\begin{assumption}[No unmeasured confounding]
\begin{align}
\label{eqn:unconfoundedness-ya}&(Y^{am_1m_2}, M_1^a, M_2^a) \perp A \mid (X, V = 1)
\end{align}
\end{assumption}

\begin{assumption}[Y-M ignorability]
\begin{align}
\label{eqn:unconfoundedness-ym}&Y^{am_1m_2} \perp (M_1, M_2) \mid (A = a, X, V = 1) 
\end{align}
\end{assumption}

Specifically, no unmeasured confounding states that among the vaccine-accepting population, vaccinations are independent of the potential mediators and outcomes conditional on covariates. Y-M ignorability states that the mediators are jointly independent of potential outcomes conditional on covariates and vaccination status. Finally, we we assume random sample selection: this states that survey response and complete-case indicators are random with respect to the potential outcomes given the covariates and vaccination status.

\begin{assumption}[Random sample selection]\label{assumption:rss}
\begin{align}
\label{eqn:amarx}&(S, R) \perp (Y^{am_1m_2}, M_1^a, M_2^a) \mid (X, A, V = 1) 
\end{align}
\end{assumption}

We additionally assume that the survey responses fully capture all population relationships (``no measurement error''), that all vaccine-accepting individuals in our sample have a probability of vaccination bounded away from zero and one, and that the joint mediator probabilities are all bounded away from zero (``positivity''). We formalize all identifying assumptions in Appendix~\ref{app:identification} and show that they imply that our causal parameters are identified in terms of the observed data. For example, $\mathbb{P}_n[\mu_a(X_i)] = \mathbb{P}_n[\mathbb{E}[Y_i \mid X_i, A_i = a, Z_i = 1]]$. We largely defer discussing possible violations of these assumptions to Section~\ref{sec:limitations}. 

\subsection{Analytic timeframe}\label{sec:timeframe}

We choose to limit our causal analysis to February 2021 to strengthen the credibility of our identifying assumptions. The challenge is that we only observe cross-sectional data, and the causal assumptions required for our analysis simplify a more complicated longitudinal process. Our identifying assumptions arguably become less tenable over time. As one example, consider a key unobserved variable: time since vaccination. Assuming no carryover effects implies that this variable is not required either to define or identify our causal estimands. Absent this assumption, time since vaccination would act as an unmeasured confounder inducing a bias for our outcomes analysis that plausibly increases with time from the first administered COVID-19 vaccines (we discuss this point further in Appendix~\ref{app:identification}). Our mediation analysis similarly becomes less tenable over time: for example, vaccinations may improve economic conditions over time, which in turn may act as a post-treatment confounder with respect to feelings of social isolation or worries about health. We therefore limit our analysis to February, relatively soon after the start of COVID-19 vaccination administrations. 

\subsection{Estimation and inference}

We use influence-function based estimators throughout (\cite{benkeser2021nonparametric}, \cite{kennedy2019nonparametric}), requiring estimaion of several nuisance functions. For example, for the interventional effects, we estimate the outcome model, the propensity score model, and the joint mediator probabilities. For our primary results we use logistic regression and multinomial logistic regression. We allow for separate models for each treatment level for our outcomes analysis, and each treatment-mediator level for our mediation analysis. Similarly, to estimate mediator probabilities we allow separate models within each treatment level. Our propensity score model includes only main effects for each covariate. These estimators are multiply robust in the sense that the consistency of the estimates do not require that all models are correctly specified, but rather that some subset are (\cite{benkeser2021nonparametric}). We generate 95\% confidence intervals and conduct hypothesis tests using the empirical variance of our estimated influence function values and standard normal quantiles. These intervals and tests are valid if all of our models are correctly specified. To alleviate concerns about model misspecification, we also estimate all nuisance functions using XGBoost and sample-splitting. Assuming the consistency of the influence-function estimates and that that the nuisance functions are estimated at $n^{-1/4}$ rates in the $L_2$ norm, these estimates are  asymptotically normal and root-n consistent with a simple expression for the variance (\cite{benkeser2021nonparametric}). Appendix~\ref{app:estimation} provides further details.

We also use a data-driven approach to identify subgroups with large or small treatment effects. Specifically, we estimate the influence function values using auxiliary data from March 1-14, 2021 using XGBoost. We regress these pseudo-outcomes on a subset of the demographic covariates using depth-four regression trees. After generating the trees, we use a ``human-in-the-loop'' approach to determine whether the candidate subgroups are interpretable and substantively interesting. We then conduct estimation and inference using on these candidate hypotheses using the February data.\footnote{This is similar to the ``honest'' tree-based estimation approach proposed by \cite{athey2016recursive}, and is an example of a DR-learner, discussed in \cite{kennedy2020towards}} All t-tests referenced below set $\alpha = 0.05$ and only control the Type I error rate for each test marginally.

\section{Results}\label{sec:results}

\subsection{Outcome analysis}\label{ssec:outcomes}

Table~\ref{tab:cc-glm} displays our primary results. We estimate comparable effect sizes for each variable decreases of 3.7 (4.0, 3.5), 3.3 (3.5, 3.0), and 4.3 (4.6, 4.0) in the prevalence of depression, isolation, and worries about health, respectively. These effects translate to 19 percent, 16 percent, and 15 percent decreases in the prevalence of depression, social isolation and worries about health, respectively.\footnote{As noted in Appendix~\ref{app:identification}, we may interpret these results as bounding the effect of the pandemic on depression. For example, under some assumptions we can say that the pandemic increased depression by at least 3.7 percentage points on average, and that the percentage of depressed respondents was at least 24 percent higher on average (100*($\frac{1}{1 - .19} - 1$)) than it would have been absent the pandemic among respondents in our analytic sample.} Table~\ref{tab:cc-xgb} displays our results when using XGBoost to estimate the nuisance functions. These estimates are quite similar to the GLM estimates. 

  \begingroup%
  \renewcommand\normalsize{\footnotesize}% Specify your font modification
  \input{tables/risk-difference-table-cc}%
  \endgroup%

  \begingroup%
  \renewcommand\normalsize{\footnotesize}% Specify your font modification
  \input{tables/risk-difference-table-xgb-cc}%
  \endgroup%

\subsection{Subgroup heterogeneity: pre-specified groups}

Figure~\ref{fig:alloutcomes} displays the estimated effects across pre-specified subgroups. Consistent with the unadjusted estimates, among those who reported their age, vaccinations have the largest magnitude effect sizes among respondents between 18-34 years old for all outcomes. For example, we estimate that vaccinations reduced depression by -5.9 (-6.8, -5.0) percentage points among respondents aged 18-34, while only -2.5 (-2.8, -2.2) percentage points among respondents aged 65 and older. We also find larger magnitude point estimates for females relative to males and non-binary respondents for depression. Specifically, we estimate the effect of vaccinations on depression as being -4.0 (-4.3, -3.7) percentage points among female identifying respondents and -3.1 (-3.5, -2.7) among male and non-binary respondents. Lastly, we see larger effect estimates among respondents for worries about health among respondents with two or more high-risk health conditions (-4.8 (-5.3, -4.3)) relative to those with none (-3.7 (-4.2, -3.2)).\footnote{Our estimates are somewhat large in absolute magnitude among respondents who do not provide demographic information. For example, among those who did not respond about their age, we estimate a -5.8 (-7.2, -4.3) percentage point reduction in depression and a -7.3 (-8.9, -5.6) reduction in worries about health. Because non-response is highly correlated across questions, we observe similar magnitude estimates for other non-response categories, though the point estimates for isolation tend to be comparable to other groups.}

\begin{figure}[H]
\begin{center}
    \includegraphics[scale=0.5]{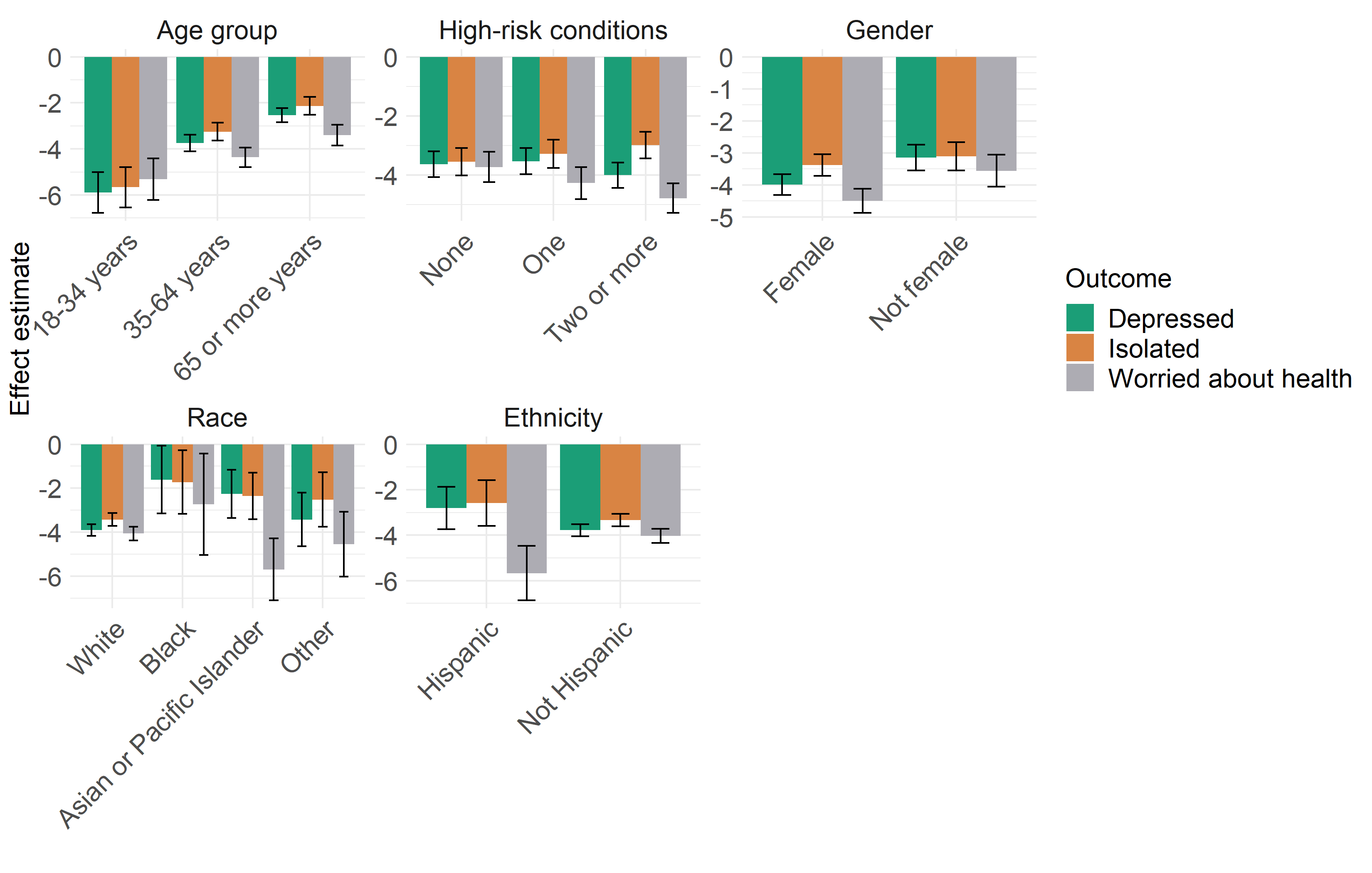}
    \caption{Outcomes analysis, February 2021}
    \label{fig:alloutcomes}
\end{center}
\end{figure}

The XGBoost estimates are very similar. These results and results on the risk-ratio scale are available in Appendix~\ref{app:results-outcomes}. 

\subsubsection{Data-driven subgroup heterogeneity}

We identify three possibly heterogeneous subgroups using the March data with respect to depression: respondents aged 18-24, non-Hispanic respondents aged 25 and older, and Hispanic respondents aged 25 and older. These categories exclude non-respondents who to age, or to ethnicity among those 25 or older. Figure~\ref{fig:hteoutcomesl1} displays the effect sizes and confidence intervals estimated on the February data.\footnote{We present estimates using our XGBoost models: our GLM estimates do not include interactions between the covariates, and are therefore unlikely to correctly capture this heterogeneity.} Consistent with our hypotheses, we estimate effects of -12.2 (-13.7, -10.7) among respondents aged 18-24, -3.6 (-3.8, -3.3) among non-Hispanic respondents  aged 25 or older, and -1.9 (-2.6, -1.2) among Hispanic respondents aged 25 or older. The youngest respondents also appear to have larger magnitude treatment effects for isolation and worries about health.\footnote{We display significance levels associated with the pairwise t-tests in Appendix~\ref{app:results-outcomes}.}

\begin{figure}[H]
\begin{center}
    \includegraphics[scale=0.5]{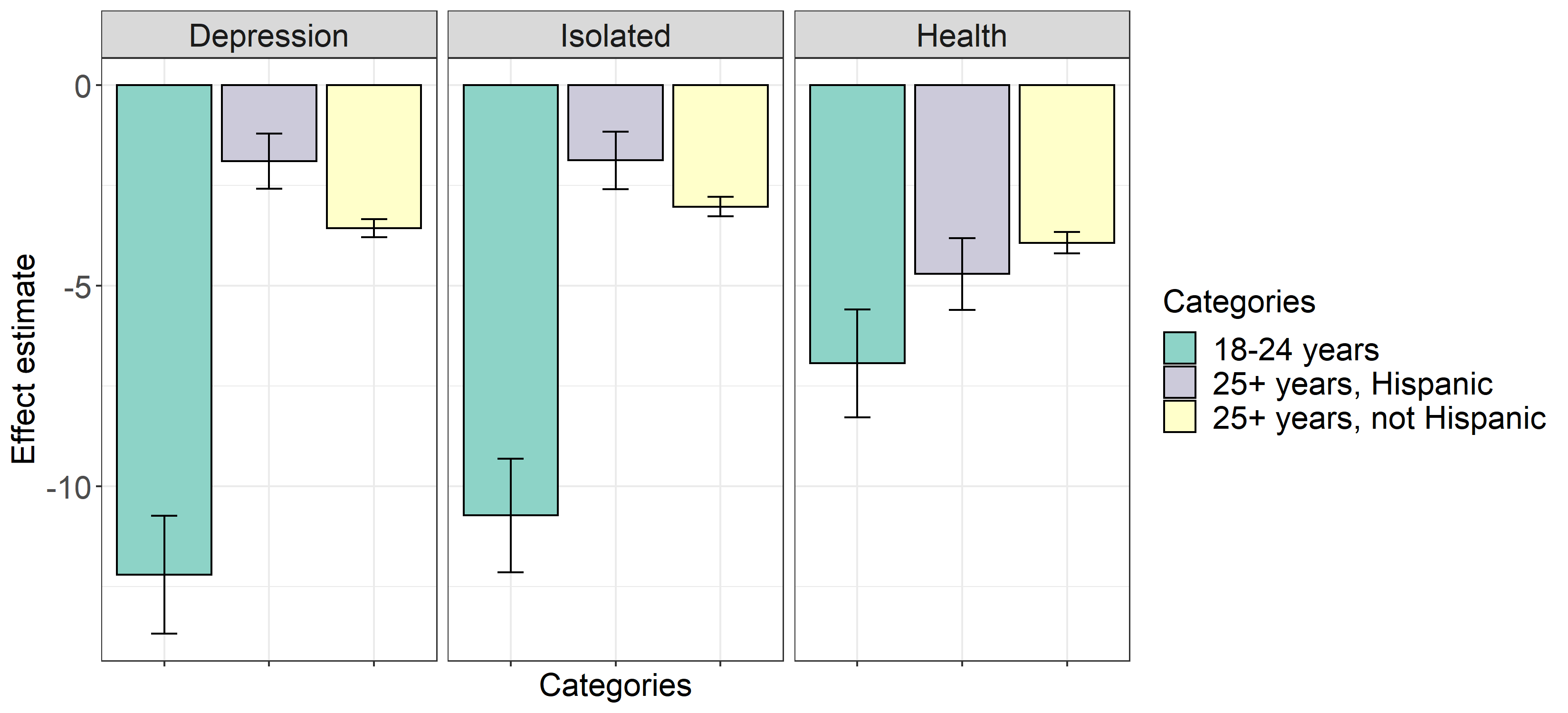}
    \caption{Data-driven subgroup estimates}
    \label{fig:hteoutcomesl1}
\end{center}
\end{figure}

We examine additional subgroups within these three primary groups. Among respondents aged 18-24, we estimate larger absolute magnitude effect sizes among those who live with an elderly individual (-21.6 (-26.5, -16.8); N = 2,542) versus those who do not (-10.9 (-12.7, -9.2); N = 20,347). We estimate effects for depression and isolation that are statistically indistinguishable from zero among Hispanic respondents whose primary Facebook user language is Spanish (N = 23,018). By contrast, all estimated effect sizes are statistically significant and negative for those whose Facebook user language is English (N = 57,236).\footnote{The differences in the estimated effects are both statistically significant at the $\alpha = 0.05$ level.} We test for differences with respect to gender and race among respondents aged 18-24, and with respect to age among Hispanic respondents aged 25 or older. We find statistically significant differences with respect to race among the youngest respondents for isolation and worries about health: the effect sizes appear smaller in absolute magnitude for non-White versus White respondents. However, our remaining hypotheses do not yield any statistically significant and substantively meaningful differences on the risk-differences scale.

Finally, we identify candidate heterogeneous subgroups using isolation and worries about health as the outcomes. Beyond the previously noted heterogeneity with respect to age, we are largely unable to confirm additional sources of heterogeneity. One exception is with respect to gender among non-White respondents aged 25 or older: we estimate that female identifying respondents had larger magnitude effects with respect to social isolation (-3.2 (-4.0, -2.5)) than non-female identifying respondents (-1.5 (-2.4, -0.7)). Additional results may be found in Appendix~\ref{app:results-outcomes}. 

Our results reveal that being aged 18-24 is a key source of effect heterogeneity, especially with respect to depression and social isolation. Moreover, living with an elderly individual appears to have a strong interaction effect among this same group. These patterns also hold on the risk-ratio scale (see Appendix~\ref{app:results-outcomes}). Additional research would be valuable to further investigate hypotheses with respect to other demographic characteristics, including interactions between age, race, gender, and ethnicity, though our results suggest that age would likely remain the factor most strongly associated with effect heterogeneity.

\subsection{Mediation analysis}\label{ssec:mediation}

Figure~\ref{fig:mediationmain} displays results from our mediation analysis. The left-hand panels displays results on the risk-differences scale while the right hand displays each estimate as a proportion of the total effect. We estimate that shifting the distribution of social isolation from what it would have been if everyone were unvaccinated to a world where everyone were vaccinated would reduce depression by 1.5 (1.6, 1.4) percentage points, while shifting the distribution of worries about health similarly reduces depression by approximately 0.3 (0.3, 0.3) percentage points. These effects would account for 39.1 percent and 8.3 percent of the total effect of vaccinations on depression, respectively. We estimate the covariant effect to be close to zero. Figure~\ref{fig:mediationmain} also displays the total indirect effect by both pathways, simply the sum of these three effects. We reject the null hypothesis that the mediated effects through isolation are equal to the mediated effects through health ($p < 0.001$). The effects are comparable when changing the reference category in the estimands (see Appendix~\ref{app:results-mediation}).

\begin{figure}[H]
\begin{center}
    \includegraphics[scale=0.4]{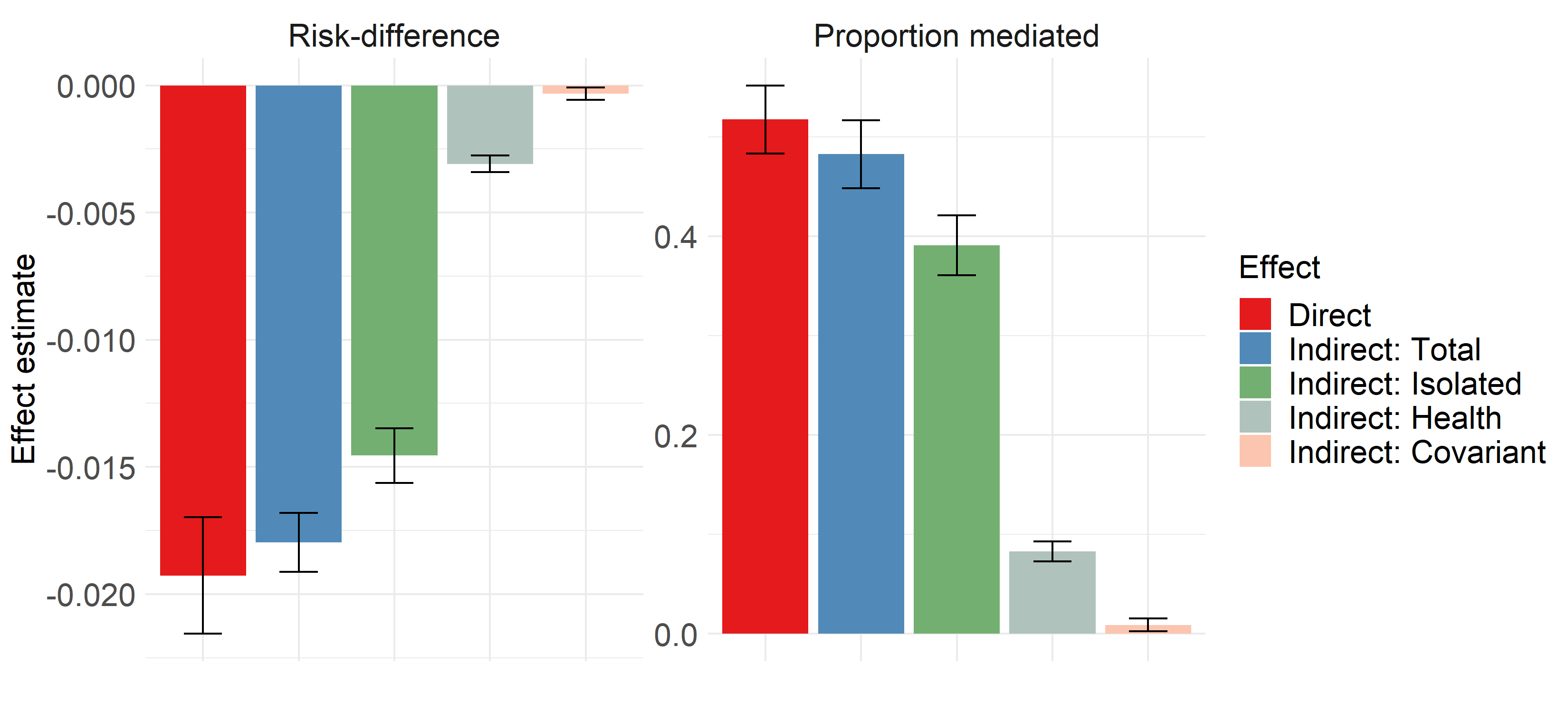}
    \caption{Mediation analysis: primary estimates}
    \label{fig:mediationmain}
\end{center}
\end{figure}

We also find evidence of a direct effect of vaccination on depression that accounts for over half of the total effect. We interpret this as reflecting a general sense of relief that vaccination provides, perhaps related to a belief in a return to normalcy. 

Figure~\ref{fig:htemedpropspaper} displays estimates of the proportion mediated across pre-specified subgroups.\footnote{We excluded Asian and Pacific Islander and Other racial categories as the estimates are very imprecise.} The patterns are consistent, with social isolation appearing to account for a larger proportion of the mediated effect relative to worries about health across all groups. Perhaps not surprisingly, worries about health appear to account for a slightly greater percentage of the total effect among people with two or more high-risk health conditions compared to those with none (though we do not formally test the statistical significance of these differences). The XGBoost results are similar. Additional results are available in Appendix~\ref{app:results-mediation}.

\begin{figure}[H]
\begin{center}
    \includegraphics[scale=0.5]{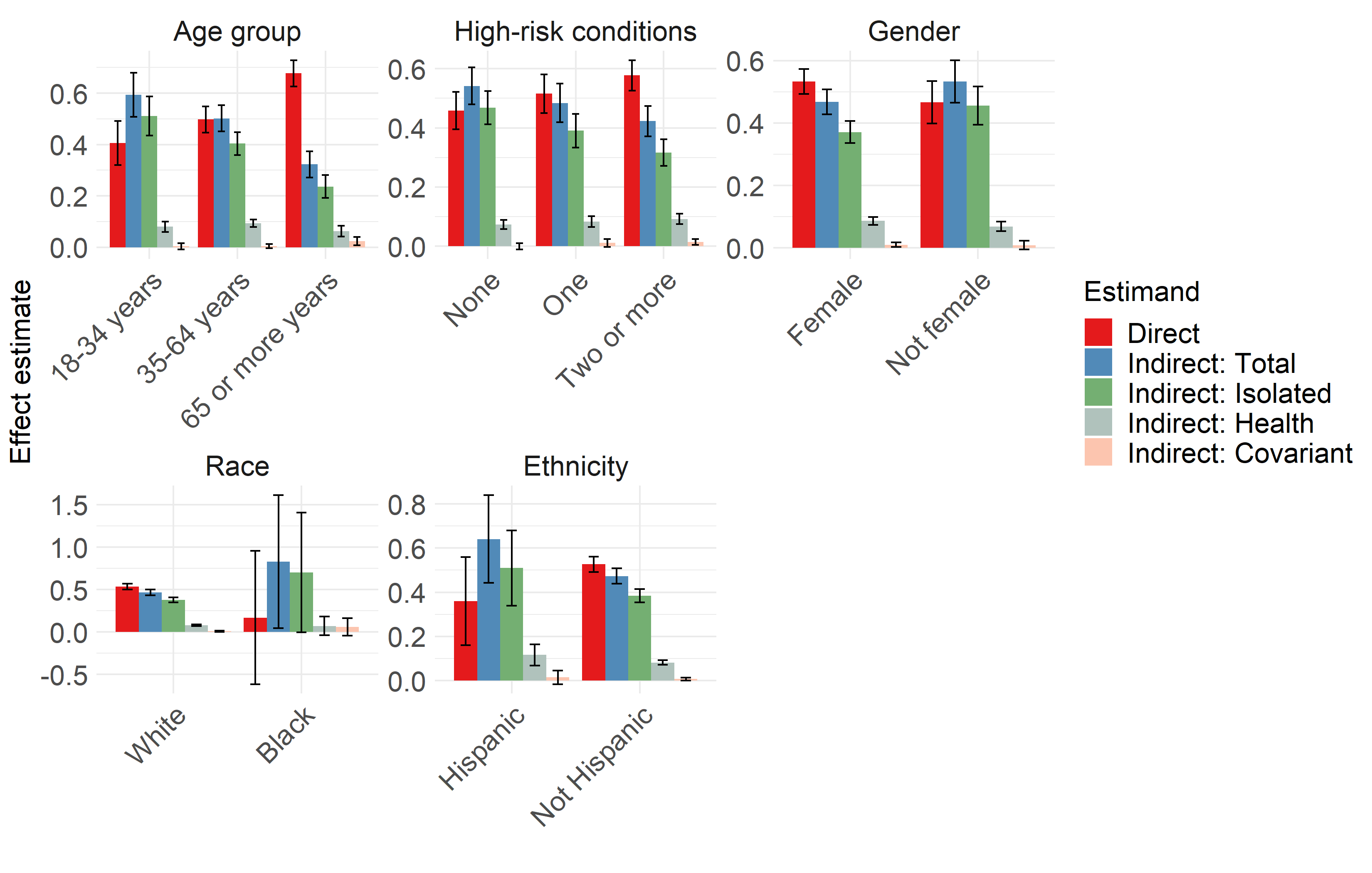}
    \caption{Mediation heterogeneity: proportion mediated by subgroup}
    \label{fig:htemedpropspaper}
\end{center}
\end{figure}

\section{Sensitivity, robustness, and limitations}\label{sec:limitations}

We consider the sensitivity of our analyses to our causal assumptions, their robustness to alternative analytic choices, and our ability to generalize our effect estimates. We highlight the limitations of our analysis throughout, and also provide heuristic arguments about how violations of our causal assumptions may bias our estimates. We conclude that our outcomes analysis is fairly robust, and that most possible violations likely bias our estimates towards rather away from zero. We also argue that generalizing the effect magnitudes of our estimates beyond our sample is not unreasonable. Finally, we discuss the limitations of our mediation analysis; however, we do not conduct formal sensitivity analyses. 

\subsection{Unobserved confounding and selection bias}\label{ssec:nuc}

\subsubsection{Sensitivity analysis}\label{ssec:sensitivity}

Our outcomes analysis assumes no unmeasured confounding and random sample selection. Together these assumptions imply that the potential outcomes are independent of treatment assignment given covariates and vaccine-acceptance on our sample. This assumption may be violated either assumption fails. Unmeasured confounders may threaten this assumption at the population level, or selection into our analytic sample may induce bias. We can relax these assumptions and instead estimate obtain bounds on our target sample estimand via a variation of a sensitivity analysis proposed by \cite{luedtke2015statistics}.  Let $\bar{\mu}_a(a', x) = \mathbb{E}[Y^a \mid X = x, A = a', Z = 1]$. We then assume there exists a $\tau$ satisfying \eqref{eqn:sensitivity1}.

\begin{align}\label{eqn:sensitivity1}
    \frac{1}{1-\tau} \ge \frac{\bar{\mu}_a(a', x)}{\bar{\mu}_a(a, x)}  \ge 1-\tau \qquad \forall (a, a', x)
\end{align}

For example, consider the left-hand side of \eqref{eqn:sensitivity1} and let $a = 1$ and $a' = 0$. The bound specifies that within each covariate stratum, the counterfactual probability of depression among unvaccinated survey respondents when vaccinated is not greater than than $1 / (1 - \tau)$ times the observed probability of depression among the vaccinated survey respondents. In Appendix~\ref{app:sensitivity}, we show that \eqref{eqn:sensitivity1} implies a bound on the treatment effect estimates using only functions of the observed data and $\tau$. We then use influence-function based estimators to estimate bounds across a range of values of $\tau$ until we find values that could explain away our estimated effects. 

We estimate that values of $\tau$ greater than or equal to 0.18, 0.15, and 0.14 would explain away our effect estimates for depression, isolation, and worries about health, respectively. In other words, our effect estimates would be statistically indistinguishable from zero if the ratios in \eqref{eqn:sensitivity1} fell outside of the (approximate) ranges of [0.83, 1.20], [0.86, 1.16], or [0.87, 1.15] for each respective outcome.\footnote{The results using the XGBoost models are virtually identical.} As a point of comparison, assume that no unmeasured confounding on the complete-cases held when controlling for the full covariate set, but that we controlled for none of them. This would yield values of $\tau$ of approximately 0.34, 0.27, and 0.26 for each respective outcome.\footnote{These values reflect cases where the observed covariates are assumed to deconfound the relationship on the observed sample. This does not account for scenarios where bias is induced by selection on the observed or potential outcomes.} Our average effect estimates are therefore robust to confounding variables that are approximately 53 to 55\% (e.g. 0.18/0.34 for depression) as associative with the outcomes as the entire observed covariate set, suggesting that our results are at least somewhat robust to these assumptions.

This analysis is quite conservative: for example, the bounds consider cases where the biases work in opposite directions for $\mu_1$ and $\mu_0$. However, the biases may work in the same direction. For example, if $\bar{\mu}_1(1, x) - \bar{\mu}_1(0, x) \approx \bar{\mu}_0(1, x) - \bar{\mu}_1(0, x)$ for all values of $x$, then then our estimates of the causal risk-differences would remain unbiased, despite biased estimates of the prevalence of each potential outcome in our sample.

\subsubsection{Heuristic reasoning: no unmeasured confounding}

We can also reason about the sign of the bias induced by violations of no unmeasured confounding at the population level. Two specific confounders we may worry about are vaccine-access and motivation to be vaccinated. We attempt to control for vaccine-access using a variety of county-level characteristics, the respondents' geographical location, employment status, and financial stress. However, these covariates may not sufficiently capture differential access across respondents. Similarly, vaccine-acceptance acts as a rough proxy for motivation. Yet even among vaccine-accepting respondents, some likely have higher motivation to be vaccinated than others. We reason that failing to control for these variables likely induces biases that work in opposite directions, resulting in an overall bias with no clear sign. 

For example, if within every covariate stratum, the probability of being vaccinated is lower for respondents with worse access, and the probability of being depressed, isolated, or worried about health is higher for these same respondents, our estimates will be negatively biased. On the other hand, if motivated respondents are more likely to be vaccinated, and more likely to be depressed, isolated, or worried about their health generally, our estimates will be positively biased. Our data is consistent with these hypotheses: when running our analyses without controlling for financial stress, employment status, or occupation, our point estimates move further away from zero (see also Section~\ref{ssec:badcontrols}). On the other hand, when including the vaccine-hesitant in our analytic sample, our estimates move closer to zero (see Table~\ref{tab:senstab1} in Appendix~\ref{app:robustness}). The total bias induced by these unmeasured confounders therefore has no clear sign.

\subsection{Generalizability}\label{ssec:generalizability}

Assuming that our estimates among this sample of CTIS respondents are valid, or more weakly that the magnitude of the effect estimates is correct, we may wish to infer similar magnitude effects among all U.S. adults. Skepticism is warranted for several reasons. First, CTIS respondents differ from the general US population across several characteristics. For example, \cite{daly2021public} found that the CTIS survey over represents White individuals, those with higher education, and those who have received COVID-19 vaccinations (see also \cite{bradley2021unrepresentative}). Second, our estimates may be subject to selection-bias: even if we knew the true sample estimand, if sample selection were based on the potential or realized outcomes, these effects may not generalize to a population of non-CTIS respondents with identical covariates. Finally, our estimates make no claims about the vaccine-hesitant, which is a substantial portion of the U.S. population (\cite{king2021covid}). 

While caution is certainly warranted, we nevertheless argue that these concerns should not entirely prohibit generalizations from our study. First, assume that our conditional effect estimates are valid. Despite our sample's non-representative covariate distribution, we do not find evidence of substantial sources of effect heterogeneity beyond age. Averages over other mixes of vaccine-accepting populations are therefore plausibly similar. Second, sample selection may in truth be a function of the observed or potential outcomes. However, this does not imply our estimates are invalid: our sample may not be representative of the prevalence of the potential outcomes, but may still be representative of the risk-differences. Third, if we wished to generalize across a population that includes a vaccine-hesitant subset, we can simply assume that there is no treatment effect on average in this subgroup. As long as the vaccine-hesitant comprise a minority of the target population, similar magnitude effect sizes may therefore hold.\footnote{As a more formal analysis, in Appendix~\ref{app:sensitivity} we extend the previous sensitivity analysis to evaluate how our effects might generalize across a population with the same covariate distribution as our sample but include entirely cases where $RS = 0$. This analysis suggests that our ability to generalize our estimates may be weak; however, it is also very conservative.}

Such generalization may still strike some as implausible. We may instead wish to consider the more limited goal of generalizing our estimates beyond the complete-cases to the entire vaccine-accepting CTIS sample. We also run this analysis and find that our estimated effect sizes are similar to our primary analysis (see Appendix \ref{app:robustness}).

\subsection{Interference, carryover, and anticipatory effects}

Our analysis makes three restrictions on the dependence of the potential outcomes across individuals and over time: SUTVA, no carryover effects, and no anticipatory effects. While violations of these assumptions are possible, we believe that such violations likely bias our effect estimates towards zero.

First consider violations of SUTVA. If the effect of receiving a COVID-19 vaccination depends on the vaccination status of other individuals in the population, then we must redefine our causal estimand. First, assume that these dependencies are limited to communities of individuals, which we index from $c = 1,..., C$. For simplicity, assume that each community has a fixed number of individuals $p$, and that $p, C > 1$. We can then define the following terms: let $\mathbf{A}_{c-i}$ represent a random vector of vaccination statuses for all individuals in a community $c$, omitting the index $ci$, and let $S(\mathbf{A}_{c-i})$ be the average number of vaccinated individuals in the random vector $\mathbf{A}_{c-i}$.  Finally, let $\bar{1}_{c-i}$ be a vector of $p - 1$ ones and $\bar{0}_{c-i}$ be a vector of $p - 1$ zeros (omitting the $ci$-th index), so that, for example, $S(\bar{1}_{c-i}) = 1$.

Consider the following estimand: 

\begin{align}
\vartheta = \mathbb{P}_n[Y^{1, \bar{1}} - Y^{0, \bar{0}}]
\end{align}
This formalizes the difference in rates of depression in a world where everyone is vaccinated versus a world where no one is vaccinated, allowing each individual's potential outcomes to depend on everyone else's vaccination status within their community. We can then consider how well have we estimated $\vartheta$. Consider the model:

\begin{align}\label{eqn:interference1}
P(Y^{A, \mathbf{A}} = 1 \mid x) = \beta_{0x} + \beta_{1x}A_{ci} + \beta_{2x}S(\mathbf{A}_{c-i})
\end{align}

Equation~(\ref{eqn:interference1}) implies that $\vartheta = \mathbb{P}_n(\beta_{1X} + \beta_{2X})$. This estimand captures both a direct effect of vaccinations, captured by $\mathbb{P}_n(\beta_{1X})$, and an indirect effect captured by $\mathbb{P}_n(\beta_{2X})$. Assume that for all $x$, $\beta_{1x}, \beta_{2x} < 0$ -- that is, that both the direct and indirect effects of vaccination are negative within all covariate strata. To make these assumptions concrete, consider depression as the outcome. These assumption would be reasonable if (1) vaccinations induce people to socialize, and socializing decreases depression, and (2) socializing increases with the number of vaccinated individuals within all covariate strata.

However, our estimates are not based on \eqref{eqn:interference1}, but instead on the model:

\begin{align}\label{eqn:interference2}
P(Y^{A, \mathbf{A}} = 1 \mid x) = \tilde{\beta}_{0x} + \tilde{\beta}_{1x}A_{ci}
\end{align}

Our estimators therefore target the biased quantity $\tilde{\vartheta} = \mathbb{P}_n[\tilde{\beta}_{1X}]$. The omitted variable bias formula implies that this estimand is positively biased with respect to $\vartheta$.\footnote{To see this, consider the model:

\begin{align}
\mathbb{E}[S(\mathbf{A}_{c-i}) \mid x] = \delta_{0x} + \delta_{1x}A_{ci}
\end{align}

Within each covariate stratum the bias of our target quantity relative to our target is equal to $\beta_{2x}(\delta_{1x} - 1)$. The sample average of this bias across all covariate strata is greater than zero: $\beta_{2x} < 0$ for all $x$ by assumption, and $\delta_{1x} \le 1$ by the definition of $S(\mathbf{A}_{c-i})$. As long as vaccinations aren't perfectly correlated within every community, there exists some $x'$ in our sample where $\delta_{1x'} < 1$, implying positive bias.

If vaccinations were independently assigned, $\delta_{1x} = 0$, $\tilde{\beta}_{1x} = \beta_{1x}$, and $\tilde{\vartheta}$ would correctly capture the direct effects of vaccination on the outcomes (but not the total direct and indirect effects). The sign of the bias of the direct effect is unclear and depends on $\delta_{1x}$.} Whether $\vartheta$ is the correct target of inference, or whether these assumptions are realistic is debatable. Nevertheless, this logic provide justification for believing that SUTVA violations would bias our effect estimates toward, rather than away, from zero.\footnote{\cite{agrawal2021impact} attempt to  estimate spillover effects by examining whether increases in community vaccination rates improve depression and anxiety symptoms among the unvaccinated population. They are unable to find evidence that this is the case, providing some empirical evidence that we may need not worry about SUTVA violations affecting our estimates.}

We also assume no carryover effects, precluding the potential outcomes from depending on the time since receipt of vaccination. Were this violated, time since vaccination would essentially be a confounder that is positively associated with the observed treatment assignment indicator. Assuming that the effects of COVID-19 vaccinations on mental health monotonically diminish over time on average within each covariate stratum, violations of no carryover effects would be bias our estimates towards zero. Finally, we assume no effects in anticipation of receiving the vaccine. If anticipatory effects were to improve mental health on average in our sample, violations would again bias our effect estimates towards zero. In short, violations of SUTVA, no carryover effects, and no anticipatory effects most plausibly bias our effect estimates towards zero.

\subsection{Bad controls}\label{ssec:badcontrols}

So-called ``bad controls'' may also bias our effect estimates (\cite{angrist2009mostly}). We argue that these biases again plausibly move our estimates toward zero.

Specifically, we control for occupation, employment status, and worries about finances. These variables are associated with vaccine-access, and -- as noted above -- failing to control for such variable likely induces a negative bias in our effect estimates. However, COVID-19 vaccinations may also affect these variables, meaning that they act as a mediator on the causal pathway from $A \to Y$, and are therefore ``bad controls.'' We argue that such effects were weak in February 2021: acquiring or switching job can take considerable time, and everyone in our sample was vaccinated for at most three months. Additionally, any  bias induced by over-controlling is likely positive. For example, COVID-19 vaccinations likely boosted economic activity in aggregate and increased employment. This would then likely reduce financial stress and the associated mental health burdens. Controlling for variables along this causal pathway therefore likely results in positive bias. 

We examine these hypotheses in our data by rerunning our analyses without controlling for worries about finances, occupation, or employment status. Our point estimates move further away from zero (see Table~\ref{tab:senstab5} in Appendix~\ref{app:robustness}). This is consistent with either over-controlling moving our point estimates closer to zero, or a negative bias induced by failing to control for these variables. Because we prefer our estimates to be biased towards zero we therefore control for these variables. 

\subsection{Measurement error}\label{ssec:msrmenterror}

Measurement error may also bias our estimates. For example, our data is self-reported, and may be subject to reporting bias. Relatedly, we assume no unmeasured confounding assumption conditional on $X$, which includes both observed covariate values and categories for missing data. The unobserved values of these covariates may act as unmeasured confounders. We are unable to fully address these issues. However, as one robustness check we run our analysis while excluding respondents who provide implausible or unreasonable answers. For example, several respondents indicate that they live with a negative number of other individuals. Following \cite{bilinski2021better}, we consider a variety of implausible answers to our questions and omit these responses when running our analysis, removing 15,762 complete-cases. All results are consistent with our primary analyses and are available in Appendix~\ref{app:robustness}.

\subsection{Mediation analysis}\label{ssec:mediation}

The validity of our mediation analysis requires strong assumptions. For example, our model precludes a prominent possible pathway through which COVID-19 vaccinations might affect depression: through household finances (\cite{wilson2020job}, \cite{berkowitz2021unemployment}). Our analysis instead assumes that vaccinations do not affect depression by changing finances. This assumption may not strictly hold: at an individual level, receipt of a COVID-19 vaccination might induce the previously unemployed to seek employment, which might then change a person's personal finances. At an aggregate level, mass distribution of COVID-19 vaccinations may have boosted economic activity, which may affect any individual's employment status and therefore finances. As noted in Section~\ref{sec:timeframe}, this is one reason we limit our analysis to February, a timeframe when such effects are plausibly negligible. 

Our analysis also requires that there exist no other post-treatment effects of vaccination that confound the mediator-outcome relationship. For example, we interpret the direct effect as reflecting the general feelings associated with a return to normalcy. Our analysis requires that this feeling does not cause social isolation or worries about health. However, causal arrows may exist from social isolation or worries about health to this variable. In this case both the direct and indirect pathways capture the effect of this variable.\footnote{\cite{perez2021covid} speculate that in addition to worries about health and social isolation, ``different work opportunities'' is a possible pathway through which COVID-19 vaccines may affect depression. While this pathway is related to household finances, worries about household finances could remain the same while one's employment status may plausibly differ. More broadly our analysis rules also out any pathways via employment changes. Limiting our analysis to February again mitigates potential bias if we think people are unlikely to change their employment within a short-time frame after being vaccinated.} These conditions and interpretations hold for any other hypothesized pathway between COVID-19 vaccines and depression.

Finally, the interventional effects are not robust to reverse causation with respect to $Y$. We must therefore interpret these results cautiously. While we do not conduct a sensitivity analysis with respect any of these assumptions, we compute all of the robustness checks noted previously and find comparable results. 

\section{Discussion}\label{sec:discussion}

We provide evidence that COVID-19 vaccinations reduced feelings of depression and anxiety, social isolation, and worries about health, among vaccine-accepting CTIS respondents in February 2021. Our results imply that not being vaccinated decreased the prevalence of each outcome by 3.7, 3.3, and 4.3 percentage points, respectively, translating to 19, 16, and 15 percent reductions in the prevalence of each outcome. We also examine effects by both pre-specified and data-driven subgroups. We observe substantial heterogeneity with respect to age, with younger age groups having larger magnitude effect sizes. Our estimates on the risk-ratio scale suggest that some of this heterogeneity may be driven by differences in base rates of each outcome by age group. However, on both the risk-differences and risk-ratio scales, we find that the total effects on each outcome are strongest among respondents aged 18-24 relative to older respondents, with particularly large effects on depression and isolation among those living with an elderly individual. We also observe that those with multiple high-risk health conditions appear to have slightly larger point estimates with respect to worries about health. Our average effect estimates are moderately robust to violations to our causal assumptions and analytic choices, and we argue that many violations would tend to bias our effect estimates towards zero on the risk-differences scale. Finally, while our estimates pertain only to our sample of CTIS respondents, we argue that inferring similar magnitude effect sizes beyond our sample may not be unreasonable.

We conclude by positing a model where vaccinations affect depression through social isolation and worries about health. We then decompose the total effect into a direct effect and indirect effects via interventions on the distribution of social isolation and worries about health while holding vaccination status fixed. We find that the interventional indirect effect via social isolation is larger in absolute magnitude than the interventional indirect effect via worries about health. We find evidence of a direct effect that accounts for approximately half of the total effect on depression, which we interpret as reflecting a pathway via a perceived return to normalcy (see also \cite{agrawal2021impact}). These patterns are broadly consistent across our pre-specified subgroups.

Our study has several implications. First, we find especially large effect sizes among respondents aged 18-24. These results complement findings that the pandemic had particularly negative effects on younger individuals (\cite{santomauro2021global}, \cite{sojli2021covid}), and suggest that COVID-19 vaccinations may have undone some of these effects, at least temporarily. Second, our results suggest that effects on social isolation play a larger mediating role than effects on worries about health in explaining the effect of COVID-19 vaccinations on depression. Finally, our results complement and support the findings from other studies documenting and understanding this unintended beneficial outcome of COVID-19 vaccinations with respect to mental health (\cite{perez2021covid}, \cite{agrawal2021impact}). For policy-makers seeking to alleviate the mental health burden of the pandemic, our results suggest targeting interventions that reduce feelings of social isolation, particularly among younger individuals. 

We conclude with three notes of caution. First, our results rely on strong causal assumptions. Second, our primary statistical target of inference is the treatment effect on a subset of vaccine-accepting CTIS-respondents, and not the U.S. population more broadly. Finally, the CTIS measure of depression and anxiety does not assess clinical criteria for mental disorders. Our measure, like those used in other studies, likely captures a broad range of psychological distress and demoralization than would be considered a clinically significant disorder, and our findings may reflect improvements in mental health status in a non-clinical range of severity (\cite{murphy1986diagnosis}; \cite{tecuta2015demoralization}). For instance, symptoms of demoralization, such as helplessness, may be more responsive to the introduction of a vaccine than symptoms of major depression, such as anhedonia (\cite{wellen2010differentiation}). While the results indicate improvements in mental health status, these results should not be interpreted as indicating that COVID-19 vaccinations had an impact on clinical disorders comparable to the impact of psychotherapy or antidepressant medication. Regardless, future work should continue to study mental health during the pandemic to seek to better understand the extent of the challenges people are facing and to design policies that can help alleviate suffering, particularly among society's most vulnerable populations.

%% file: tables/sample-chars-paper.tex
\begin{table}[!h]

\caption{Percent CTIS respondents vaccinated by age group, February 1-28th, 2021 \label{tab:schars}}
\centering
\begin{tabular}[t]{lrr}
\toprule
Age group & Percent vaccinated & N\\
\midrule
18-24 years & 12.1 & 41941\\
25-34 years & 19.0 & 123112\\
35-44 years & 21.1 & 156206\\
45-54 years & 21.6 & 166973\\
55-64 years & 21.5 & 201971\\
65-74 years & 47.4 & 201611\\
75+ years & 62.2 & 87256\\
\bottomrule
\multicolumn{3}{l}{\rule{0pt}{1em}Note: numbers excluded non-respondents to COVID-19 vaccination}\\
\end{tabular}
\end{table}

%% file: tables/risk-difference-table-cc.tex
\begin{table}[!h]

\caption{Effect estimates, GLM models (95\% CI) \label{tab:cc-glm}}
\centering
\begin{tabular}[t]{lllll}
\toprule
Outcome & Adjusted RD & Adjusted RR & Unadjusted RD & Unadjusted RR\\
\midrule
Depressed & -3.72 (-3.98, -3.47) & 0.81 (0.79, 0.82) & -9.82 (-9.98, -9.66) & 0.55 (0.54, 0.55)\\
Isolated & -3.26 (-3.53, -2.99) & 0.84 (0.83, 0.85) & -8.39 (-8.57, -8.22) & 0.62 (0.62, 0.63)\\
Worried about health & -4.26 (-4.55, -3.96) & 0.85 (0.84, 0.86) & -11.35 (-11.55, -11.16) & 0.64 (0.63, 0.64)\\
\bottomrule
\multicolumn{5}{l}{\rule{0pt}{1em}RD indicates risk-differences; RR indicates risk-ratio}\\
\end{tabular}
\end{table}

%% file: tables/risk-difference-table-xgb-cc.tex
\begin{table}[!h]

\caption{Effect estimates, XGBoost models (95\% CI) \label{tab:cc-xgb}}
\centering
\begin{tabular}[t]{lllll}
\toprule
Outcome & Adjusted RD & Adjusted RR & Unadjusted RD & Unadjusted RR\\
\midrule
Depressed & -3.73 (-3.95, -3.52) & 0.81 (0.80, 0.82) & -9.82 (-9.98, -9.66) & 0.55 (0.54, 0.55)\\
Isolated & -3.20 (-3.42, -2.97) & 0.84 (0.83, 0.85) & -8.39 (-8.57, -8.22) & 0.62 (0.62, 0.63)\\
Worried about health & -4.14 (-4.39, -3.89) & 0.85 (0.85, 0.86) & -11.35 (-11.55, -11.16) & 0.64 (0.63, 0.64)\\
\bottomrule
\multicolumn{5}{l}{\rule{0pt}{1em}RD indicates risk-differences; RR indicates risk-ratio}\\
\end{tabular}
\end{table}

%% file: appendices/01-sample-chars.tex
\section{Sample characteristics}\label{app:samplechars}

Table~\ref{tab:flowtab} displays the participant flow into the analytic sample. Tables~\ref{tab:sumtab1}-\ref{tab:sumtab3} display the demographic characteristics of both the analytic sample and the entire February sample. Table~\ref{tab:sumtab1} focuses on the treatment, outcomes, and mediators. Table~\ref{tab:sumtab2} focuses on respondent-level demographic characteristics. Table~\ref{tab:sumtab3} focuses on county-level characteristics, where we divide the continuous variables into quintiles at the county-level and display the proportion of respondents falling into each quintile. Each table separately display the percentage and number of respondents in each subgroup, and the percent vaccinated within each demographic subgroup. We display this information for both the analytic sample and the entire February dataset. Importantly, the percent vaccinated within each subgroup represents the percent among those who responded to both the corresponding demographic question and the COVID-19 vaccination question on the CTIS.

  \begingroup%
  \renewcommand\normalsize{\footnotesize}% Specify your font modification
  \input{tables/flow-table}%
  \endgroup%

\begin{landscape}

  \begingroup%
  \renewcommand\normalsize{\footnotesize}% Specify your font modification
  \input{tables/sample-table-1}%
  \endgroup%

  \begingroup%
  \renewcommand\normalsize{\footnotesize}% Specify your font modification
  \input{tables/sample-table-2}%
  \endgroup%

  \begingroup%
  \renewcommand\normalsize{\footnotesize}% Specify your font modification
  \input{tables/sample-table-3}%
  \endgroup%

\end{landscape}

%% file: tables/flow-table.tex
% latex table generated in R 4.0.4 by xtable 1.8-4 package
% Tue Feb 08 09:36:15 2022
\begin{table}[ht]
\centering
\caption{Participant flow} 
\label{tab:flowtab}
\begin{tabular}{lrr}
  \hline
Sample subset & Total \# & Non-hesitant \# \\ 
  \hline
  Entire Sample & 1232398 & 991629 \\ 
  \hspace{3mm}English or Spanish surveys & 1229868 & 989417 \\ 
  \hspace{3mm}Alaska & 3690 & 2974 \\ 
  \hspace{3mm}Missing FIPS & 39526 & 29685 \\ 
  %Initial Sample & 1186652 & 956758 \\ 
  \hspace{3mm}Missing Intent & 75443 & 75443 \\ 
  \hspace{3mm}Vaccine Hesitant & 229894 &   0 \\ 
  \hline
  Non-hesitant sample & 881315 & - \\ 
  \hspace{3mm}Missing Depressed & 93355 & - \\ 
  \hspace{3mm}Missing Vaccinated & 5355 & - \\ 
  \hspace{3mm}Missing Isolated & 94575 & - \\ 
  \hspace{3mm}Missing Worried & 91156 & - \\ 
  \hline
  Complete-cases & 758594 & - \\ 
   \hline
\end{tabular}
\end{table}

%% file: tables/sample-table-1.tex
\begingroup\fontsize{10}{12}\selectfont

\begin{longtable}[t]{lllll}
\caption{Summary statistics: hesitancy, outcome, mediators, and treatment assignment \label{tab:sumtab1}}\\
\toprule
\multicolumn{1}{c}{ } & \multicolumn{2}{c}{Subset: vaccine-accepting complete-cases} & \multicolumn{2}{c}{Entire February sample} \\
\cmidrule(l{3pt}r{3pt}){2-3} \cmidrule(l{3pt}r{3pt}){4-5}
Group & \% of subset (N) & \% of group vaccinated & \% of subset (N) & \% of group vaccinated\\
\midrule
\addlinespace[0.3em]
\multicolumn{5}{l}{\textbf{Vaccine intent}}\\
\hspace{1em}Already vaccinated & 29.4 (278416) & 100 & 28.7 (319989) & 100\\
\hspace{1em}Yes, definitely & 38.3 (362403) & - & 37 (412120) & -\\
\hspace{1em}Yes, probably & 12.4 (117775) & - & 12.9 (143851) & -\\
\hspace{1em}No, probably not & 9.9 (93596) & - & 10.3 (114328) & -\\
\hspace{1em}No, definitely not & 9.9 (94135) & - & 10.3 (114270) & -\\
\hspace{1em}No response & 0.1 (1053) & - & 0.9 (9878) & -\\
\addlinespace[0.3em]
\multicolumn{5}{l}{\textbf{Vaccinated}}\\
\hspace{1em}Not vaccinated & 70.6 (668962) & - & 71.3 (794447) & -\\
\hspace{1em}Vaccinated & 29.4 (278416) & 100 & 28.7 (319989) & 100\\
\hspace{1em}No response & 0 (0) & - & 0 (0) & -\\
\addlinespace[0.3em]
\multicolumn{5}{l}{\textbf{Depressed}}\\
\hspace{1em}Not depressed or anxious & 81.2 (769435) & 31.9 & 71.3 (794576) & 31.9\\
\hspace{1em}Depressed or anxious & 18.8 (177943) & 18.3 & 16.6 (184752) & 18.4\\
\hspace{1em}No response & 0 (0) & - & 12.1 (135108) & 23.8\\
\addlinespace[0.3em]
\multicolumn{5}{l}{\textbf{Worried about health}}\\
\hspace{1em}Worried & 25.1 (237887) & 23.3 & 22.3 (247971) & 23.5\\
\hspace{1em}Somewhat worried & 41.8 (396144) & 32.5 & 36.6 (407658) & 32.8\\
\hspace{1em}Not too worried & 22.8 (216225) & 33.4 & 20 (222949) & 33.8\\
\hspace{1em}Not worried at all & 10.3 (97122) & 22.5 & 9.1 (101314) & 23\\
\hspace{1em}No response & 0 (0) & - & 12.1 (134544) & 21.8\\
\addlinespace[0.3em]
\multicolumn{5}{l}{\textbf{Isolated}}\\
\hspace{1em}None of the time & 43.8 (414960) & 31 & 38.4 (428360) & 31\\
\hspace{1em}Some of the time & 36.9 (349454) & 31.8 & 32.4 (361024) & 31.8\\
\hspace{1em}Most of the time & 13.2 (125013) & 23.7 & 11.5 (128698) & 23.7\\
\hspace{1em}All of the time & 6.1 (57951) & 15.5 & 5.4 (59796) & 15.6\\
\hspace{1em}No response & 0 (0) & - & 12.3 (136558) & 23.8\\
\bottomrule
\multicolumn{5}{l}{\rule{0pt}{1em}Note: all numbers exclude non-respondents to COVID-19 vaccination}\\
\end{longtable}
\endgroup{}

%% file: tables/sample-table-2.tex
\begingroup\fontsize{10}{12}\selectfont

\begin{longtable}[t]{lllll}
\caption{Summary statistics: individual-level characteristics \label{tab:sumtab2}}\\
\toprule
\multicolumn{1}{c}{ } & \multicolumn{2}{c}{Subset: vaccine-accepting complete-cases} & \multicolumn{2}{c}{Entire February sample} \\
\cmidrule(l{3pt}r{3pt}){2-3} \cmidrule(l{3pt}r{3pt}){4-5}
Group & \% of subset (N) & \% of group vaccinated & \% of subset (N) & \% of group vaccinated\\
\midrule
\addlinespace[0.3em]
\multicolumn{5}{l}{\textbf{Age group}}\\
\hspace{1em}18-24 years & 4.4 (41398) & 12.2 & 3.8 (41941) & 12.1\\
\hspace{1em}25-34 years & 12.7 (120509) & 19.1 & 11 (123112) & 19\\
\hspace{1em}35-44 years & 16 (151208) & 21.2 & 14 (156206) & 21.1\\
\hspace{1em}45-54 years & 16.9 (159671) & 21.8 & 15 (166973) & 21.6\\
\hspace{1em}55-64 years & 20.1 (190546) & 21.5 & 18.1 (201971) & 21.5\\
\hspace{1em}65-74 years & 19.8 (187875) & 47.5 & 18.1 (201611) & 47.4\\
\hspace{1em}75 or more years & 8.2 (78033) & 62.3 & 7.8 (87256) & 62.2\\
\hspace{1em}No response & 1.9 (18138) & 25.7 & 12.1 (135366) & 21.7\\
\addlinespace[0.3em]
\multicolumn{5}{l}{\textbf{Race}}\\
\hspace{1em}Native American & 1.3 (12277) & 31.3 & 1.2 (13322) & 31.6\\
\hspace{1em}Asian or Pacific Islander & 2.3 (22154) & 30.6 & 2.1 (23466) & 31.1\\
\hspace{1em}Black & 6.7 (63327) & 24.5 & 6.2 (69637) & 25\\
\hspace{1em}White & 79.1 (749693) & 31.2 & 70.3 (783131) & 31.5\\
\hspace{1em}Other & 4.9 (46391) & 15.4 & 4.6 (51216) & 15.4\\
\hspace{1em}Multi-racial & 2.5 (23358) & 19.2 & 2.2 (24204) & 19.5\\
\hspace{1em}No response & 3.2 (30178) & 21.5 & 13.4 (149460) & 21.1\\
\addlinespace[0.3em]
\multicolumn{5}{l}{\textbf{Ethnicity}}\\
\hspace{1em}Hispanic & 11.5 (109142) & 19.8 & 10.7 (119759) & 19.8\\
\hspace{1em}Not Hispanic & 85.9 (813430) & 30.8 & 76.4 (851381) & 31.1\\
\hspace{1em}No response & 2.6 (24806) & 26.3 & 12.9 (143296) & 22\\
\addlinespace[0.3em]
\multicolumn{5}{l}{\textbf{Gender}}\\
\hspace{1em}Male & 31.9 (302520) & 28.2 & 28.6 (318320) & 28.4\\
\hspace{1em}Female & 64.4 (610049) & 30.6 & 57.7 (643546) & 30.8\\
\hspace{1em}Non-binary & 0.6 (5545) & 15.6 & 0.5 (5687) & 15.9\\
\hspace{1em}No response & 3.1 (29264) & 19 & 13.2 (146883) & 20.6\\
\addlinespace[0.3em]
\multicolumn{5}{l}{\textbf{Facebook User Language}}\\
\hspace{1em}English & 96.8 (917134) & 30 & 96 (1069824) & 29.5\\
\hspace{1em}Spanish & 3.2 (30244) & 10 & 4 (44612) & 10.1\\
\hspace{1em}Other & 0 (0) & - & 0 (0) & -\\
\addlinespace[0.3em]
\multicolumn{5}{l}{\textbf{Education level}}\\
\hspace{1em}High school, GED, or less & 19 (179567) & 18.3 & 17.7 (197564) & 18.8\\
\hspace{1em}Some college or two year degree & 35.3 (334843) & 26.8 & 31.6 (351924) & 27.3\\
\hspace{1em}Four year degree & 23.6 (223765) & 31.8 & 20.7 (231012) & 32.2\\
\hspace{1em}Masters & 13.9 (131363) & 41.5 & 12.2 (135543) & 41.8\\
\hspace{1em}JD or MD & 3.2 (29984) & 46.9 & 2.8 (31067) & 47.1\\
\hspace{1em}PhD & 2.2 (20724) & 43 & 1.9 (21410) & 43\\
\hspace{1em}No response & 2.9 (27132) & 25.8 & 13.1 (145916) & 22\\
\addlinespace[0.3em]
\multicolumn{5}{l}{\textbf{Employment}}\\
\hspace{1em}Work outside home & 37.4 (354388) & 32 & 32.9 (366618) & 32\\
\hspace{1em}Work at home & 13.3 (126067) & 20.1 & 11.6 (129359) & 20.3\\
\hspace{1em}Does not work for pay & 45.8 (433452) & 30.2 & 41.8 (465898) & 30.6\\
\hspace{1em}No response & 3.5 (33471) & 26.1 & 13.7 (152561) & 22.2\\
\addlinespace[0.3em]
\multicolumn{5}{l}{\textbf{Occupation}}\\
\hspace{1em}Healthcare & 8.6 (81621) & 70.1 & 7.6 (84363) & 70.1\\
\hspace{1em}Education & 5.6 (52691) & 37.3 & 4.8 (53993) & 37.4\\
\hspace{1em}Services & 3.6 (34052) & 13.9 & 3.2 (35576) & 14.3\\
\hspace{1em}Protective service & 0.7 (6314) & 43.1 & 0.6 (6504) & 43.3\\
\hspace{1em}Not working & 45.8 (433452) & 30.2 & 41.8 (465898) & 30.6\\
\hspace{1em}Other & 33 (312650) & 17.9 & 29 (323006) & 18\\
\hspace{1em}No response & 2.8 (26598) & 26.3 & 13 (145096) & 22.1\\
\addlinespace[0.3em]
\multicolumn{5}{l}{\textbf{High-risk conditions}}\\
\hspace{1em}None & 36.5 (345704) & 24.6 & 35.2 (392333) & 24.5\\
\hspace{1em}One & 28.2 (267362) & 30.6 & 27.1 (301554) & 30.7\\
\hspace{1em}Two & 17.6 (166883) & 33.9 & 16.9 (188635) & 34.1\\
\hspace{1em}Three or more & 16.7 (158361) & 33.4 & 16.2 (180074) & 33.6\\
\hspace{1em}Missing & 1 (9068) & 21 & 4.7 (51840) & 12.5\\
\addlinespace[0.3em]
\multicolumn{5}{l}{\textbf{Household: adults 65 and over}}\\
\hspace{1em}No & 47.3 (448427) & 21 & 46.5 (517715) & 20.3\\
\hspace{1em}Yes & 37.2 (352093) & 43.8 & 37.5 (418388) & 42.9\\
\hspace{1em}Missing & 15.5 (146858) & 20.4 & 16 (178333) & 19.9\\
\addlinespace[0.3em]
\multicolumn{5}{l}{\textbf{Children}}\\
\hspace{1em}Children in school full time & 10.2 (96769) & 22.5 & 9 (100368) & 22.5\\
\hspace{1em}Children in school part time & 4.6 (43318) & 24 & 4 (44569) & 24\\
\hspace{1em}Children not in school & 21.4 (203211) & 24.3 & 23.6 (263008) & 22.8\\
\hspace{1em}No children & 47.6 (450498) & 31.9 & 46.4 (516689) & 31.4\\
\hspace{1em}No response & 16.2 (153582) & 34.6 & 17 (189802) & 33.9\\
\addlinespace[0.3em]
\multicolumn{5}{l}{\textbf{Lives alone}}\\
\hspace{1em}Lives with others & 97.8 (926150) & 29.4 & 97.3 (1084376) & 28.8\\
\hspace{1em}Lives alone & 1.2 (10901) & 23.8 & 1.3 (14461) & 22.8\\
\hspace{1em}No response & 1.1 (10327) & 31.8 & 1.4 (15599) & 30.9\\
\addlinespace[0.3em]
\multicolumn{5}{l}{\textbf{Worried about finances}}\\
\hspace{1em}Very worried & 18.8 (178283) & 14.6 & 17.1 (190407) & 14.8\\
\hspace{1em}Somewhat worried & 24.3 (230569) & 24.9 & 22 (244857) & 25.2\\
\hspace{1em}Not too worried & 28 (265430) & 32.5 & 24.9 (277803) & 32.9\\
\hspace{1em}Not worried at all & 28.7 (272223) & 39.8 & 25.5 (284565) & 40.2\\
\hspace{1em}No response & 0.1 (873) & 32.2 & 10.5 (116804) & 21\\
\addlinespace[0.3em]
\multicolumn{5}{l}{\textbf{Received Flu Vaccine}}\\
\hspace{1em}Yes & 61.2 (580095) & 39.8 & 55 (612686) & 40.1\\
\hspace{1em}No or unsure & 37.6 (356036) & 12.7 & 33.6 (374011) & 12.8\\
\hspace{1em}No response & 0.4 (3811) & 26.7 & 10.7 (119721) & 21.2\\
\addlinespace[0.3em]
\multicolumn{5}{l}{\textbf{Ever tested for COVID-19}}\\
\hspace{1em}Never tested & 44.3 (420020) & 26.8 & 44.4 (494688) & 26.1\\
\hspace{1em}Previously tested, not positive/unknown & 45 (426020) & 33.9 & 44.5 (495480) & 33.3\\
\hspace{1em}Previously tested positive & 10.5 (99225) & 20.6 & 10.9 (121059) & 20.3\\
\hspace{1em}No response & 0.2 (2113) & 39.8 & 0.3 (3209) & 38.7\\
\bottomrule
\multicolumn{5}{l}{\rule{0pt}{1em}Note: all numbers exclude non-respondents to COVID-19 vaccination}\\
\end{longtable}
\endgroup{}

%% file: tables/sample-table-3.tex
\begingroup\fontsize{10}{12}\selectfont

\begin{longtable}[t]{lllll}
\caption{Summary statistics: county-level characteristics \label{tab:sumtab3}}\\
\toprule
\multicolumn{1}{c}{ } & \multicolumn{2}{c}{Subset: vaccine-accepting complete-cases} & \multicolumn{2}{c}{Entire February sample} \\
\cmidrule(l{3pt}r{3pt}){2-3} \cmidrule(l{3pt}r{3pt}){4-5}
Group & \% of subset (N) & \% of group vaccinated & \% of subset (N) & \% of group vaccinated\\
\midrule
\addlinespace[0.3em]
\multicolumn{5}{l}{\textbf{County urban classification}}\\
\hspace{1em}Large central metro & 22.6 (214116) & 27.1 & 22.9 (255300) & 26.2\\
\hspace{1em}Large fringe metro & 22.1 (209298) & 27.6 & 21.9 (243660) & 27.1\\
\hspace{1em}Medium metro & 27.5 (260430) & 30.6 & 27.4 (305233) & 30\\
\hspace{1em}Small metro & 11.5 (108819) & 31.9 & 11.4 (127459) & 31.1\\
\hspace{1em}Micropolitan & 10 (94437) & 30.8 & 10 (111281) & 30\\
\hspace{1em}Non-core & 6.4 (60278) & 31.8 & 6.4 (71503) & 31.3\\
\hspace{1em}No response & 0 (0) & - & 0 (0) & \vphantom{1} -\\
\addlinespace[0.3em]
\multicolumn{5}{l}{\textbf{County population per square mile quintile}}\\
\hspace{1em}1st quintile & 2 (18714) & 36.6 & 2 (21818) & 36.2\\
\hspace{1em}2nd quintile & 4 (37875) & 32.1 & 4 (44855) & 31.4\\
\hspace{1em}3rd quintile & 6.9 (65545) & 30.7 & 6.9 (77419) & 30\\
\hspace{1em}4th quintile & 14.7 (138973) & 29.9 & 14.7 (163554) & 29.2\\
\hspace{1em}5th quintile & 72.4 (686271) & 28.8 & 72.4 (806790) & 28.1\\
\addlinespace[0.3em]
\multicolumn{5}{l}{\textbf{County population quintile}}\\
\hspace{1em}1st quintile & 1.1 (10233) & 35.3 & 1.1 (12134) & 35.1\\
\hspace{1em}2nd quintile & 2.8 (26960) & 32.5 & 2.9 (32000) & 31.7\\
\hspace{1em}3rd quintile & 5.7 (54136) & 30.5 & 5.8 (64338) & 29.8\\
\hspace{1em}4th quintile & 12.9 (122343) & 30.3 & 12.9 (143984) & 29.6\\
\hspace{1em}5th quintile & 77.4 (733706) & 29 & 77.3 (861980) & 28.3\\
\addlinespace[0.3em]
\multicolumn{5}{l}{\textbf{County percent in poverty quintile}}\\
\hspace{1em}1st quintile & 30.2 (286253) & 28.7 & 29.6 (329680) & 28.3\\
\hspace{1em}2nd quintile & 21.4 (202372) & 29.7 & 21.2 (236068) & 29.1\\
\hspace{1em}3rd quintile & 23.7 (224853) & 28.9 & 23.9 (265980) & 28.1\\
\hspace{1em}4th quintile & 15.8 (149990) & 29.9 & 16.1 (179415) & 29\\
\hspace{1em}5th quintile & 8.9 (83910) & 31.3 & 9.3 (103293) & 30\\
\addlinespace[0.3em]
\multicolumn{5}{l}{\textbf{County gini quintile}}\\
\hspace{1em}1st quintile & 9.4 (88737) & 28.5 & 9.2 (102709) & 28\\
\hspace{1em}2nd quintile & 14.9 (141454) & 29.6 & 14.7 (164232) & 29.1\\
\hspace{1em}3rd quintile & 17.9 (169122) & 29 & 17.7 (197780) & 28.4\\
\hspace{1em}4th quintile & 29.3 (277916) & 29.7 & 29.2 (325916) & 29.1\\
\hspace{1em}5th quintile & 28.5 (270149) & 29.4 & 29.1 (323799) & 28.6\\
\addlinespace[0.3em]
\multicolumn{5}{l}{\textbf{County percent uninsured quintile}}\\
\hspace{1em}1st quintile & 31.1 (294470) & 29.1 & 30.5 (339515) & 28.7\\
\hspace{1em}2nd quintile & 20.8 (196935) & 29.2 & 20.7 (230724) & 28.5\\
\hspace{1em}3rd quintile & 18.5 (175327) & 29.4 & 18.6 (207344) & 28.7\\
\hspace{1em}4th quintile & 17.1 (162180) & 29.4 & 17.3 (192769) & 28.6\\
\hspace{1em}5th quintile & 12.5 (118466) & 30.3 & 12.9 (144084) & 29.3\\
\addlinespace[0.3em]
\multicolumn{5}{l}{\textbf{County Biden vote lead quintile}}\\
\hspace{1em}1st quintile & 3 (28096) & 28 & 3 (33565) & 27.3\\
\hspace{1em}2nd quintile & 5.8 (54815) & 29.6 & 5.8 (64888) & 29.1\\
\hspace{1em}3rd quintile & 10.7 (101312) & 30 & 10.7 (119217) & 29.4\\
\hspace{1em}4th quintile & 21.3 (201454) & 30.3 & 21.2 (236127) & 29.6\\
\hspace{1em}5th quintile & 59.3 (561701) & 29 & 59.3 (660639) & 28.3\\
\addlinespace[0.3em]
\multicolumn{5}{l}{\textbf{County deaths last two weeks of January sextile}}\\
\hspace{1em}1st sextile & 2.2 (20901) & 32.1 & 2.2 (24378) & 31.4\\
\hspace{1em}2nd sextile & 19.7 (186968) & 29.3 & 19.2 (213853) & 28.9\\
\hspace{1em}3rd sextile & 31.1 (295100) & 29.4 & 31.1 (346996) & 28.7\\
\hspace{1em}4th sextile & 23.9 (226849) & 29 & 24.2 (269917) & 28.2\\
\hspace{1em}5th sextile & 15.2 (143757) & 29.7 & 15.3 (170893) & 29\\
\hspace{1em}6th sextile & 7.8 (73803) & 29.4 & 7.9 (88399) & 28.5\\
\hspace{1em}No response & 0 (0) & - & 0 (0) & -\\
\bottomrule
\multicolumn{5}{l}{\rule{0pt}{1em}Note: all numbers exclude non-respondents to COVID-19 vaccination}\\
\end{longtable}
\endgroup{}

%% file: appendices/02-estimands.tex
\section{Other estimands}\label{app:estimands}

As supplemental analyses we consider three additional estimands beyond those included in the main paper.

\subsection{Observed effects}

Equation (\ref{eqn:estimand:wate}) attempts to answer the question: what would have happened in a world where all respondents in our sample had been vaccinated in February 2021 contrasted against a world where no one had been vaccinated? Imagining this former world in February 2021 is arguably challenging; since production and distribution constraints likely prevented this from being possible. We can instead ask, what was the effect of the observed distribution of vaccinations? That is, how much did vaccine distribution in the factual world reduce depression compared to a counterfactual world where no one was been vaccinated in our analytic sample? Recalling that $Z = RSV$, equation~\ref{eqn:estimand:observed} formalizes this:

\begin{align}\label{eqn:estimand:observed}
    \psi^{obs} &= \mathbb{P}_n\left(\mathbb{E}[Y - Y^0 \mid X, Z = 1]\right)
\end{align}

\subsection{Incremental effects}

We can generalize the above question to further ask: how would rates of depression, isolation, and worries about health have changed were we able to change everyone's probability of being vaccinated in February 2021? To answer this, we use the incremental intervention effects proposed by \cite{kennedy2019nonparametric}, formalized in equation~(\ref{eqn:estimand:incremenetaleffects}):

\begin{align}\label{eqn:estimand:incremenetaleffects}
    \psi^{Q(\delta)} = \mathbb{P}_n\left[\frac{\delta \pi_1(X) Y^1 + (1 - \pi_1(X))Y^0}{\delta \pi_1(X) + 1 - \pi_1(X)}\right]
\end{align}
where $\pi_1(X) = P(A = 1 \mid X, V = 1)$ and $\delta$ represents a parameter that multiplies everyone's observed odds of being vaccinated. For example, take $\delta = 1.5$ and consider an individual with a 10 percent chance of being vaccinated (or 1/9 odds). This intervention would increase this individual's odds to 3/18, leaving them with approximately 14 percent probability of being vaccinated. As special cases of interest, $\delta = 1$ corresponds to average outcomes under the observed distribution of vaccinations (or $\mathbb{P}_n[\mathbb{E}[Y \mid X, Z = 1]]$), $\delta = 0$ corresponds to the average outcomes when no one is vaccinated ($\mathbb{P}_n[\mu_0(X)]$), and $\delta \to \infty$ corresponds to the average outcomes when everyone is vaccinated ($\mathbb{P}_n[\mu_1(X)]$).

We can also estimate contrasts setting different values of $\delta$ to see how mental health would have changed under different interventions on the probability of being vaccinated. Notice that \eqref{eqn:estimand:observed} is an example of such a contrast between $\delta = 1$ and $\delta = 0$. This highlights an important distinction between \eqref{eqn:estimand:observed} and \eqref{eqn:estimand:incremenetaleffects}: the former involves a causal contrast while the latter reflects a counterfactual rate. We refer to \cite{kennedy2019nonparametric} and \cite{bonvini2021incremental} for more details.

\subsection{Pandemic effects}\label{app:estimands:pandemic}

One reason we are interested in the effect of COVID-19 vaccinations on mental health is because we believe that it may ``undo'' some of the effects of the pandemic on mental health. We may therefore also wish to understand the effects of the pandemic on mental health in a world without COVID-19 vaccinations. Let $C$ be an indicator of the COVID-19 pandemic. We can then consider the analogue to our primary causal estimand ($\psi$):

\begin{align}\label{eqn:pandemic}
\psi^c &= \mathbb{P}_n[\mathbb{E}[Y^{C = 1, A = 0} - Y^{C = 0, A = 0} \mid X, Z = 1]]
\end{align}

We can also consider effects the risk-ratio scale. Under some assumptions, we can use our estimates of $\psi$ to obtain lower bounds on both quantities (see Appendix~\ref{app:identification}).

%% file: appendices/03-identification.tex
\section{Identification}\label{app:identification}

This section details our causal models. The first subsection outlines all identifying assumptions. The second subsection provides expressions of our causal estimands in terms of the observed data distribution, and provides briefs proofs of these identification results. These results cover both the estimands in the main paper and the additional estimands outlined in Appendix~\ref{app:estimands}.

\subsection{Causal assumptions}

As noted previously, the CTIS is a repeated cross-sectional sample of Facebook Users. Additionally, the CTIS only indicates whether a respondent was vaccinated, but not when, and does not contain information on anyone's feelings of depression, isolation, or worries about health over time. While in truth the causal process may be complex (see Remark~\ref{remark2}), we make several simplifying assumptions to identify the causal effects given this data.

Let $t = 0, ..., T$ index calendar days starting from December 1, 2020 (slightly prior to the first administrations of COVID-19 vaccinations) and $T$ indexes December 1, 2021. More generally we could consider any arbitrary start date before COVID-19 vaccinations were first administered and afterwards. Define $\mathbf{a}_{1t}$ to be a fixed vector of all vaccination-statuses across all $N$ individuals in the U.S. population at time $t$. Assumption~\ref{ass:nocarryover} precludes carryover or anticipatory effects either for the outcome or the mediators during February 2021. Specifically, we assume that, for any time-period $t$ within February 2021:

\begin{assumption}[No carryover effects]\label{ass:nocarryover}
\begin{align}\label{eqn:carryover}
    &Y_{it}^{\mathbf{a}_1\mathbf{m}_{11}\mathbf{m}_{21}, ..., \mathbf{a}_T\mathbf{m}_{1T}\mathbf{m}_{2T}} = Y_{it}^{\mathbf{a}_t\mathbf{m}_{1t}\mathbf{m}_{2t}} \\
    &M_{1it}^{\mathbf{a}_1, ..., \mathbf{a}_T} = M_{1it}^{\mathbf{a}_t} \\
    &M_{2it}^{\mathbf{a}_1, ..., \mathbf{a}_T} = M_{2it}^{\mathbf{a}_t} 
\end{align}
\end{assumption}

Assumption~\ref{ass:consistency}, or consistency, further restricts the dependence of the potential outcomes (dropping the indices $it$ above for clarity). Consistency implies that each unit has potential outcomes (potential mediator values) equal to their observed outcomes (observed mediators) given their treatment and mediator values (and their treatment values), and does not depend on the treatment assignment or mediator value (or treatment assignment) of any other individual.

\begin{assumption}[Consistency]\label{ass:consistency}
\begin{align}
    &(A_{it} = a_{it}, M_{1it} = m_{1it}, M_{2it} = m_{2it}) \implies Y_{it} = Y_{it}^{a_{it}m_{1it}m_{2it}} = Y^{\mathbf{a}_{t}\mathbf{m}_{1t}\mathbf{m}_{2t}} \\
    &A_{it} = a_{it} \implies M_{1it} = M_{1it}^{a_{it}} = M_1^{\mathbf{a}_{t}}, M_{2it} = M_{2it}^{a_{it}} = M_{2it}^{\mathbf{a}_{t}}
\end{align}
\end{assumption}

To ease notation, we drop the indices $it$ moving forward. We again assume no unmeasured confounding and M-Y ignorability, detailed in equations (\ref{eqn:unconfoundedness-ya})-(\ref{eqn:unconfoundedness-ym}), and reproduce these assumptions below for completeness. 

\begin{assumption}[No unmeasured confounding]\label{ass:nuc}
\begin{align}
&(Y^{am_1m_2}, M_1^a, M_2^a) \perp A \mid (X, V = 1)
\end{align}
\end{assumption}

\begin{assumption}[Y-M ignorability]\label{ass:myi}
\begin{align}
&Y^{am_1m_2} \perp (M_1, M_2) \mid (A = a, X, V = 1) 
\end{align}
\end{assumption}

No unmeasured confounding states that treatment assignment is independent of the potential outcomes and potential mediators conditional on covariates and vaccine-acceptance, while Y-M ignorability states that the potential outcomes are independent of the mediators conditional on treatment assignment, covariates, and vaccine-acceptance. This precludes the existence of other causal descendants of $A$ that confound the mediator-outcome relationship. 

\begin{remark}
For our outcomes analysis we only require Assumption~\ref{ass:nuc} for identification. By contrast, our mediation analysis requires Assumption~\ref{ass:myi}. We also implicitly assume no reverse causation between the mediators and outcome: that is, that depression does not cause feelings of isolation and worries about health. 
\end{remark}

\begin{remark}\label{remark1}
These assumptions reduce a potentially more complex longitudinal process into a simpler structure that permits causal identification given our data. To illustrate, we consider our outcomes analysis, and let $W_t = [Y_t, M_{1t}, M_{2t}]$. For simplicity we assume Consistency and No Anticipatory Effects. We also illustrate only four time points, though the insights hold across a longer time-frame:

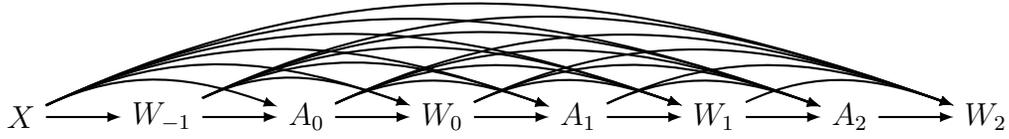
\begin{figure}[H]\label{fig:dag1}
\begin{center}
\begin{tikzpicture}
[
array/.style={rectangle split, 
	rectangle split parts = 3, 
	rectangle split horizontal, 
    minimum height = 2em
    }
]
 \node (0) {$X$};
 \node [right =of 0] (1) {$W_{-1}$};
 \node [right =of 1] (2) {$A_0$};
 \node [right =of 2] (3) {$W_0$};
 \node [right =of 3] (4) {$A_1$};
 \node [right =of 4] (5) {$W_1$};
 \node [right =of 5] (6) {$A_2$};
 \node [right =of 6] (7) {$W_2$};
 \draw[Arrow] (0.east) -- (1.west);
 \draw[Arrow] (1.east) -- (2.west);
 \draw[Arrow] (2.east) -- (3.west);
 \draw[Arrow] (3.east) -- (4.west);
 \draw[Arrow] (4.east) -- (5.west);
 \draw[Arrow] (5.east) -- (6.west);
 \draw[Arrow] (6.east) -- (7.west);
 \draw[Arrow] (0) to [out = 25, in = 160] (2);
 \draw[Arrow] (0) to [out = 25, in = 160] (3);
 \draw[Arrow] (0) to [out = 25, in = 160] (4);
 \draw[Arrow] (0) to [out = 25, in = 160] (5);
 \draw[Arrow] (0) to [out = 25, in = 160] (6);
 \draw[Arrow] (0) to [out = 25, in = 160] (7);
 \draw[Arrow] (1) to [out = 25, in = 160] (3);
 \draw[Arrow] (1) to [out = 25, in = 160] (4);
 \draw[Arrow] (1) to [out = 25, in = 160] (5);
 \draw[Arrow] (1) to [out = 25, in = 160] (6);
 \draw[Arrow] (1) to [out = 25, in = 160] (7);
 \draw[Arrow] (2) to [out = 25, in = 160] (4);
 \draw[Arrow] (2) to [out = 25, in = 160] (5);
 \draw[Arrow] (2) to [out = 25, in = 160] (6);
 \draw[Arrow] (2) to [out = 25, in = 160] (7);
 \draw[Arrow] (3) to [out = 25, in = 160] (5);
 \draw[Arrow] (3) to [out = 25, in = 160] (6);
 \draw[Arrow] (3) to [out = 25, in = 160] (7);
 \draw[Arrow] (4) to [out = 25, in = 160] (6);
 \draw[Arrow] (4) to [out = 25, in = 160] (7);
 \draw[Arrow] (5) to [out = 25, in = 160] (7);
\end{tikzpicture}
\caption{Causal structure: unrestricted}
\label{fig:dag1}
\end{center}
\end{figure}

This model allows for both vaccination status and the outcomes to depend on the prior periods vaccination status and prior outcomes. This model further implies the following conditional independence assumptions for any $(a_0, a_1, a_2)$:

\begin{align*}
&W_{0}(a_0) \perp A_0 \mid (X, W_{-1}) \\
&W_{1}(a_1, a_0) \perp A_1 \mid (X, W_{-1}, W_0, A_0) \\
&W_{2}(a_2, a_1, a_0) \perp A_2 \mid (X, W_{-1}, W_0, W_1, A_0, A_1) 
\end{align*}

Because we do not observe each person's outcomes or vaccination status over time -- only at the time they respond to the survey -- we assume no carryover effects and no unmeasured confounding. Figure~\ref{fig:dag2} illustrates these exclusions contrasted against Figure~\ref{fig:dag1}. First, we assume that vaccination status at any time $t$ is not a function of any contemporaneous or prior outcomes. Second, we assume that that the outcomes at time $t$ are only a function of that period's treatment status. We do allow that prior-period outcomes can affect current outcomes; however, they cannot affect treatment assignment at any point; otherwise, these variables would be a confounder.

\begin{figure}[H]
\begin{center}
\begin{tikzpicture}
[
array/.style={rectangle split, 
	rectangle split parts = 3, 
	rectangle split horizontal, 
    minimum height = 2em
    }
]
 \node (0) {$X$};
 \node [right =of 0] (1) {$W_{-1}$};
 \node [right =of 1] (2) {$A_0$};
 \node [right =of 2] (3) {$W_0$};
 \node [right =of 3] (4) {$A_1$};
 \node [right =of 4] (5) {$W_1$};
 \node [right =of 5] (6) {$A_2$};
 \node [right =of 6] (7) {$W_2$};
 \draw[Arrow] (0.east) -- (1.west);
% \draw[Arrow] (1.east) -- (2.west);
 \draw[Arrow] (2.east) -- (3.west);
% \draw[Arrow] (3.east) -- (4.west);
 \draw[Arrow] (4.east) -- (5.west);
% \draw[Arrow] (5.east) -- (6.west);
 \draw[Arrow] (6.east) -- (7.west);
 \draw[Arrow] (0) to [out = 25, in = 160] (2);
 \draw[Arrow] (0) to [out = 25, in = 160] (3);
 \draw[Arrow] (0) to [out = 25, in = 160] (4);
 \draw[Arrow] (0) to [out = 25, in = 160] (5);
 \draw[Arrow] (0) to [out = 25, in = 160] (6);
 \draw[Arrow] (0) to [out = 25, in = 160] (7);
 \draw[Arrow] (1) to [out = 25, in = 160] (3);
% \draw[Arrow] (1) to [out = 25, in = 160] (4);
 \draw[Arrow] (1) to [out = 25, in = 160] (5);
% \draw[Arrow] (1) to [out = 25, in = 160] (6);
 \draw[Arrow] (1) to [out = 25, in = 160] (7);
 \draw[Arrow] (2) to [out = 25, in = 160] (4);
% \draw[Arrow] (2) to [out = 25, in = 160] (5);
 \draw[Arrow] (2) to [out = 25, in = 160] (6);
% \draw[Arrow] (2) to [out = 25, in = 160] (7);
 \draw[Arrow] (3) to [out = 25, in = 160] (5);
% \draw[Arrow] (3) to [out = 25, in = 160] (6);
 \draw[Arrow] (3) to [out = 25, in = 160] (7);
 \draw[Arrow] (4) to [out = 25, in = 160] (6);
% \draw[Arrow] (4) to [out = 25, in = 160] (7);
 \draw[Arrow] (5) to [out = 25, in = 160] (7);
\end{tikzpicture}
\caption{Causal structure: restricted}
\label{fig:dag2}
\end{center}
\end{figure}
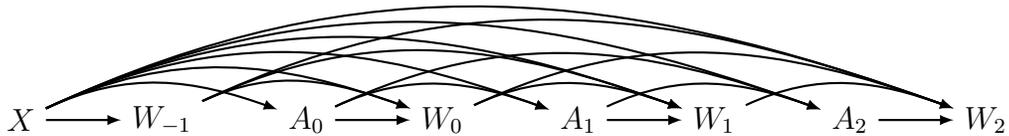
\end{remark}

\begin{remark}\label{remark2}
Allowing for carryover effects would imply that only causal estimands at time $t = 0$ would be identified (again assuming $W_{-1}$ is not a confounder). This motivates our decision to analyze data only from February 2021: this time-period is ``closer'' to $t = 0$. Heuristically, we may expect that the bias induced from violations of this assumption are smaller when closer to this ``first'' time-period and that this bias grows in magnitude over time. This may occur, for example, when treatment effects change monotonically with time since vaccination. As a related point, if our observed covariates are highly correlated with time since vaccination, controlling for these variables may reduce this bias.
\end{remark}

\begin{remark}
Modeling carryover effects may be of interest. For example, we may wish to know whether the effects of COVID-19 vaccinations on mental health endure over time. Unfortunately, even if we observed the differences in a pair of potential outcomes for each individual sampled at time $t$, we could not necessarily identify how effects change as a function of time since treatment (``timevacc''). To illustrate this point, consider the model:

\begin{align*}
Y_{it}^1 - Y_{it}^0 = \beta_{0X} + \beta_{1X} t + \beta_{2X} timevacc_{it} 
\end{align*}
where $timevacc_{it} = \max\{0, \sum_{k=0}^t(A_{ik} - 1)\}$. $\beta_{0X}$ captures the initial effect of treatment, $\beta_{1X}$ represents heterogeneity in the treatment effect as a linear function of the time index, $\beta_{2X}$ captures the durability of the treatment effect over time, and the coefficients are a function of each individual's covariates $X$. Assume that we observe $Y_{it}^1 - Y_{it}^0$ for each individual when sampled at time $t$, and let $Y_{it}^1 - Y_{it}^0$ equal $\beta_{0X} + \beta_{1X} t$ for individuals not yet vaccinated. That is, the treatment effect is assumed to equal what would have happened were they vaccinated that same time period.

Consider the case where for all $x$, $\beta_{0x} > 0$, $\beta_{1x} > 0$, $\beta_{2x} < 0$, and $\lvert\beta_{1x}\rvert > \lvert\beta_{2x}\rvert$. That is, there is some positive initial effect of treatment which increases with the time index, but that decreases with time since treatment (i.e. the effects are not ``durable''). In this case we would observe that the average effects increase with time; however, we would be wrong to infer that effects are increasing with the length of treatment. This is because in this scenario we cannot separately identify heterogeneity across time from carryover effects.

This illustrates the difficulty of modeling such effects generally. And in fact, as noted in Remark~\ref{remark2}, we must actually preclude the existence of carryover effects identify our causal estimands given our data structure. Therefore, our particular data makes it difficult, if not impossible, to learn about effect durability, or carryover effects generally.
\end{remark}

We next assume random sample selection to rule out collider-bias by conditioning on our sample. This assumption is detailed in \eqref{eqn:amarx}, but reproduced below for completeness.

\begin{assumption}[Random sample selection]\label{ass:amarx1}
\begin{align}
    \label{eqn:marx1}& (S, R) \perp (Y^{am_1m_2}, M_1^a, M_2^a) \mid (A, X, V = 1) \\
\end{align}
\end{assumption}
where $R = R_yR_aR_{m_1}R_{m_2}$. 

Finally, we assume positivity of the probability of treatment assignment, the joint mediator probabilities, and being a complete-case for all values of $(a, m_1, m_2)$ in the support of the data and $\epsilon > 0$:\footnote{Our primary estimand is the complete-case estimator, which does not technically require \eqref{eqn:positivity3}; however, this is required to estimate the effects over all non-hesitant respondents where $V = 1$.}

\begin{assumption}[Positivity]\label{ass:positivity}
\begin{align}
    \label{eqn:positivity1}&P\left\{\min_a P(A = a \mid Z = 1, X) > \epsilon\right\} = 1 \\
    \label{eqn:positivity2}&P\left\{\min_{a, m_1, m_2}P(M_1 = m_1, M_2 = m_2 \mid Z = 1, A = a, X) > \epsilon\right\} = 1 \\
    \label{eqn:positivity3}&P\left\{P(R = 1 \mid V = 1, S = 1, X) > \epsilon\right\} = 1 
\end{align}
\end{assumption}

These assumptions are sufficient to identify our causal parameters.

\subsubsection{Identification}

We provide all identification results conditional on $X$, noting that the target estimands are simply the empirical average of these expressions across the analytic sample.

\begin{proposition}[Mean potential outcomes]\label{prop:mpo}
Let $W^a = [Y^a, M_1^a, M_2^a]$. Under assumptions (\ref{ass:nocarryover}), (\ref{ass:consistency}) (\ref{ass:nuc}), (\ref{ass:amarx1}), and (\ref{ass:positivity}) $\mathbb{E}[W^a \mid X, Z = 1] = \mathbb{E}[W \mid X, A = a, Z = 1]$.
\end{proposition}

\begin{proof}[Proof of Proposition~\ref{prop:mpo}]\label{proof:mpo}
First, notice that $Y^a = \sum_{m_1', m_2'}\mathds{1}(M_1^a = m_1, M_2^a = m_2)Y^{am_1m_2}$. Therefore, $(Y^{am_1m_2}, M_1^a, M_2^a) \perp A \mid X \implies (Y^a, M_1^a, M_2^a) \perp A \mid X$. Then: 

\begin{align*}
\mathbb{E}[W^a \mid Z = 1, X] &= \sum_{a'}\mathbb{E}[W^a \mid Z = 1, X, A = a']P(A = a' \mid X, Z = 1) \\
&= \sum_{a'}\mathbb{E}[W^a \mid X, V = 1, A = a']P(A = a' \mid X, Z = 1) \\
&= \sum_{a'}\mathbb{E}[W^a \mid X, V = 1, A = a]P(A = a' \mid X, Z = 1) \\
&= \mathbb{E}[W^a \mid X, V = 1, A = a] \\
&= \mathbb{E}[W^a \mid X, Z = 1, A = a] \\ 
&= \mathbb{E}[W \mid X, Z = 1, A = a]
\end{align*}
where the first equality holds by the Law of Iterated Expectations, the second by random sample selection, the third by no unmeasured confounding, the fourth follows directly by the third, the fifth by no sample selection, and the sixth by consistency.
\end{proof}

\begin{proposition}[Joint counterfactual outcomes]\label{prop:jointpo}
The following identification result holds:
\begin{align*}
\mathbb{E}[Y^{am_1m_2} \mid Z = 1, X] = \mathbb{E}[Y \mid Z = 1, X, A = a, M_1 = m_1, M_2 = m_2]
\end{align*}
\end{proposition}

\begin{proof}[Proof of Proposition~\ref{prop:jointpo}]
Let $M = M_1 \times M_2$. Then:
\begin{align*}
\mathbb{E}[Y^{am} \mid Z = 1, X] &= \mathbb{E}[Y^{am} \mid X, V = 1, A = a] \\
&= \mathbb{E}[Y^{am} \mid X, V = 1, A = a, M = m] \\ 
&= \mathbb{E}[Y^{am} \mid X, A = a, M^a = m, V = 1, RS = 1] \\ 
&= \mathbb{E}[Y \mid X, A = a, M = m, Z = 1]  
\end{align*}
where the first equality holds by LIE, no sample selection, and no unmeasured confounding (following the first four equalities in Proof~\ref{proof:mpo}); the second by Y-M ignorability, the third by consistency and the fact that $(R, S) \perp (Y^{am}, M^a) \mid X, A, V = 1 \implies (R, S) \perp Y^{am} \mid (M^a = m, X, V = 1, A = a)$, and the final line by consistency.
\end{proof}

\begin{proposition}[Interventional effects]\label{prop:int}
We provide the identifying expression for $\psi_{M_1}$, noting that the other interventional effects can be derived similarly.

\begin{align*}
    &\sum_{m_1, m_2}\mathbb{E}[Y^{1m_1m_2} \mid X, Z = 1]\{p(M_1^1 = m_1 \mid X, Z = 1) \\ &- p(M_1^0 = m_1 \mid X, Z = 1)\}p(M_2^0 = m_2 \mid X, Z = 1) \\
    &= \sum_{m_1, m_2}\mathbb{E}[Y \mid X, 1, m_1, m_2, Z = 1]\{p(m_1 \mid X, 1, Z = 1) \\
    &- p(m_1 \mid X, 0, Z = 1)\}p(m_2 \mid X, 0, Z = 1)
\end{align*}
\end{proposition}

\begin{proof}[Proof of Proposition~\ref{prop:int}]
The proof follows directly from Propositions \ref{proposition:mpo} and \ref{proposition:jointpo}.
\end{proof}

\begin{remark}
Equation \ref{eqn:positivity2} is technically stronger than necessary to identify $\psi_{M_1}$. For example, for any $x$, we only need that $p(m_1, m_2 \mid 1, x) > 0$ whenever $p(m_2 \mid 0, x) > 0$. 
\end{remark}

\begin{proposition}[Incremental effects]\label{prop:incfx}
Assume \eqref{eqn:marax}
\begin{align}\label{eqn:marax}
(R, S) \perp (A, Y^a) \mid (X, V = 1)
\end{align}

In other words, missingness is independent of treatment assignment given covariates and vaccine-acceptance. Then:

\begin{align*}
\mathbb{E}[Y^{Q(\delta)} \mid X, Z = 1] &= \frac{\delta P(A = 1 \mid X, Z = 1)\mathbb{E}[Y \mid A = 1, X, Z = 1]}{\delta P(A = 1 \mid X, Z = 1) + P(A = 0 \mid X, Z = 1)} \\
&+ \frac{P(A = 0 \mid X, Z = 1)\mathbb{E}[Y \mid A = 0, X, Z = 1]}{\delta P(A = 1 \mid X, Z = 1) + P(A = 0 \mid X, Z = 1)}
\end{align*}
\end{proposition}

\begin{proof}[Proof of Proposition~\ref{prop:incfx}]
Under \eqref{eqn:marax}, $P(A = 1 \mid X, V = 1) = P(A = 1 \mid X, Z = 1)$. Moreover,

\begin{align*}
\mathbb{E}[Y^a \mid X, Z = 1] &= \mathbb{E}[Y^a \mid X, V = 1] \\
&= \mathbb{E}[Y^a \mid X, V = 1, A = a] \\
&= \mathbb{E}[Y^a \mid X, Z = 1, A = a] \\
&= \mathbb{E}[Y \mid X, Z = 1, A = a]
\end{align*}
where the first equality follows by \eqref{eqn:marax}, the second by no unmeasured confounding, the third by \eqref{eqn:marax}, and the final line by consistency. The result follows by combining these results.
\end{proof}

\begin{remark}
This estimand requires no positivity assumption with respect to $\pi_a(X)$ (\eqref{eqn:positivity1}) for identification (\cite{kennedy2019nonparametric}).
\end{remark}

\begin{remark}
Unlike our other analyses, we additionally assume \eqref{eqn:marax}. This precludes the missingness process or sample selection process from being a function of treatment assignment given the covariates and vaccine-acceptance. This assumption may not hold: for example, \cite{bradley2021unrepresentative} shows that the CTIS overrepresents vaccine-uptake compared to the U.S. population, suggesting that perhaps there is selection into the sample based on this variable. On the other hand, \cite{salomon2021us} also shows that the CTIS overrepresents individuals with higher education, a variable also associated with higher vaccination rates. As long as higher vaccine-uptake among CTIS respondents can by fully explained by other observed covariates, then \eqref{eqn:marax} could hold.

Alternatively, we could imagine these shifts as the result of an intervention that only targeted respondents to the CTIS, so that the intervention is on $P(A = 1 \mid X, Z = 1)$ rather than $P(A = 1 \mid X, V = 1)$. However, no feasible real-world intervention on the treatment assignment process would likely correspond to such an intervention, rendering such an estimand arguably less useful to consider.
\end{remark}

\begin{proposition}[Pandemic bounds]\label{prop3}
Assume the following:

\begin{assumption}\label{assumption:pandemic}
\begin{align}
\label{eqn:pandemica1}&\mathbb{E}[Y^{A = 1} \mid X = x, V = 1] \ge \mathbb{E}[Y^{C = 0, A = 0} \mid X = x, V = 1] \qquad \forall x 
\end{align}
\end{assumption}

It follows that:

\begin{align}
&\mathbb{E}[Y^{A = 0} - Y^{A = 1} \mid X = x, V = 1] \le \mathbb{E}[Y^{C = 1, A = 0} - Y^{C = 0, A = 0} \mid X = x, V = 1] \\
&\frac{\mathbb{E}[Y^{A = 0} \mid X = x, V = 1]}{\mathbb{E}[Y^{A = 1} \mid X = x, V = 1]} \le \frac{\mathbb{E}[Y^{C = 1, A = 0} \mid X = x, V = 1]}{\mathbb{E}[Y^{C = 0, A = 0} \mid X = x, V = 1]}
\end{align}
\end{proposition}

\begin{proof}[Proof of Proposition~\ref{prop3}]
We use the identity that for all $x$: 

\begin{align*}
\mathbb{E}[Y^{A = 0} \mid V = 1, X = x] &= \mathbb{E}[Y^{C = 1, A = 0} \mid V = 1, X = x] 
\end{align*}

This holds by definition: our estimates all take place in a world in which there is a COVID-19 pandemic. The remainder of the proof follows directly from \eqref{eqn:pandemica1}.
\end{proof}

\begin{remark}
Equation~\ref{eqn:pandemica1} is perhaps better motivated by the following assumed inequalities:

\begin{align*}
\mathbb{E}[Y^{A = 1} \mid V = 1, X = x] &= \mathbb{E}[Y^{C = 1, A = 1} \mid V = 1, X = x] \\
&\ge \mathbb{E}[Y^{C = 0, A = 1} \mid V = 1, X = x] \\
&= \mathbb{E}[Y^{C = 0, A = 0} \mid V = 1, X = x] 
\end{align*}

The first equality holds by definition. The second inequality imagines a counterfactual world where everyone were vaccinated for COVID-19 and there were no COVID-19 pandemic. While perhaps a strange counterfactual, we believe this is reasonable: all else equal, rates of depression would on average be lower absent a COVID-19 pandemic, regardless of COVID-19 vaccination status. The final equality states that these rates are equal those in a world with no COVID-19 pandemic and no COVID-19 vaccinations. This seems reasonable if we believe that the effect of COVID-19 vaccines on mental health arguably only operate during a COVID-19 pandemic. Assumption~\ref{assumption:pandemic} then follows.
\end{remark}

\begin{remark}
If we further assume random sample selection, then the same the identifying expressions for $\mathbb{E}[Y^a \mid X, Z = 1]$ detailed Proposition~\ref{proposition:mpo} allows us to identify the bounds conditional on our sample. This then allows us to interpret our sample-average effect estimates, for example, as lower absolute magnitude bounds on the effects of the COVID-19 pandemic on mental health among respondents in our sample.
\end{remark}

\begin{remark}
Using potential outcomes in this setting is perhaps philosophically challenging since ``the COVID-19 pandemic'' is arguably not a well-defined exposure, nor is it manipulable. Nevertheless, in general terms we can consider: what would have happened if COVID-19 and any of the associated responses had never happened? Equation~\ref{assumption:pandemic} simply states that expected mental health in this world would be at least as ``good'' as in a version of our world where COVID-19 emerged but everyone were vaccinated by February 2021. This assumption implicitly precludes the possibility that this counterfactual world would have realized some other large negative shock to mental health.
\end{remark}

%% file: appendices/04-estimation.tex
\section{Estimation}\label{app:estimation}

This section provides additional details on our estimation approach. As noted previously, we use influence-function based estimators throughout.\footnote{Technically, these estimators are derived treating the covariates as if representative from some super-population with covariate distribution $dP(X \mid Z = 1)$.} These estimators provide consistent estimates of the sample average treatment effects and possibly conservative variance estimates (\cite{imbens2004nonparametric}). Specifically, all of our estimators take the form:

\begin{align*}
    \frac{1}{n}\sum_{i=1}^n\frac{\mathds{1}(R_i = 1, V_i = 1)}{P(R_i = 1, V_i = 1 \mid S_i = 1)}\hat{\varphi}(O)
\end{align*}
where $\hat{\varphi}(O)$ represents a scaled estimate of the uncentered influence function for any of the functionals considered in Section~\ref{app:identification}. To be precise, we use estimates of $\varphi$ for the average potential outcome under treatment $A = a$, the incremental effects, and interventional effects via $M_1$, respectively: 

\begin{align*}
    \varphi_{1, a}(O) = \frac{\mathds{1}(A = a)}{\pi_a(X)}\left(Y - \mu_a(X)\right) + \mu_a(X)
\end{align*}

\begin{align*}
    \varphi_2(O) = \frac{\delta \pi_1(X)\varphi_{1, 1}(Z) + \pi_0(X)\varphi_{1, 0}(Z)}{\delta \pi_1(X) + \pi_0(X)} + \frac{\delta\{\mu_1(X) - \mu_0(X)\}(A - \pi_1(X))}{\{\delta \pi_1(X) + \pi_0(X)\}^2}
\end{align*}

\begin{align*}
    \varphi_3(O) &= \frac{\mathds{1}(A = a)}{\pi_a(X)}\frac{\{p(M_1 \mid a, X) - p(M_1 \mid a', X)\}p(M_2 \mid a', X)}{p(M_1, M_2, \mid a, X)}(Y - \mu_a(M_1, M_2, X)) \\
    \nonumber &+ \frac{\mathds{1}(A = a)}{\pi_a(X)}\{\mu_{a, M_2^'}(M_1, X) - \mu_{a, M_1\times M_2^'}(X)\} \\ 
    \nonumber &- \frac{\mathds{1}(A = a')}{\pi_{a'}(X)}\{\mu_{a, M_2^'}(M_1, X) - \mu_{a, M_1^'\times M_2^'}(X)\} \\
    \nonumber &+ \frac{\mathds{1}(A = a')}{\pi_{a'}(X)}\left(\mu_{a, M_1}(M_2, X) - \mu_{a, M_1\times M_2'}(X) - (\mu_{a, M_1^'}(M_2, X) - \mu_{a, M_1^'\times M_2^'}(X))\right) \\
    \nonumber &+ \mu_{a, M_1\times M_2^'}(X) - \mu_{a, M_1^'\times M_2^'}(X)
\end{align*}
where

\begin{align*}
&\pi_a(X) = P(A = a \mid X, Z = 1) \\
&\mu_a(X) = \mathbb{E}[Y \mid X, A = a, Z = 1] \\
&\mu_a(M_1, M_2, X) = \mathbb{E}[Y \mid X, M_1, M_2, A = a, Z = 1] \\
&p(. \mid a, X) = p(. \mid a, X, Z = 1) \\
&\mu_{a, M_1\times M_2'}(X) = \sum_{m_1, m_2}\mu_a(m_1, m_2, X)p(m_1 \mid a, X)p(m_2 \mid a', X)  \\
&\mu_{a, M_1^'}(M_2, X) = \sum_{m_1}\mu_a(m_1, M_2, X)p(m_1 \mid a', X)
\end{align*}
with other terms defined analogously. The result for $\varphi_1(O)$ is well-known (see, e.g. \cite{kennedy2019nonparametric}). The derivation of $\varphi_2(O)$ is contained in \cite{kennedy2019nonparametric} and for $\varphi_3(O)$ in \cite{benkeser2021nonparametric}. We refer to \cite{benkeser2021nonparametric} for the expressions for the one-step estimators for the remaining interventional effects.

To obtain estimates of the effects across all vaccine-accepting respondents (e.g. averaged over the entire stratum where $SV = 1$), we multiply the bias-correction terms (specifically, the terms involving inverse-probability weights) by $\frac{\mathds{1}(R = 1)}{P(R = 1 \mid X, V = 1, S = 1)}$ and average the resulting function over all observations where $V = 1$ (and $S = 1$) to obtain our estimates. Our variance estimates are estimated by taking the empirical variance of the estimated influence functions across our analytic sample.

We use logistic regression and XGBoost to estimate all nuisance functions. In the former case we estimate each function on the full-sample, and in the latter we use cross-fitting. Additionally, for the XGBoost estimates, we estimate twenty specifications across a grid of hyper-parameters using the ``xgboost'', ``tidymodels'', and ``tune'' packages, and stack the resulting models using the ``stacks'' package in R.

%% file: appendices/05-results.tex
\section{Additional results}\label{app:results}

This section contains additional results not displayed in the main paper. We first present model diagnostics, and then present additional tables and figures with respect to our outcomes analysis, mediation analysis, and incremental effects analysis. Results from the incremental effects analysis are only contained in this Appendix.

\subsection{Model diagnostics}\label{app:results-diagnostics}

We next examine the distribution of propensity score (``pi'') and complete-case weights (``eta''), noting again that the complete-case weights are only used for our analyses in Appendix~\ref{app:robustness}. Table~\ref{tab:weightdiagnostics} displays the quantiles of the distribution for each of these models.

  \begingroup%
  \renewcommand\normalsize{\footnotesize}% Specify your font modification
  \input{tables/weight_table}%
  \endgroup%

Figure~\ref{fig:calibration} displays the calibration curves with respect to our nuisance functions for our outcomes analysis. We display results for both our GLM and XGBoost models. $R$ reflects the calibration of the complete case model; $A$ represents the vaccination model; $RA$ represents the product of vaccination model with the complete case model (and $(1-A)R$ analogously); $Y$ represents the depression; $M1$ the (binary) isolation model; and $M2$ represents the (binary) worries about health model. The complete-case model -- used only for the full sample estimates displayed in Appendix~\ref{app:robustness} -- appear to under-predict the probabilities of being a complete-case at the lower end of the estimated probabilities. On the other hand, the product of the complete-case model and the propensity score models are well-calibrated with respect to the observed data. 

\begin{figure}[H]
\begin{center}
    \includegraphics[scale=0.5]{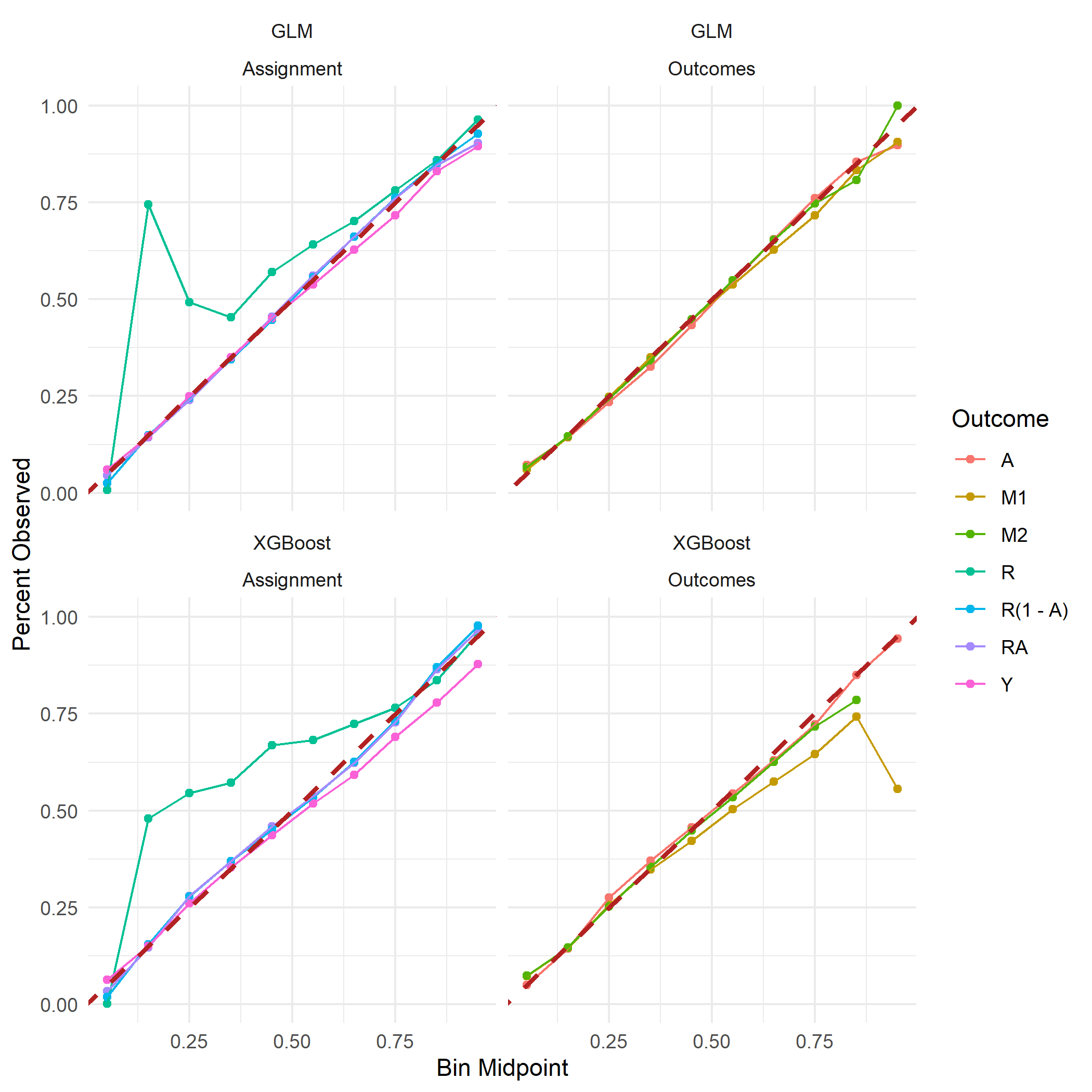}
    \caption{Model calibration}
    \label{fig:calibration}
\end{center}
\end{figure}

%Tables~\ref{tab:calibglm} and ~\ref{tab:calibxgb} present these same numbers. 

%\myinput{tables/calibration-glm}

%\myinput{tables/calibration-xbg}

\subsection{Outcomes analysis}\label{app:results-outcomes}

Figure~\ref{fig:htemedprops} displays our heterogeneity analysis for pre-specified subgroups for our outcomes analysis using our XGBoost estimates, analogous to Figure~\ref{fig:alloutcomes} in Section~\ref{sec:results}. 

The following tables display the subgroup estimates for the pre-specified subgroups with respect to our outcomes analysis. The first table presents results on the risk-differences scale and the second on the risk-ratio scale. The second two tables are identical but using XGBoost. In general we see that the results are quite similar except for the non-response category, which is much smaller for the XGBoost than the logistic regression models. 

\begin{figure}[H]
\begin{center}
    \includegraphics[scale=0.5]{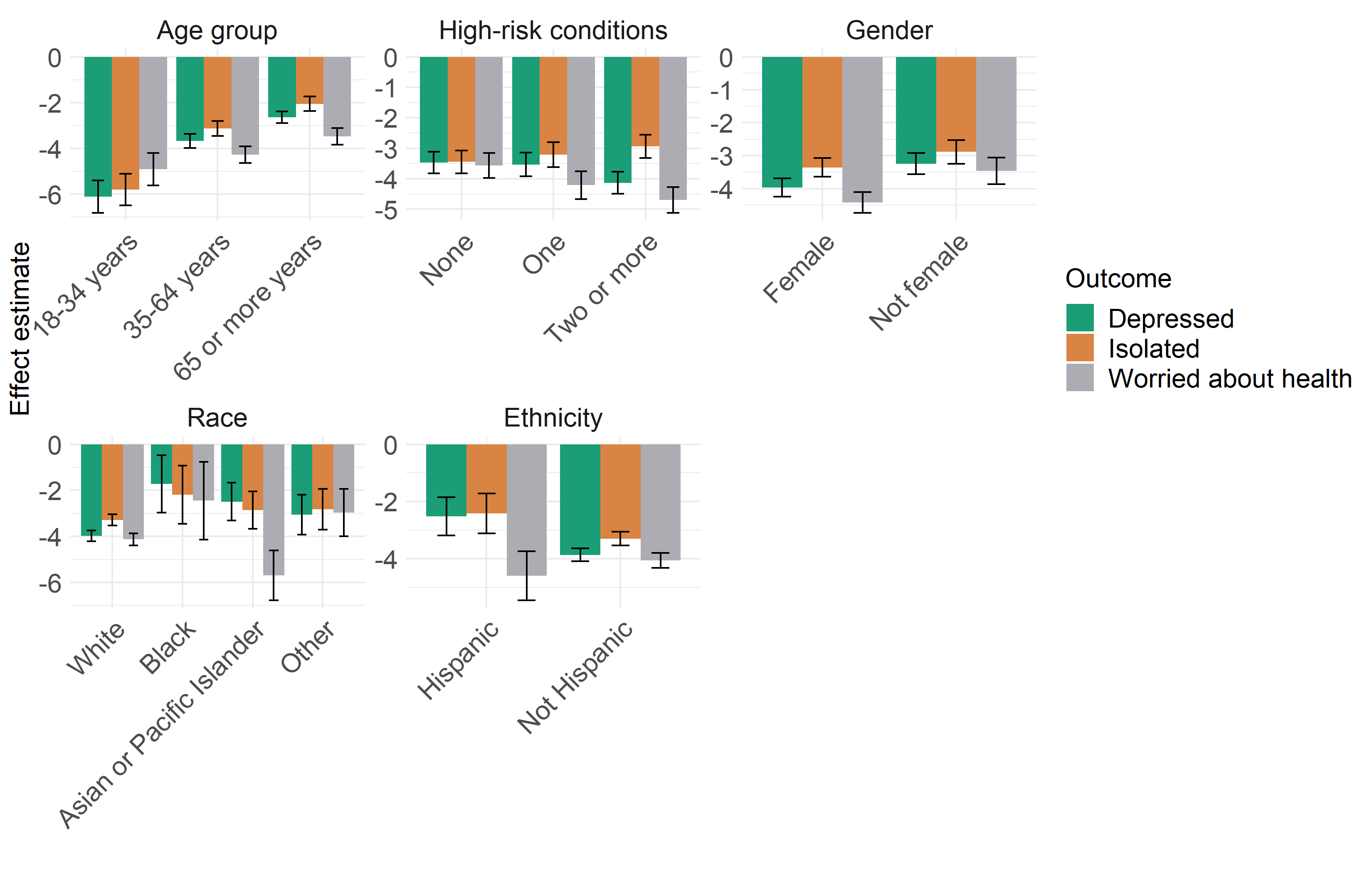}
    \caption{Outcome heterogeneity: complete-case estimates (XGBoost estimates)}
    \label{fig:htemedprops}
\end{center}
\end{figure}

\begin{landscape}

  \begingroup%
  \renewcommand\normalsize{\footnotesize}% Specify your font modification
  \input{tables/outcomes-hte-table-rd}%
  \endgroup%

  \begingroup%
  \renewcommand\normalsize{\footnotesize}% Specify your font modification
  \input{tables/outcomes-hte-table-rr}%
  \endgroup%

  \begingroup%
  \renewcommand\normalsize{\footnotesize}% Specify your font modification
  \input{tables/outcomes-hte-table-rd-xgb}%
  \endgroup%

  \begingroup%
  \renewcommand\normalsize{\footnotesize}% Specify your font modification
  \input{tables/outcomes-hte-table-rr-xgb}%
  \endgroup%

\end{landscape}

\newpage

This next set of tables display the results of the subgroup estimates for the data-driven groups for our outcomes analysis. The first set of tables displays the results for the groups that we expected to be heterogeneous with respect to depression, which we discuss in the primary paper. In addition to this analysis, we also present the results for the subgroups that we examined for isolation and worries about health, though these results were largely inconclusive. All estimates are using XGBoost. We also estimates these numbers using logistic regression models. Many of these patterns in heterogeneity no longer appear as prominent, likely in part because we do not include any interactions between covariates in any of these models. These results are available on request.

\begin{landscape}

  \begingroup%
  \renewcommand\normalsize{\footnotesize}% Specify your font modification
  \input{tables/dep-tree-table-rd-xgb}%
  \endgroup%

  \begingroup%
  \renewcommand\normalsize{\footnotesize}% Specify your font modification
  \input{tables/dep-tree-table-rr-xgb}%
  \endgroup%

  \begingroup%
  \renewcommand\normalsize{\footnotesize}% Specify your font modification
  \input{tables/health-tree-table-xgb}%
  \endgroup%

  \begingroup%
  \renewcommand\normalsize{\footnotesize}% Specify your font modification
  \input{tables/isolation-tree-table-xgb}%
  \endgroup%

\end{landscape}

\subsection{Interventional effects}\label{app:results-mediation}

Figure~\ref{fig:xghteprops} displays the heterogeneity results from our XGBoost estimates, analogous to Figure~\ref{fig:htemedpropspaper} in Section~\ref{sec:results}. The subsequent tables present results pertaining to our mediation analysis. Table~\ref{tab:medtab1} contains all point estimates on average across our analytic sample using our GLM estimates, while Table~\ref{tab:medtab6} presents our XGBoost estimates.

\begin{figure}[H]
\begin{center}
    \includegraphics[scale=0.5]{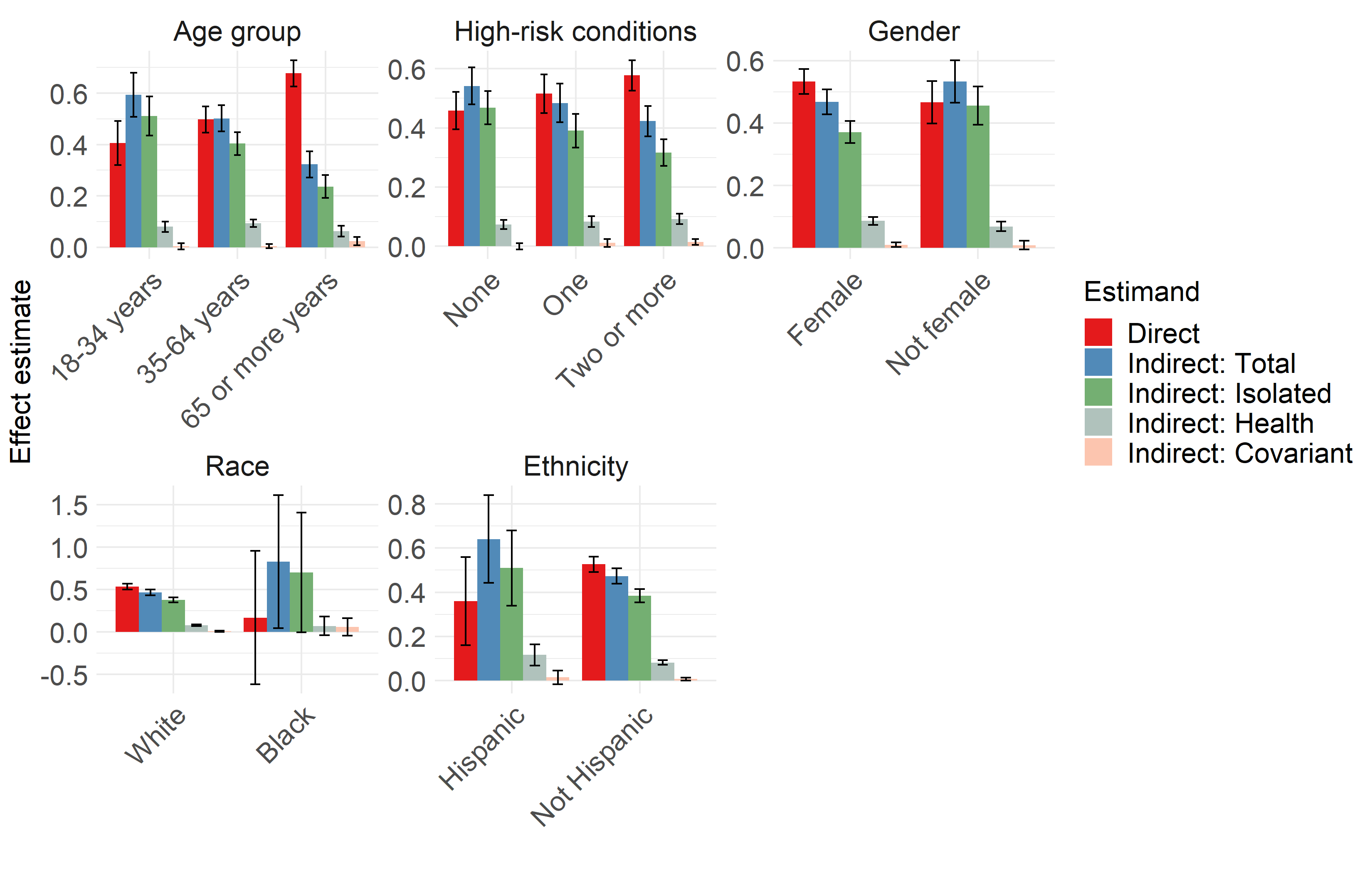}
    \caption{Mediation heterogeneity, complete-case estimates (XGBoost estimates)}
    \label{fig:xghteprops}
\end{center}
\end{figure}

  \begingroup%
  \renewcommand\normalsize{\footnotesize}% Specify your font modification
  \input{tables/mediation-table-one}  \endgroup%

  \begingroup%
  \renewcommand\normalsize{\footnotesize}% Specify your font modification
  \input{tables/mediation-table-six}  \endgroup%

The next sets of tables below display results contain results our subgroup effect estimates with respect to our pre-specified subgroups. The first two tables are on the risk-differences scale and set the indirect effect reference category at 0 and 1, respectively. The second two tables express the quantities as a proportion of the total effect. The final four tables repeat these analyses using our XGBoost estimates.

\begin{landscape}

  \begingroup%
  \renewcommand\normalsize{\footnotesize}% Specify your font modification
  \input{tables/mediation-hte-table-0}%
  \endgroup%

  \begingroup%
  \renewcommand\normalsize{\footnotesize}% Specify your font modification
  \input{tables/mediation-hte-table-1}%
  \endgroup%

  \begingroup%
  \renewcommand\normalsize{\footnotesize}% Specify your font modification
  \input{tables/mediation-hte-prop-table-0}%
  \endgroup%

  \begingroup%
  \renewcommand\normalsize{\footnotesize}% Specify your font modification
  \input{tables/mediation-hte-prop-table-1}%
  \endgroup%

\end{landscape}

\begin{landscape}

  \begingroup%
  \renewcommand\normalsize{\footnotesize}% Specify your font modification
  \input{tables/mediation-hte-table-0-xgb}%
  \endgroup%

  \begingroup%
  \renewcommand\normalsize{\footnotesize}% Specify your font modification
  \input{tables/mediation-hte-table-1-xgb}%
  \endgroup%

  \begingroup%
  \renewcommand\normalsize{\footnotesize}% Specify your font modification
  \input{tables/mediation-hte-prop-table-0-xgb}%
  \endgroup%

  \begingroup%
  \renewcommand\normalsize{\footnotesize}% Specify your font modification
  \input{tables/mediation-hte-prop-table-1-xgb}%
  \endgroup%

\end{landscape}

\subsection{Other estimands}\label{app:results-incremental}

We consider the effects of the observed distribution of vaccinations and incremental effects on the propensity scores. Table~\ref{tab:inctabxgbcc} shows the results from the XGBoost models, displaying the estimated rate of each outcome and uniform confidence bands setting $\delta = 0$ to $\delta = 5$. Again $\delta$ is a parameter that multiplies each individual's odds of being vaccinated, and reflects the expected prevalence of each outcome under the corresponding intervention $\mathbb{P}_n[Y^{Q(\delta)}]$. We also consider the effect of the observed distribution of vaccinations, which is simply the contrast $\mathbb{P}_n[Y^{Q(1)} - Y^{Q(0)}]$. 

  \begingroup%
  \renewcommand\normalsize{\footnotesize}% Specify your font modification
  \input{tables/outcomes-inc-table-xgb-cc.tex}%
  \endgroup%

We estimate that the observed distribution of vaccinations in February 2021 caused a -1.2 (-1.3, -1.2), -1.1 (-1.2, -1.0), and -1.4 (-1.5, -1.3) percentage point change in depression, isolation, and worries about health, respectively. Using our GLM models we estimate -1.2 (-1.3, -1.1), -1.1 (-1.2, -1.0), and -1.4 (-1.6, -1.3) percentage point decreases in each outcome (see Table~\ref{tab:inctabglmcc} below). This implies that the observed distribution of vaccinations captured around one-third of the total estimated effects for each outcome. By comparison, 36.7 percent of this sample was vaccinated.

The incremental effect estimates exhibit diminishing marginal returns with respect to increasing the odds of vaccination. For example, increasing the odds of vaccination two-fold would have yielded only approximately 0.3 to 0.4 percentage point reductions in the prevalence of each outcome, while increasing the odds five-fold would only yield total reductions of approximately 1 percentage point for each outcome. Taken together with our average effect estimates, these estimates suggest that further investments in vaccine production or distribution during this time-frame may not have been ``worth it'' in terms of mental health past some point. Of course, a formal cost-benefit analysis would both need to quantify the both the costs of such an intervention, and the total benefits -- not only in terms of mental health -- but economic and health benefits. 

We also consider these estimands by age group and ethnicity. Figure~\ref{fig:incplothet} illustrates the incremental effect curves across a range of $\delta$. The curves look somewhat different for each group: this reflects both the heterogeneity in the treatment effects and the different distribution of propensity scores across both groups. The observed distribution of vaccinations changed depression by -1.3 (-1.9, -0.7) percentage points among the youngest age group, accounting for approximately 11 percent of the total effect among this group. The observed distribution of vaccinations reduced depression by -1.3 (-1.4, -1.2) percentage points among non-Hispanic adults aged 25 or older. On the other hand, a five-fold increase in the odds of treatment by five among the youngest group would have further decreased depression by an additional 3 percentage points, while a five-fold increase in the treatment odds would have only decreased depression by an additional 0.9 percentage points. The effect curves also illustrate that the returns to increasing the odds of vaccination decrease most slowly among the youngest respondents relative to the other two groups. 

These results are consistent with our previous findings that the youngest respondents were most likely to benefit from COVID-19 vaccinations, while they were also least likely to be vaccinated: 18.0 percent of the youngest group were vaccinated, while 25.9 percent and 38.9 percent of the older Hispanic and non-Hispanic groups were vaccinated. Our incremental effects analysis is consistent with this finding, and further illustrates that the youngest respondents would obtain the largest relative returns from an intervention increasing everyone's odds of vaccination.

\begin{figure}[H]
\begin{center}
    \includegraphics[scale=0.4]{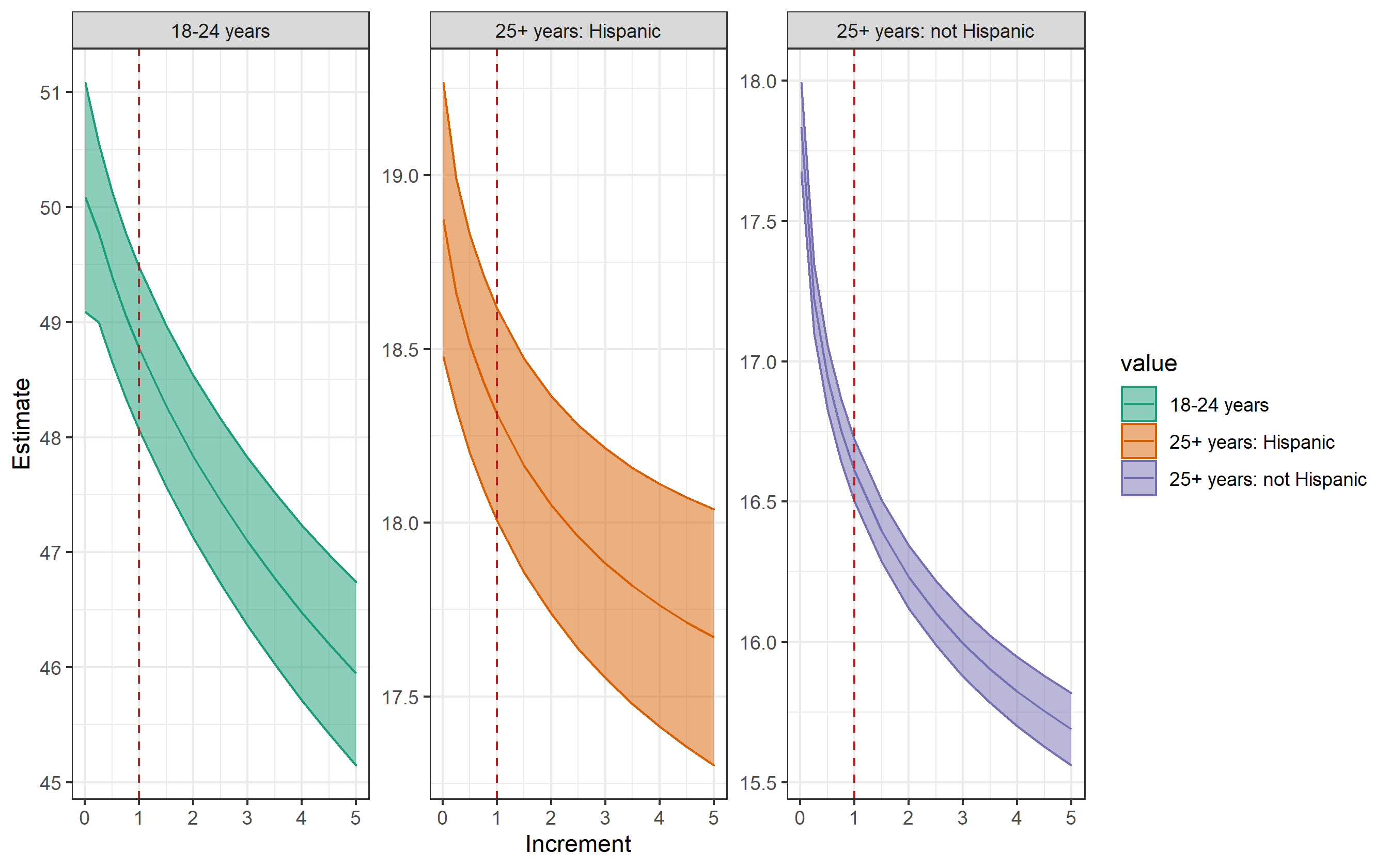}
    \caption{Heterogeneous incremental effects, February 2021}
    \label{fig:incplothet}
\end{center}
\end{figure}

The table below displays the incremental effects estimates using the GLM estimates.

  \begingroup%
  \renewcommand\normalsize{\footnotesize}% Specify your font modification
  \input{tables/outcomes-inc-table-glm-cc}%
  \endgroup%

The following tables display the incremental effect estimates using XGBoost models, stratified by subgroup. The GLM estimates are available on request.

  \begingroup%
  \renewcommand\normalsize{\footnotesize}% Specify your font modification
  \input{tables/age-depression-inc-table-xgb}%
  \endgroup%

  \begingroup%
  \renewcommand\normalsize{\footnotesize}% Specify your font modification
  \input{tables/age-isolation-inc-table-xgb}%
  \endgroup%

  \begingroup%
  \renewcommand\normalsize{\footnotesize}% Specify your font modification
  \input{tables/age-health-inc-table-xgb}%
  \endgroup%

%% file: tables/weight_table.tex
% latex table generated in R 4.0.2 by xtable 1.8-4 package
% Mon Feb 21 13:31:14 2022
\begin{table}[ht]
\centering
\caption{Propensity score models: weight distribution} 
\label{tab:weightdiagnostics}
\begin{tabular}{llrrrrrrrrr}
  \hline
Parameter & Model & 0\% & 1\% & 10\% & 25\% & 50\% & 75\% & 90\% & 99\% & 100\% \\ 
  \hline
Eta & GLM & 0.0005 & 0.0032 & 0.7122 & 0.9081 & 0.9587 & 0.9800 & 0.9893 & 0.9954 & 0.9988 \\ 
  Pi & GLM & 0.0043 & 0.0296 & 0.0827 & 0.1499 & 0.3162 & 0.5625 & 0.7608 & 0.9368 & 0.9972 \\ 
  Eta & XGboost & 0.0002 & 0.0002 & 0.6983 & 0.9278 & 0.9582 & 0.9701 & 0.9749 & 0.9782 & 0.9785 \\ 
  Pi & XGboost & 0.0678 & 0.0775 & 0.0970 & 0.1288 & 0.2500 & 0.5924 & 0.8297 & 0.9114 & 0.9297 \\ 
   \hline
\end{tabular}
\end{table}

%% file: tables/outcomes-hte-table-rd.tex
\begin{table}[!h]

\caption{Heterogeneous effects, GLM estimates, risk-differences (95\% CI) \label{tab:outcomeshtetablerd}}
\centering
\begin{tabular}[t]{lrlllr}
\toprule
Subgroup & N & Depressed & Isolated & Worried about health & N Full\\
\midrule
\addlinespace[0.3em]
\multicolumn{6}{l}{\textbf{Age}}\\
\hspace{1em}18-34 years & 114811 & -5.89 (-6.78, -5.01) & -5.66 (-6.54, -4.78) & -5.32 (-6.22, -4.42) & 117012\\
\hspace{1em}35-64 years & 392806 & -3.75 (-4.11, -3.38) & -3.25 (-3.63, -2.87) & -4.36 (-4.80, -3.93) & 412978\\
\hspace{1em}65 or more years & 238072 & -2.53 (-2.84, -2.22) & -2.13 (-2.51, -1.75) & -3.40 (-3.86, -2.95) & 261036\\
\hspace{1em}No response & 12905 & -5.75 (-7.16, -4.34) & -3.24 (-4.67, -1.82) & -7.28 (-8.92, -5.64) & 90289\\
\addlinespace[0.3em]
\multicolumn{6}{l}{\textbf{Ethnicity}}\\
\hspace{1em}Hispanic & 85915 & -2.81 (-3.73, -1.88) & -2.59 (-3.59, -1.58) & -5.67 (-6.87, -4.47) & 95550\\
\hspace{1em}Not Hispanic & 654822 & -3.78 (-4.05, -3.52) & -3.34 (-3.62, -3.06) & -4.03 (-4.34, -3.73) & 689480\\
\hspace{1em}No response & 17857 & -5.94 (-7.29, -4.58) & -3.57 (-5.15, -1.99) & -5.61 (-7.26, -3.95) & 96285\\
\addlinespace[0.3em]
\multicolumn{6}{l}{\textbf{Gender}}\\
\hspace{1em}Female & 491608 & -3.99 (-4.32, -3.66) & -3.38 (-3.72, -3.04) & -4.50 (-4.88, -4.12) & 522122\\
\hspace{1em}Not female & 249729 & -3.14 (-3.54, -2.74) & -3.10 (-3.55, -2.66) & -3.56 (-4.05, -3.06) & 264406\\
\hspace{1em}No response & 17257 & -4.50 (-6.35, -2.65) & -2.21 (-4.06, -0.35) & -7.41 (-9.65, -5.17) & 94787\\
\addlinespace[0.3em]
\multicolumn{6}{l}{\textbf{High-risk conditions}}\\
\hspace{1em}None & 263250 & -3.63 (-4.07, -3.20) & -3.56 (-4.02, -3.09) & -3.73 (-4.24, -3.21) & 297477\\
\hspace{1em}One & 219199 & -3.53 (-3.98, -3.08) & -3.29 (-3.77, -2.80) & -4.27 (-4.82, -3.73) & 247411\\
\hspace{1em}Two or more & 269353 & -4.01 (-4.43, -3.58) & -2.99 (-3.44, -2.54) & -4.79 (-5.29, -4.28) & 306998\\
\hspace{1em}Missing & 6792 & -2.13 (-5.93, 1.67) & -1.83 (-5.30, 1.64) & -3.09 (-6.62, 0.44) & 29429\\
\addlinespace[0.3em]
\multicolumn{6}{l}{\textbf{Race}}\\
\hspace{1em}White & 613930 & -3.91 (-4.18, -3.64) & -3.43 (-3.72, -3.14) & -4.06 (-4.37, -3.75) & 645144\\
\hspace{1em}Black & 19683 & -1.61 (-3.15, -0.07) & -1.73 (-3.17, -0.28) & -2.74 (-5.05, -0.42) & 20929\\
\hspace{1em}Asian or Pacific Islander & 45241 & -2.26 (-3.35, -1.16) & -2.35 (-3.41, -1.29) & -5.70 (-7.11, -4.29) & 50535\\
\hspace{1em}Other & 57888 & -3.43 (-4.64, -2.21) & -2.52 (-3.76, -1.28) & -4.54 (-6.02, -3.06) & 63659\\
\hspace{1em}No response & 21852 & -4.24 (-5.71, -2.77) & -3.85 (-5.53, -2.17) & -7.25 (-9.26, -5.24) & 101048\\
\bottomrule
\end{tabular}
\end{table}

%% file: tables/outcomes-hte-table-rr.tex
\begin{table}[!h]

\caption{Heterogeneous effects, GLM estimates, risk-ratios (95\% CI) \label{tab:outcomeshtetablerr}}
\centering
\begin{tabular}[t]{lrlllr}
\toprule
Subgroup & N & Depressed & Isolated & Worried about health & N Full\\
\midrule
\addlinespace[0.3em]
\multicolumn{6}{l}{\textbf{Age}}\\
\hspace{1em}18-34 years & 114811 & 0.85 (0.82, 0.87) & 0.83 (0.80, 0.85) & 0.85 (0.82, 0.87) & 117012\\
\hspace{1em}35-64 years & 392806 & 0.81 (0.79, 0.83) & 0.84 (0.82, 0.86) & 0.86 (0.84, 0.87) & 412978\\
\hspace{1em}65 or more years & 238072 & 0.73 (0.70, 0.75) & 0.85 (0.83, 0.88) & 0.84 (0.82, 0.86) & 261036\\
\hspace{1em}No response & 12905 & 0.70 (0.63, 0.77) & 0.82 (0.75, 0.90) & 0.77 (0.72, 0.82) & 90289\\
\addlinespace[0.3em]
\multicolumn{6}{l}{\textbf{Ethnicity}}\\
\hspace{1em}Hispanic & 85915 & 0.86 (0.82, 0.91) & 0.88 (0.83, 0.92) & 0.88 (0.86, 0.91) & 95550\\
\hspace{1em}Not Hispanic & 654822 & 0.80 (0.79, 0.81) & 0.83 (0.82, 0.85) & 0.84 (0.83, 0.86) & 689480\\
\hspace{1em}No response & 17857 & 0.68 (0.62, 0.75) & 0.81 (0.73, 0.89) & 0.81 (0.76, 0.87) & 96285\\
\addlinespace[0.3em]
\multicolumn{6}{l}{\textbf{Gender}}\\
\hspace{1em}Female & 491608 & 0.81 (0.80, 0.83) & 0.85 (0.83, 0.86) & 0.85 (0.84, 0.86) & 522122\\
\hspace{1em}Not female & 249729 & 0.79 (0.76, 0.81) & 0.82 (0.79, 0.84) & 0.85 (0.83, 0.87) & 264406\\
\hspace{1em}No response & 17257 & 0.81 (0.73, 0.88) & 0.90 (0.81, 0.98) & 0.77 (0.71, 0.84) & 94787\\
\addlinespace[0.3em]
\multicolumn{6}{l}{\textbf{High-risk conditions}}\\
\hspace{1em}None & 263250 & 0.80 (0.78, 0.82) & 0.81 (0.79, 0.84) & 0.86 (0.84, 0.88) & 297477\\
\hspace{1em}One & 219199 & 0.81 (0.79, 0.83) & 0.83 (0.81, 0.86) & 0.85 (0.83, 0.86) & 247411\\
\hspace{1em}Two or more & 269353 & 0.81 (0.79, 0.82) & 0.86 (0.84, 0.88) & 0.85 (0.83, 0.86) & 306998\\
\hspace{1em}Missing & 6792 & 0.91 (0.75, 1.07) & 0.92 (0.77, 1.07) & 0.89 (0.77, 1.01) & 29429\\
\addlinespace[0.3em]
\multicolumn{6}{l}{\textbf{Race}}\\
\hspace{1em}White & 613930 & 0.80 (0.78, 0.81) & 0.83 (0.82, 0.85) & 0.85 (0.83, 0.86) & 645144\\
\hspace{1em}Black & 19683 & 0.88 (0.78, 0.99) & 0.88 (0.77, 0.98) & 0.92 (0.86, 0.99) & 20929\\
\hspace{1em}Asian or Pacific Islander & 45241 & 0.84 (0.77, 0.92) & 0.83 (0.76, 0.90) & 0.82 (0.78, 0.86) & 50535\\
\hspace{1em}Other & 57888 & 0.86 (0.81, 0.91) & 0.89 (0.84, 0.95) & 0.89 (0.86, 0.93) & 63659\\
\hspace{1em}No response & 21852 & 0.77 (0.69, 0.85) & 0.80 (0.71, 0.88) & 0.81 (0.76, 0.86) & 101048\\
\bottomrule
\end{tabular}
\end{table}

%% file: tables/outcomes-hte-table-rd-xgb.tex
\begin{table}[!h]

\caption{Heterogeneous effects, XGBoost estimates, risk-differences (95\% CI) \label{tab:outcomeshtetablerdxgb}}
\centering
\begin{tabular}[t]{lrlllr}
\toprule
Subgroup & N & Depressed & Isolated & Worried about health & N Full\\
\midrule
\addlinespace[0.3em]
\multicolumn{6}{l}{\textbf{Age}}\\
\hspace{1em}18-34 years & 114811 & -6.11 (-6.82, -5.40) & -5.80 (-6.48, -5.11) & -4.91 (-5.61, -4.21) & 117012\\
\hspace{1em}35-64 years & 392806 & -3.68 (-3.99, -3.36) & -3.14 (-3.46, -2.81) & -4.28 (-4.64, -3.91) & 412978\\
\hspace{1em}65 or more years & 238072 & -2.65 (-2.90, -2.39) & -2.06 (-2.37, -1.74) & -3.47 (-3.83, -3.10) & 261036\\
\hspace{1em}No response & 12905 & -4.48 (-6.09, -2.86) & -3.02 (-4.65, -1.40) & -5.35 (-7.23, -3.48) & 90289\\
\addlinespace[0.3em]
\multicolumn{6}{l}{\textbf{Ethnicity}}\\
\hspace{1em}Hispanic & 85915 & -2.52 (-3.19, -1.84) & -2.41 (-3.11, -1.72) & -4.60 (-5.46, -3.74) & 95550\\
\hspace{1em}Not Hispanic & 654822 & -3.86 (-4.09, -3.64) & -3.30 (-3.54, -3.06) & -4.06 (-4.32, -3.80) & 689480\\
\hspace{1em}No response & 17857 & -4.87 (-6.22, -3.53) & -3.29 (-4.68, -1.90) & -4.71 (-6.30, -3.12) & 96285\\
\addlinespace[0.3em]
\multicolumn{6}{l}{\textbf{Gender}}\\
\hspace{1em}Female & 491608 & -3.97 (-4.25, -3.69) & -3.36 (-3.64, -3.07) & -4.42 (-4.74, -4.11) & 522122\\
\hspace{1em}Not female & 249729 & -3.24 (-3.57, -2.92) & -2.88 (-3.24, -2.52) & -3.46 (-3.87, -3.06) & 264406\\
\hspace{1em}No response & 17257 & -4.16 (-5.71, -2.61) & -3.16 (-4.71, -1.61) & -5.69 (-7.40, -3.98) & 94787\\
\addlinespace[0.3em]
\multicolumn{6}{l}{\textbf{High-risk conditions}}\\
\hspace{1em}None & 263250 & -3.47 (-3.83, -3.12) & -3.45 (-3.82, -3.07) & -3.56 (-3.97, -3.15) & 297477\\
\hspace{1em}One & 219199 & -3.53 (-3.92, -3.14) & -3.21 (-3.62, -2.80) & -4.21 (-4.67, -3.75) & 247411\\
\hspace{1em}Two or more & 269353 & -4.13 (-4.50, -3.77) & -2.93 (-3.31, -2.55) & -4.70 (-5.13, -4.27) & 306998\\
\hspace{1em}Missing & 6792 & -4.54 (-6.96, -2.12) & -3.49 (-5.89, -1.09) & -1.86 (-4.48, 0.76) & 29429\\
\addlinespace[0.3em]
\multicolumn{6}{l}{\textbf{Race}}\\
\hspace{1em}White & 613930 & -3.97 (-4.21, -3.74) & -3.28 (-3.53, -3.04) & -4.13 (-4.40, -3.86) & 645144\\
\hspace{1em}Black & 19683 & -1.72 (-2.97, -0.47) & -2.19 (-3.45, -0.92) & -2.45 (-4.13, -0.77) & 20929\\
\hspace{1em}Asian or Pacific Islander & 45241 & -2.49 (-3.32, -1.67) & -2.86 (-3.67, -2.05) & -5.70 (-6.78, -4.61) & 50535\\
\hspace{1em}Other & 57888 & -3.06 (-3.93, -2.20) & -2.82 (-3.71, -1.93) & -2.97 (-4.01, -1.93) & 63659\\
\hspace{1em}No response & 21852 & -3.17 (-4.40, -1.94) & -3.36 (-4.61, -2.11) & -5.72 (-7.23, -4.21) & 101048\\
\bottomrule
\end{tabular}
\end{table}

%% file: tables/outcomes-hte-table-rr-xgb.tex
\begin{table}[!h]

\caption{Heterogeneous effects, XGBoost estimates, risk-ratio (95\% CI) \label{tab:outcomeshtetablerrxgb}}
\centering
\begin{tabular}[t]{lrlllr}
\toprule
Subgroup & N & Depressed & Isolated & Worried about health & N Full\\
\midrule
\addlinespace[0.3em]
\multicolumn{6}{l}{\textbf{Age}}\\
\hspace{1em}18-34 years & 114811 & 0.84 (0.82, 0.86) & 0.82 (0.80, 0.84) & 0.86 (0.84, 0.88) & 117012\\
\hspace{1em}35-64 years & 392806 & 0.81 (0.80, 0.83) & 0.85 (0.83, 0.86) & 0.86 (0.85, 0.87) & 412978\\
\hspace{1em}65 or more years & 238072 & 0.72 (0.70, 0.74) & 0.86 (0.84, 0.88) & 0.84 (0.82, 0.85) & 261036\\
\hspace{1em}No response & 12905 & 0.76 (0.68, 0.84) & 0.83 (0.75, 0.92) & 0.82 (0.77, 0.88) & 90289\\
\addlinespace[0.3em]
\multicolumn{6}{l}{\textbf{Ethnicity}}\\
\hspace{1em}Hispanic & 85915 & 0.88 (0.85, 0.91) & 0.89 (0.85, 0.92) & 0.90 (0.89, 0.92) & 95550\\
\hspace{1em}Not Hispanic & 654822 & 0.80 (0.79, 0.81) & 0.84 (0.83, 0.85) & 0.84 (0.83, 0.85) & 689480\\
\hspace{1em}No response & 17857 & 0.74 (0.67, 0.81) & 0.82 (0.75, 0.89) & 0.84 (0.79, 0.89) & 96285\\
\addlinespace[0.3em]
\multicolumn{6}{l}{\textbf{Gender}}\\
\hspace{1em}Female & 491608 & 0.81 (0.80, 0.83) & 0.85 (0.83, 0.86) & 0.86 (0.85, 0.87) & 522122\\
\hspace{1em}Not female & 249729 & 0.78 (0.76, 0.80) & 0.83 (0.81, 0.85) & 0.86 (0.84, 0.87) & 264406\\
\hspace{1em}No response & 17257 & 0.82 (0.76, 0.88) & 0.86 (0.79, 0.92) & 0.82 (0.77, 0.87) & 94787\\
\addlinespace[0.3em]
\multicolumn{6}{l}{\textbf{High-risk conditions}}\\
\hspace{1em}None & 263250 & 0.81 (0.79, 0.83) & 0.82 (0.80, 0.84) & 0.86 (0.85, 0.88) & 297477\\
\hspace{1em}One & 219199 & 0.81 (0.79, 0.83) & 0.84 (0.82, 0.86) & 0.85 (0.83, 0.86) & 247411\\
\hspace{1em}Two or more & 269353 & 0.80 (0.78, 0.82) & 0.87 (0.85, 0.88) & 0.85 (0.84, 0.86) & 306998\\
\hspace{1em}Missing & 6792 & 0.81 (0.71, 0.90) & 0.85 (0.75, 0.95) & 0.93 (0.84, 1.02) & 29429\\
\addlinespace[0.3em]
\multicolumn{6}{l}{\textbf{Race}}\\
\hspace{1em}White & 613930 & 0.79 (0.78, 0.81) & 0.84 (0.83, 0.85) & 0.84 (0.83, 0.85) & 645144\\
\hspace{1em}Black & 19683 & 0.88 (0.79, 0.96) & 0.85 (0.76, 0.93) & 0.93 (0.88, 0.98) & 20929\\
\hspace{1em}Asian or Pacific Islander & 45241 & 0.83 (0.77, 0.88) & 0.79 (0.74, 0.85) & 0.82 (0.79, 0.85) & 50535\\
\hspace{1em}Other & 57888 & 0.87 (0.84, 0.91) & 0.88 (0.85, 0.92) & 0.93 (0.91, 0.95) & 63659\\
\hspace{1em}No response & 21852 & 0.83 (0.76, 0.89) & 0.82 (0.76, 0.89) & 0.85 (0.81, 0.89) & 101048\\
\bottomrule
\end{tabular}
\end{table}

%% file: tables/dep-tree-table-rd-xgb.tex
\begin{table}[!h]

\caption{All outcomes: data-driven subgroup heterogeneity, risk-differences (XGBoost estimates)}
\centering
\begin{threeparttable}
\begin{tabular}[t]{llllllllr}
\toprule
\multicolumn{2}{c}{ } & \multicolumn{2}{c}{Depression} & \multicolumn{2}{c}{Isolation} & \multicolumn{2}{c}{Health} & \multicolumn{1}{c}{ } \\
\cmidrule(l{3pt}r{3pt}){3-4} \cmidrule(l{3pt}r{3pt}){5-6} \cmidrule(l{3pt}r{3pt}){7-8}
Variable & Value & Estimate (95\% CI) & T-test & Estimate (95\% CI) & T-test & Estimate (95\% CI) & T-test & N\\
\midrule
\addlinespace[0.3em]
\multicolumn{9}{l}{\textbf{All data}}\\
\hspace{1em}- & 18-24 years & -12.21 (-13.68, -10.74) & Ref & -10.73 (-12.14, -9.31) & Ref & -6.93 (-8.28, -5.59) & Ref & 28073\\
\hspace{1em}- & 25+ years, Hispanic & -1.89 (-2.58, -1.21) & *** & -1.87 (-2.59, -1.16) & *** & -4.71 (-5.60, -3.81) & ** & 80254\\
\hspace{1em}- & 25+ years, not Hispanic & -3.56 (-3.79, -3.33) & *** & -3.03 (-3.27, -2.79) & *** & -3.93 (-4.19, -3.66) & *** & 632157\\
\hspace{1em}- & Other & -4.88 (-6.21, -3.56) & *** & -3.38 (-4.76, -2.01) & *** & -4.62 (-6.20, -3.04) & * & 18110\\
\addlinespace[0.3em]
\multicolumn{9}{l}{\textbf{18-24 years}}\\
\hspace{1em}Female & Not female & -13.76 (-16.07, -11.44) & Ref & -11.36 (-13.69, -9.04) & Ref & -6.29 (-8.50, -4.08) & Ref & 7956\\
\hspace{1em}Female & Female & -11.30 (-13.16, -9.44) &  & -10.35 (-12.13, -8.58) &  & -7.28 (-8.96, -5.61) &  & 19571\\
\hspace{1em}Female & Unknown & -22.25 (-33.68, -10.82) &  & -14.83 (-25.64, -4.02) &  & -3.80 (-14.78,  7.17) &  & 546\\
\hspace{1em}Lives with elderly & No & -10.94 (-12.67, -9.21) & Ref & -10.09 (-11.78, -8.41) & Ref & -5.22 (-6.81, -3.63) & Ref & 20347\\
\hspace{1em}Lives with elderly & Missing & -12.55 (-15.95, -9.15) &  & -10.21 (-13.46, -6.95) &  & -10.85 (-13.91, -7.80) & ** & 5184\\
\hspace{1em}Lives with elderly & Yes & -21.62 (-26.45, -16.79) & *** & -16.85 (-21.31, -12.40) & ** & -12.65 (-17.07, -8.24) & ** & 2542\\
\hspace{1em}White & White & -12.10 (-13.80, -10.39) & Ref & -11.55 (-13.18, -9.92) & Ref & -7.90 (-9.43, -6.37) & Ref & 21515\\
\hspace{1em}White & Not white & -12.24 (-15.28, -9.20) &  & -7.72 (-10.75, -4.69) & * & -3.70 (-6.66, -0.73) & * & 6019\\
\hspace{1em}White & Unknown & -16.22 (-24.48, -7.95) &  & -11.55 (-20.10, -3.00) &  & -4.68 (-12.59,  3.22) &  & 539\\
\addlinespace[0.3em]
\multicolumn{9}{l}{\textbf{25+ years, Hispanic}}\\
\hspace{1em}Age & 25-44 years & -1.45 (-2.64, -0.26) & Ref & -1.38 (-2.56, -0.21) & Ref & -4.48 (-5.85, -3.10) & Ref & 33726\\
\hspace{1em}Age & 45 years or older & -2.22 (-3.03, -1.40) &  & -2.23 (-3.12, -1.34) &  & -4.87 (-6.04, -3.70) &  & 46528\\
\hspace{1em}Language & English & -2.52 (-3.39, -1.65) & Ref & -2.80 (-3.63, -1.97) & Ref & -4.96 (-6.02, -3.90) & Ref & 57236\\
\hspace{1em}Language & Spanish & -0.33 (-1.36,  0.69) & ** & 0.42 (-0.96,  1.81) & *** & -4.07 (-5.72, -2.43) &  & 23018\\
\bottomrule
\end{tabular}
\begin{tablenotes}
\item T-tests reflect test of difference in means between the category specified in row and the reference category (Ref) within each Subset and Variable
\item *** indicates p < 0.001; ** indicates p < 0.01; * indicates p < 0.05
\item Subgroups chosen using data from March 1-14th using specified outcome
\end{tablenotes}
\end{threeparttable}
\end{table}

%% file: tables/dep-tree-table-rr-xgb.tex
\begin{table}[!h]

\caption{All outcomes: data-driven subgroup heterogeneity, risk-ratios (XGBoost estimates)}
\centering
\begin{threeparttable}
\begin{tabular}[t]{llllllllr}
\toprule
\multicolumn{2}{c}{ } & \multicolumn{2}{c}{Depression} & \multicolumn{2}{c}{Isolation} & \multicolumn{2}{c}{Health} & \multicolumn{1}{c}{ } \\
\cmidrule(l{3pt}r{3pt}){3-4} \cmidrule(l{3pt}r{3pt}){5-6} \cmidrule(l{3pt}r{3pt}){7-8}
Variable & Value & Estimate (95\% CI) & T-test & Estimate (95\% CI) & T-test & Estimate (95\% CI) & T-test & N\\
\midrule
\addlinespace[0.3em]
\multicolumn{9}{l}{\textbf{All data}}\\
\hspace{1em}- & 18-24 years & 0.76 (0.73, 0.78) & Ref & 0.72 (0.69, 0.76) & Ref & 0.81 (0.77, 0.84) & Ref & 28073\\
\hspace{1em}- & 25+ years, Hispanic & 0.90 (0.86, 0.94) & *** & 0.91 (0.87, 0.94) & *** & 0.90 (0.88, 0.92) & *** & 80254\\
\hspace{1em}- & 25+ years, not Hispanic & 0.80 (0.79, 0.81) & ** & 0.85 (0.83, 0.86) & *** & 0.85 (0.84, 0.86) & * & 632157\\
\hspace{1em}- & Other & 0.73 (0.67, 0.80) &  & 0.82 (0.75, 0.89) & * & 0.84 (0.79, 0.90) &  & 18110\\
\addlinespace[0.3em]
\multicolumn{9}{l}{\textbf{18-24 years}}\\
\hspace{1em}Female & Not female & 0.68 (0.63, 0.73) & Ref & 0.69 (0.63, 0.75) & Ref & 0.79 (0.72, 0.86) & Ref & 7956\\
\hspace{1em}Female & Female & 0.79 (0.75, 0.82) & *** & 0.74 (0.70, 0.78) &  & 0.81 (0.76, 0.85) &  & 19571\\
\hspace{1em}Female & Unknown & 0.66 (0.50, 0.82) &  & 0.70 (0.50, 0.90) &  & 0.91 (0.64, 1.17) &  & 546\\
\hspace{1em}Lives with elderly & No & 0.78 (0.75, 0.81) & Ref & 0.74 (0.70, 0.78) & Ref & 0.84 (0.80, 0.89) & Ref & 20347\\
\hspace{1em}Lives with elderly & Missing & 0.74 (0.68, 0.81) &  & 0.73 (0.65, 0.81) &  & 0.72 (0.65, 0.80) & ** & 5184\\
\hspace{1em}Lives with elderly & Yes & 0.60 (0.53, 0.68) & *** & 0.61 (0.52, 0.70) & * & 0.73 (0.65, 0.82) & * & 2542\\
\hspace{1em}White & White & 0.77 (0.74, 0.80) & Ref & 0.71 (0.67, 0.75) & Ref & 0.77 (0.73, 0.81) & Ref & 21515\\
\hspace{1em}White & Not white & 0.72 (0.66, 0.79) &  & 0.79 (0.71, 0.87) &  & 0.91 (0.83, 0.98) & ** & 6019\\
\hspace{1em}White & Unknown & 0.63 (0.48, 0.79) &  & 0.67 (0.46, 0.89) &  & 0.90 (0.75, 1.06) &  & 539\\
\addlinespace[0.3em]
\multicolumn{9}{l}{\textbf{25+ years, Hispanic}}\\
\hspace{1em}Age & 25-44 years & 0.94 (0.89, 0.99) & Ref & 0.94 (0.89, 0.99) & Ref & 0.91 (0.88, 0.94) & Ref & 33726\\
\hspace{1em}Age & 45 years or older & 0.85 (0.80, 0.90) & * & 0.88 (0.83, 0.92) &  & 0.90 (0.87, 0.92) &  & 46528\\
\hspace{1em}Language & English & 0.88 (0.85, 0.92) & Ref & 0.86 (0.81, 0.90) & Ref & 0.89 (0.86, 0.91) & Ref & 57236\\
\hspace{1em}Language & Spanish & 0.97 (0.89, 1.06) &  & 1.02 (0.96, 1.08) & *** & 0.93 (0.90, 0.96) & * & 23018\\
\bottomrule
\end{tabular}
\begin{tablenotes}
\item T-tests reflect test of difference in means between the category specified in row and the reference category (Ref) within each Subset and Variable
\item *** indicates p < 0.001; ** indicates p < 0.01; * indicates p < 0.05
\item Subgroups chosen using data from March 1-14th using specified outcome
\end{tablenotes}
\end{threeparttable}
\end{table}

%% file: tables/health-tree-table-xgb.tex
\begin{table}[!h]

\caption{Worries about health: data-driven subgroup heterogeneity (XGBoost estimates)}
\centering
\begin{threeparttable}
\begin{tabular}[t]{lllllr}
\toprule
\multicolumn{1}{c}{ } & \multicolumn{2}{c}{Risk-difference} & \multicolumn{2}{c}{Risk-ratio} \\
\cmidrule(l{3pt}r{3pt}){2-3} \cmidrule(l{3pt}r{3pt}){4-5}
Strata & Estimate (95\% CI) & T-test & Estimate (95\% CI) & T-test & N\\
\midrule
\addlinespace[0.3em]
\multicolumn{6}{l}{\textbf{Female}}\\
\hspace{1em}Other occupation & -4.07 (-4.58, -3.55) & Ref & 0.87 (0.85, 0.89) & Ref & 192529\\
\hspace{1em}Healthcare & -5.02 (-6.03, -4.01) &  & 0.82 (0.79, 0.85) & ** & 56546\\
\hspace{1em}Not working / Unknown & -4.57 (-5.01, -4.13) &  & 0.85 (0.84, 0.87) &  & 242533\\
\addlinespace[0.3em]
\multicolumn{6}{l}{\textbf{Not female}}\\
\hspace{1em}Other occupation & -2.95 (-3.53, -2.38) & Ref & 0.87 (0.85, 0.90) & Ref & 120789\\
\hspace{1em}Healthcare & -4.64 (-6.57, -2.70) &  & 0.78 (0.71, 0.86) & * & 12294\\
\hspace{1em}Not working / Unknown & -3.87 (-4.46, -3.28) & * & 0.85 (0.82, 0.87) &  & 116646\\
\bottomrule
\end{tabular}
\begin{tablenotes}
\item T-tests reflect test of difference in means between the category specified in row and the reference category (Ref) within each Subset and Variable
\item *** indicates p < 0.001; ** indicates p < 0.01; * indicates p < 0.05
\item Subgroups chosen using data from March 1-14th using specified outcome
\end{tablenotes}
\end{threeparttable}
\end{table}

%% file: tables/isolation-tree-table-xgb.tex
\begin{table}[!h]

\caption{Isolation: data-driven subgroup heterogeneity (XGBoost estimates)}
\centering
\begin{threeparttable}
\begin{tabular}[t]{lllllr}
\toprule
\multicolumn{1}{c}{ } & \multicolumn{2}{c}{Risk-difference} & \multicolumn{2}{c}{Risk-ratio} \\
\cmidrule(l{3pt}r{3pt}){2-3} \cmidrule(l{3pt}r{3pt}){4-5}
Strata & Estimate (95\% CI) & T-test & Estimate (95\% CI) & T-test & N\\
\midrule
\addlinespace[0.3em]
\multicolumn{6}{l}{\textbf{18-24 years: Not White}}\\
\hspace{1em}Not female & -5.63 (-10.41, -0.85) & Ref & 0.83 ( 0.68,  0.97) & Ref & 1941\\
\hspace{1em}Female & -9.81 (-13.67, -5.95) &  & 0.74 ( 0.65,  0.84) &  & 3906\\
\hspace{1em}Unknown & 16.17 (-8.96, 41.31) &  & 1.39 ( 0.74,  2.05) &  & 172\\
\addlinespace[0.3em]
\multicolumn{6}{l}{\textbf{18-24 years: White}}\\
\hspace{1em}Not female & -13.27 (-15.98, -10.56) & Ref & 0.65 ( 0.59,  0.72) & Ref & 5842\\
\hspace{1em}Female & -10.51 (-12.54, -8.49) &  & 0.74 ( 0.69,  0.78) & * & 15322\\
\hspace{1em}Unknown & -28.19 (-39.20, -17.19) & ** & 0.47 ( 0.29,  0.65) &  & 351\\
\addlinespace[0.3em]
\multicolumn{6}{l}{\textbf{25+ years: Not White}}\\
\hspace{1em}Not female & -1.52 (-2.37, -0.66) & Ref & 0.90 ( 0.85,  0.96) & Ref & 43686\\
\hspace{1em}Female & -3.22 (-3.96, -2.48) & ** & 0.83 ( 0.79,  0.87) & * & 70434\\
\hspace{1em}Unknown & 1.90 (-2.90,  6.71) &  & 1.08 ( 0.88,  1.28) &  & 2537\\
\addlinespace[0.3em]
\multicolumn{6}{l}{\textbf{25+ years: White}}\\
\hspace{1em}Not female & -2.86 (-3.27, -2.46) & Ref & 0.83 ( 0.81,  0.85) & Ref & 193701\\
\hspace{1em}Female & -3.02 (-3.33, -2.70) &  & 0.86 ( 0.85,  0.87) & * & 394616\\
\hspace{1em}Unknown & -6.03 (-9.36, -2.70) &  & 0.77 ( 0.65,  0.89) &  & 3839\\
\bottomrule
\end{tabular}
\begin{tablenotes}
\item T-tests reflect test of difference in means between the category specified in row and the reference category (Ref) within each Subset and Variable
\item *** indicates p < 0.001; ** indicates p < 0.01; * indicates p < 0.05
\item Subgroups chosen using data from March 1-14th using specified outcome
\end{tablenotes}
\end{threeparttable}
\end{table}

%% file: tables/mediation-table-one.tex
\begin{table}[!h]

\caption{Mediation analysis: GLM estimates (95\% CI) \label{tab:medtab1}}
\centering
\begin{threeparttable}
\begin{tabular}[t]{lllll}
\toprule
Effect & Ref = 0 & Proportion (Ref = 0) & Ref = 1 & Proportion (Ref = 1)\\
\midrule
IDE & -1.93 (-2.16, -1.70) & 51.76 (48.33, 55.18) & -2.06 (-2.31, -1.82) & 55.37 (51.54, 59.20)\\
IIE & -1.80 (-1.91, -1.68) & 48.24 (44.82, 51.67) & -1.66 (-1.72, -1.61) & 44.63 (40.80, 48.46)\\
IIE - Isolated & -1.45 (-1.56, -1.35) & 39.07 (36.07, 42.08) & -1.39 (-1.50, -1.28) & 37.42 (34.18, 40.66)\\
IIE - Health & -0.31 (-0.34, -0.28) & 8.28 (7.29, 9.28) & -0.24 (-0.29, -0.20) & 6.57 (5.22, 7.92)\\
IIE - Covariant & -0.03 (-0.06, -0.01) & 0.89 (0.25, 1.53) & -0.02 (-0.15, 0.10) & 0.65 (-0.79, 2.08)\\
\bottomrule
\end{tabular}
\begin{tablenotes}
\item IDE represents the interventional direct effect; IIE represents the total interventional indirect effect; IIE - X represents interventional indirect effect via the indicated channel X; Ref represents the reference category for the indirect effects
\end{tablenotes}
\end{threeparttable}
\end{table}

%% file: tables/mediation-table-six.tex
\begin{table}[!h]

\caption{Mediation analysis: XGBoost estimates (95\% CI) \label{tab:medtab6}}
\centering
\begin{threeparttable}
\begin{tabular}[t]{lllll}
\toprule
Effect & Ref = 0 & Proportion (Ref = 0) & Ref = 1 & Proportion (Ref = 1)\\
\midrule
IDE & -1.85 (-2.04, -1.65) & 49.43 (46.37, 52.50) & -2.10 (-2.32, -1.88) & 56.16 (52.63, 59.69)\\
IIE & -1.89 (-1.99, -1.78) & 50.57 (47.50, 53.63) & -1.64 (-1.69, -1.59) & 43.84 (40.31, 47.37)\\
IIE - Isolated & -1.63 (-1.73, -1.54) & 43.71 (41.00, 46.42) & -1.64 (-1.74, -1.54) & 43.97 (40.80, 47.13)\\
IIE - Health & -0.33 (-0.37, -0.30) & 8.93 (7.92, 9.94) & -0.28 (-0.32, -0.24) & 7.47 (6.24, 8.70)\\
IIE - Covariant & 0.08 (0.05, 0.10) & -2.07 (-2.78, -1.37) & 0.28 (0.17, 0.40) & -7.60 (-9.06, -6.13)\\
\bottomrule
\end{tabular}
\begin{tablenotes}
\item IDE represents the interventional direct effect; IIE represents the total interventional indirect effect; IIE - X represents interventional indirect effect via the indicated channel X; Ref represents the reference category for the indirect effects
\end{tablenotes}
\end{threeparttable}
\end{table}

%% file: tables/mediation-hte-table-0.tex
\begin{table}[!h]

\caption{Heterogeneous mediated effects (Ref = 0): GLM estimates \label{tab:outcomeshtetable2}}
\centering
\begin{threeparttable}
\begin{tabular}[t]{llllll}
\toprule
Label & IDE & IIE & IIE - Isolated & IIE - Health & IIE - Covariant\\
\midrule
\addlinespace[0.3em]
\multicolumn{6}{l}{\textbf{Age}}\\
\hspace{1em}18-34 years & -2.39 (-3.19, -1.60) & -3.50 (-3.91, -3.09) & -3.01 (-3.39, -2.62) & -0.47 (-0.57, -0.37) & -0.02 (-0.10, 0.05)\\
\hspace{1em}35-64 years & -1.87 (-2.20, -1.53) & -1.88 (-2.05, -1.71) & -1.51 (-1.67, -1.36) & -0.35 (-0.40, -0.30) & -0.02 (-0.05, 0.01)\\
\hspace{1em}65 or more years & -1.71 (-2.00, -1.43) & -0.82 (-0.95, -0.68) & -0.60 (-0.72, -0.48) & -0.16 (-0.21, -0.11) & -0.06 (-0.10, -0.02)\\
\hspace{1em}No response & -3.61 (-4.88, -2.35) & -2.13 (-2.85, -1.42) & -1.67 (-2.33, -1.00) & -0.34 (-0.60, -0.08) & -0.13 (-0.30, 0.05)\\
\addlinespace[0.3em]
\multicolumn{6}{l}{\textbf{Ethnicity}}\\
\hspace{1em}Hispanic & -1.01 (-1.85, -0.17) & -1.80 (-2.21, -1.39) & -1.43 (-1.82, -1.04) & -0.33 (-0.43, -0.22) & -0.04 (-0.13, 0.04)\\
\hspace{1em}No response & -3.90 (-5.04, -2.77) & -2.03 (-2.76, -1.31) & -1.60 (-2.28, -0.92) & -0.24 (-0.45, -0.02) & -0.20 (-0.38, -0.01)\\
\hspace{1em}Not Hispanic & -1.99 (-2.23, -1.75) & -1.79 (-1.91, -1.67) & -1.45 (-1.57, -1.34) & -0.31 (-0.34, -0.27) & -0.03 (-0.05, 0.00)\\
\addlinespace[0.3em]
\multicolumn{6}{l}{\textbf{Gender}}\\
\hspace{1em}Female & -2.13 (-2.42, -1.83) & -1.87 (-2.01, -1.72) & -1.48 (-1.62, -1.35) & -0.35 (-0.39, -0.30) & -0.04 (-0.07, -0.01)\\
\hspace{1em}No response & -2.94 (-4.72, -1.15) & -1.57 (-2.53, -0.61) & -1.06 (-1.97, -0.15) & -0.59 (-0.91, -0.28) & 0.09 (-0.15, 0.33)\\
\hspace{1em}Not female & -1.47 (-1.83, -1.11) & -1.68 (-1.87, -1.48) & -1.43 (-1.61, -1.25) & -0.22 (-0.26, -0.17) & -0.03 (-0.07, 0.01)\\
\addlinespace[0.3em]
\multicolumn{6}{l}{\textbf{High-risk conditions}}\\
\hspace{1em}Missing & -0.14 (-3.42, 3.13) & -1.99 (-3.37, -0.60) & -1.64 (-2.95, -0.34) & -0.05 (-0.42, 0.32) & -0.29 (-0.66, 0.07)\\
\hspace{1em}None & -1.67 (-2.05, -1.28) & -1.97 (-2.17, -1.77) & -1.70 (-1.89, -1.52) & -0.27 (-0.32, -0.22) & 0.00 (-0.04, 0.04)\\
\hspace{1em}One & -1.82 (-2.23, -1.41) & -1.71 (-1.92, -1.50) & -1.38 (-1.57, -1.19) & -0.29 (-0.35, -0.23) & -0.04 (-0.08, 0.01)\\
\hspace{1em}Two or more & -2.31 (-2.71, -1.92) & -1.69 (-1.89, -1.50) & -1.27 (-1.45, -1.09) & -0.37 (-0.43, -0.31) & -0.06 (-0.10, -0.02)\\
\addlinespace[0.3em]
\multicolumn{6}{l}{\textbf{Race}}\\
\hspace{1em}Asian or Pacific Islander & -0.67 (-1.66, 0.31) & -1.58 (-2.07, -1.10) & -1.18 (-1.63, -0.73) & -0.40 (-0.54, -0.27) & 0.00 (-0.11, 0.12)\\
\hspace{1em}Black & -0.27 (-1.75, 1.21) & -1.34 (-2.16, -0.52) & -1.13 (-1.93, -0.33) & -0.11 (-0.27, 0.05) & -0.09 (-0.24, 0.05)\\
\hspace{1em}No response & -2.03 (-3.35, -0.71) & -2.21 (-2.94, -1.48) & -1.80 (-2.47, -1.14) & -0.36 (-0.56, -0.16) & -0.05 (-0.21, 0.12)\\
\hspace{1em}Other & -1.73 (-2.86, -0.61) & -1.69 (-2.23, -1.16) & -1.40 (-1.89, -0.90) & -0.26 (-0.42, -0.10) & -0.03 (-0.16, 0.09)\\
\hspace{1em}White & -2.09 (-2.33, -1.84) & -1.82 (-1.95, -1.70) & -1.48 (-1.59, -1.36) & -0.31 (-0.35, -0.28) & -0.03 (-0.06, -0.01)\\
\bottomrule
\end{tabular}
\begin{tablenotes}
\item IDE represents the interventional direct effect; IIE represents the total interventional indirect effect; IIE - X represents interventional indirect effect via the indicated channel X; Ref represents the reference category for the indirect effects
\end{tablenotes}
\end{threeparttable}
\end{table}

%% file: tables/mediation-hte-table-1.tex
\begin{table}[!h]

\caption{Heterogeneous mediated effects (Ref = 1): GLM estimates \label{tab:outcomeshtetable3}}
\centering
\begin{threeparttable}
\begin{tabular}[t]{llllll}
\toprule
Label & IDE & IIE & IIE - Isolated & IIE - Health & IIE - Covariant\\
\midrule
\addlinespace[0.3em]
\multicolumn{6}{l}{\textbf{Age}}\\
\hspace{1em}18-34 years & -2.76 (-3.64, -1.88) & -3.14 (-3.42, -2.85) & -2.90 (-3.34, -2.47) & -0.33 (-0.52, -0.13) & 0.09 (-0.42, 0.61)\\
\hspace{1em}35-64 years & -1.98 (-2.33, -1.63) & -1.76 (-1.82, -1.71) & -1.46 (-1.62, -1.30) & -0.28 (-0.35, -0.21) & -0.03 (-0.20, 0.15)\\
\hspace{1em}65 or more years & -1.77 (-2.07, -1.46) & -0.76 (-0.80, -0.73) & -0.55 (-0.66, -0.44) & -0.14 (-0.18, -0.11) & -0.07 (-0.20, 0.06)\\
\hspace{1em}No response & -3.75 (-5.19, -2.31) & -2.00 (-2.26, -1.74) & -1.46 (-2.09, -0.84) & -0.40 (-0.72, -0.08) & -0.14 (-0.89, 0.62)\\
\addlinespace[0.3em]
\multicolumn{6}{l}{\textbf{Ethnicity}}\\
\hspace{1em}Hispanic & -1.17 (-2.06, -0.28) & -1.64 (-1.80, -1.48) & -1.50 (-1.91, -1.09) & -0.15 (-0.35, 0.05) & 0.01 (-0.45, 0.47)\\
\hspace{1em}No response & -3.91 (-5.19, -2.63) & -2.03 (-2.28, -1.78) & -1.48 (-2.13, -0.83) & -0.28 (-0.54, -0.02) & -0.27 (-1.05, 0.51)\\
\hspace{1em}Not Hispanic & -2.13 (-2.39, -1.87) & -1.66 (-1.71, -1.60) & -1.38 (-1.49, -1.26) & -0.26 (-0.30, -0.21) & -0.02 (-0.16, 0.11)\\
\addlinespace[0.3em]
\multicolumn{6}{l}{\textbf{Gender}}\\
\hspace{1em}Female & -2.18 (-2.49, -1.87) & -1.81 (-1.86, -1.77) & -1.45 (-1.59, -1.31) & -0.33 (-0.38, -0.27) & -0.03 (-0.19, 0.12)\\
\hspace{1em}No response & -3.22 (-4.98, -1.46) & -1.28 (-1.76, -0.79) & -1.17 (-2.14, -0.20) & 0.22 (-0.48, 0.93) & -0.34 (-1.42, 0.75)\\
\hspace{1em}Not female & -1.75 (-2.16, -1.34) & -1.39 (-1.52, -1.26) & -1.29 (-1.49, -1.10) & -0.12 (-0.20, -0.04) & 0.02 (-0.22, 0.25)\\
\addlinespace[0.3em]
\multicolumn{6}{l}{\textbf{High-risk conditions}}\\
\hspace{1em}Missing & -0.54 (-4.14, 3.07) & -1.59 (-2.24, -0.94) & -1.71 (-3.04, -0.38) & -0.18 (-0.85, 0.49) & 0.30 (-1.10, 1.69)\\
\hspace{1em}None & -1.70 (-2.13, -1.26) & -1.94 (-2.07, -1.81) & -1.69 (-1.89, -1.50) & -0.22 (-0.31, -0.12) & -0.03 (-0.27, 0.21)\\
\hspace{1em}One & -1.97 (-2.41, -1.53) & -1.56 (-1.62, -1.49) & -1.26 (-1.46, -1.07) & -0.27 (-0.34, -0.20) & -0.03 (-0.24, 0.19)\\
\hspace{1em}Two or more & -2.53 (-2.94, -2.12) & -1.48 (-1.54, -1.42) & -1.20 (-1.38, -1.01) & -0.25 (-0.33, -0.18) & -0.03 (-0.23, 0.18)\\
\addlinespace[0.3em]
\multicolumn{6}{l}{\textbf{Race}}\\
\hspace{1em}Asian or Pacific Islander & -1.13 (-2.42, 0.17) & -1.13 (-1.77, -0.49) & -0.98 (-1.55, -0.42) & -0.43 (-0.75, -0.10) & 0.28 (-0.56, 1.13)\\
\hspace{1em}Black & -0.58 (-2.00, 0.83) & -1.03 (-1.31, -0.74) & -1.06 (-1.78, -0.34) & 0.19 (-0.29, 0.68) & -0.16 (-0.97, 0.65)\\
\hspace{1em}No response & -2.02 (-3.49, -0.55) & -2.22 (-2.54, -1.90) & -1.72 (-2.40, -1.05) & -0.34 (-0.68, -0.01) & -0.16 (-0.95, 0.64)\\
\hspace{1em}Other & -1.74 (-2.89, -0.60) & -1.68 (-1.90, -1.47) & -1.54 (-2.08, -1.00) & -0.07 (-0.29, 0.15) & -0.07 (-0.66, 0.51)\\
\hspace{1em}White & -2.21 (-2.47, -1.95) & -1.70 (-1.74, -1.66) & -1.41 (-1.52, -1.29) & -0.26 (-0.30, -0.21) & -0.03 (-0.16, 0.09)\\
\bottomrule
\end{tabular}
\begin{tablenotes}
\item IDE represents the interventional direct effect; IIE represents the total interventional indirect effect; IIE - X represents interventional indirect effect via the indicated channel X; Ref represents the reference category for the indirect effects
\end{tablenotes}
\end{threeparttable}
\end{table}

%% file: tables/mediation-hte-prop-table-0.tex
\begin{table}[!h]

\caption{Heterogeneous proportion mediated (Ref = 0): GLM estimates \label{tab:outcomeshtetable4}}
\centering
\begin{threeparttable}
\begin{tabular}[t]{llllll}
\toprule
Label & IDE & IIE & IIE - Isolated & IIE - Health & IIE - Covariant\\
\midrule
\addlinespace[0.3em]
\multicolumn{6}{l}{\textbf{Age}}\\
\hspace{1em}18-34 years & 40.61 (32.05, 49.18) & 59.39 (50.82, 67.95) & 51.04 (43.43, 58.66) & 7.96 (5.98, 9.94) & 0.38 (-0.85, 1.62)\\
\hspace{1em}35-64 years & 49.78 (44.71, 54.85) & 50.22 (45.15, 55.29) & 40.40 (35.96, 44.85) & 9.37 (7.92, 10.82) & 0.45 (-0.38, 1.27)\\
\hspace{1em}65 or more years & 67.74 (62.59, 72.89) & 32.26 (27.11, 37.41) & 23.62 (19.16, 28.09) & 6.28 (4.20, 8.36) & 2.36 (0.68, 4.04)\\
\hspace{1em}No response & 62.86 (51.18, 74.54) & 37.14 (25.46, 48.82) & 28.98 (18.37, 39.58) & 5.95 (1.22, 10.68) & 2.22 (-0.89, 5.33)\\
\addlinespace[0.3em]
\multicolumn{6}{l}{\textbf{Ethnicity}}\\
\hspace{1em}Hispanic & 35.89 (15.95, 55.84) & 64.11 (44.16, 84.05) & 50.96 (33.90, 68.02) & 11.64 (6.85, 16.43) & 1.51 (-1.60, 4.61)\\
\hspace{1em}No response & 65.74 (55.51, 75.97) & 34.26 (24.03, 44.49) & 26.96 (17.38, 36.54) & 4.00 (0.29, 7.71) & 3.30 (0.24, 6.36)\\
\hspace{1em}Not Hispanic & 52.70 (49.21, 56.19) & 47.30 (43.81, 50.79) & 38.44 (35.37, 41.50) & 8.14 (7.10, 9.18) & 0.72 (0.07, 1.37)\\
\addlinespace[0.3em]
\multicolumn{6}{l}{\textbf{Gender}}\\
\hspace{1em}Female & 53.26 (49.24, 57.28) & 46.74 (42.72, 50.76) & 37.10 (33.61, 40.58) & 8.65 (7.40, 9.90) & 0.99 (0.28, 1.70)\\
\hspace{1em}No response & 65.18 (43.70, 86.66) & 34.82 (13.34, 56.30) & 23.58 (3.56, 43.60) & 13.19 (5.23, 21.14) & -1.94 (-7.40, 3.51)\\
\hspace{1em}Not female & 46.67 (39.86, 53.49) & 53.33 (46.51, 60.14) & 45.55 (39.41, 51.69) & 6.88 (5.30, 8.45) & 0.90 (-0.46, 2.27)\\
\addlinespace[0.3em]
\multicolumn{6}{l}{\textbf{High-risk conditions}}\\
\hspace{1em}Missing & 6.65 (-136.05, 149.35) & 93.35 (-49.35, 236.05) & 77.21 (-40.18, 194.60) & 2.32 (-14.68, 19.31) & 13.82 (-17.02, 44.66)\\
\hspace{1em}None & 45.85 (39.57, 52.12) & 54.15 (47.88, 60.43) & 46.85 (41.18, 52.51) & 7.36 (5.79, 8.92) & -0.05 (-1.09, 0.99)\\
\hspace{1em}One & 51.57 (45.03, 58.11) & 48.43 (41.89, 54.97) & 39.08 (33.38, 44.77) & 8.28 (6.44, 10.13) & 1.07 (-0.25, 2.39)\\
\hspace{1em}Two or more & 57.73 (52.55, 62.90) & 42.27 (37.10, 47.45) & 31.67 (27.16, 36.19) & 9.18 (7.44, 10.93) & 1.41 (0.41, 2.42)\\
\addlinespace[0.3em]
\multicolumn{6}{l}{\textbf{Race}}\\
\hspace{1em}Asian or Pacific Islander & 29.80 (-1.52, 61.12) & 70.20 (38.88, 101.52) & 52.47 (27.85, 77.10) & 17.89 (8.14, 27.65) & -0.17 (-5.10, 4.76)\\
\hspace{1em}Black & 16.96 (-61.45, 95.37) & 83.04 (4.63, 161.45) & 70.13 (-0.32, 140.58) & 7.06 (-4.04, 18.15) & 5.85 (-4.39, 16.10)\\
\hspace{1em}No response & 47.85 (29.14, 66.57) & 52.15 (33.43, 70.86) & 42.57 (26.50, 58.64) & 8.46 (3.06, 13.85) & 1.12 (-2.72, 4.97)\\
\hspace{1em}Other & 50.59 (32.15, 69.02) & 49.41 (30.98, 67.85) & 40.83 (24.87, 56.79) & 7.58 (1.96, 13.21) & 1.00 (-2.56, 4.55)\\
\hspace{1em}White & 53.40 (50.00, 56.79) & 46.60 (43.21, 50.00) & 37.81 (34.81, 40.82) & 7.94 (6.95, 8.93) & 0.85 (0.22, 1.48)\\
\bottomrule
\end{tabular}
\begin{tablenotes}
\item IDE represents the interventional direct effect; IIE represents the total interventional indirect effect; IIE - X represents interventional indirect effect via the indicated channel X; Ref represents the reference category for the indirect effects
\end{tablenotes}
\end{threeparttable}
\end{table}

%% file: tables/mediation-hte-prop-table-1.tex
\begin{table}[!h]

\caption{Heterogeneous proportion mediated (Ref = 1): GLM estimates}
\centering
\begin{threeparttable}
\begin{tabular}[t]{llllll}
\toprule
Label & IDE & IIE & IIE - Isolated & IIE - Health & IIE - Covariant\\
\midrule
\addlinespace[0.3em]
\multicolumn{6}{l}{\textbf{Age}}\\
\hspace{1em}18-34 years & 46.79 (36.85, 56.73) & 53.21 (43.27, 63.15) & 49.26 (40.85, 57.68) & 5.54 (2.16, 8.92) & -1.59 (-6.48, 3.30)\\
\hspace{1em}35-64 years & 52.90 (47.44, 58.36) & 47.10 (41.64, 52.56) & 38.95 (34.25, 43.65) & 7.41 (5.45, 9.37) & 0.74 (-0.71, 2.19)\\
\hspace{1em}65 or more years & 69.82 (64.41, 75.23) & 30.18 (24.77, 35.59) & 21.83 (17.32, 26.33) & 5.62 (4.07, 7.18) & 2.73 (1.21, 4.25)\\
\hspace{1em}No response & 65.20 (51.27, 79.14) & 34.80 (20.86, 48.73) & 25.49 (14.56, 36.41) & 6.96 (1.00, 12.92) & 2.35 (-2.23, 6.93)\\
\addlinespace[0.3em]
\multicolumn{6}{l}{\textbf{Ethnicity}}\\
\hspace{1em}Hispanic & 41.69 (20.98, 62.40) & 58.31 (37.60, 79.02) & 53.39 (34.84, 71.94) & 5.36 (-2.07, 12.80) & -0.44 (-6.17, 5.28)\\
\hspace{1em}No response & 65.80 (53.38, 78.23) & 34.20 (21.77, 46.62) & 24.96 (14.95, 34.98) & 4.68 (0.03, 9.32) & 4.56 (0.47, 8.65)\\
\hspace{1em}Not Hispanic & 56.25 (52.32, 60.18) & 43.75 (39.82, 47.68) & 36.39 (33.09, 39.70) & 6.77 (5.42, 8.11) & 0.59 (-0.94, 2.12)\\
\addlinespace[0.3em]
\multicolumn{6}{l}{\textbf{Gender}}\\
\hspace{1em}Female & 54.60 (50.23, 58.98) & 45.40 (41.02, 49.77) & 36.37 (32.66, 40.08) & 8.15 (6.61, 9.69) & 0.88 (-0.27, 2.02)\\
\hspace{1em}No response & 71.59 (49.07, 94.10) & 28.41 (5.90, 50.93) & 25.93 (4.18, 47.68) & -4.96 (-21.10, 11.18) & 7.44 (-3.42, 18.31)\\
\hspace{1em}Not female & 55.67 (47.52, 63.83) & 44.33 (36.17, 52.48) & 41.17 (34.42, 47.92) & 3.75 (1.10, 6.39) & -0.59 (-4.74, 3.56)\\
\addlinespace[0.3em]
\multicolumn{6}{l}{\textbf{High-risk conditions}}\\
\hspace{1em}Missing & 25.21 (-103.37, 153.79) & 74.79 (-53.79, 203.37) & 80.25 (-47.06, 207.56) & 8.41 (-32.29, 49.12) & -13.88 (-57.29, 29.53)\\
\hspace{1em}None & 46.66 (38.93, 54.39) & 53.34 (45.61, 61.07) & 46.60 (40.30, 52.91) & 5.92 (3.29, 8.54) & 0.82 (-2.74, 4.38)\\
\hspace{1em}One & 55.85 (48.83, 62.87) & 44.15 (37.13, 51.17) & 35.77 (29.77, 41.77) & 7.64 (5.32, 9.97) & 0.74 (-1.12, 2.59)\\
\hspace{1em}Two or more & 63.14 (57.77, 68.50) & 36.86 (31.50, 42.23) & 29.89 (25.12, 34.65) & 6.35 (4.33, 8.36) & 0.63 (-0.86, 2.12)\\
\addlinespace[0.3em]
\multicolumn{6}{l}{\textbf{Race}}\\
\hspace{1em}Asian or Pacific Islander & 49.94 (7.89, 91.98) & 50.06 (8.02, 92.11) & 43.60 (15.04, 72.16) & 18.91 (2.35, 35.48) & -12.45 (-40.86, 15.96)\\
\hspace{1em}Black & 36.25 (-24.88, 97.37) & 63.75 (2.63, 124.88) & 65.89 (-1.45, 133.24) & -11.98 (-46.10, 22.14) & 9.85 (-8.29, 27.98)\\
\hspace{1em}No response & 47.57 (25.24, 69.91) & 52.43 (30.09, 74.76) & 40.68 (23.10, 58.26) & 8.09 (-0.64, 16.82) & 3.66 (-3.85, 11.17)\\
\hspace{1em}Other & 50.86 (31.28, 70.44) & 49.14 (29.56, 68.72) & 44.94 (26.65, 63.23) & 2.04 (-4.34, 8.41) & 2.17 (-4.16, 8.49)\\
\hspace{1em}White & 56.52 (52.85, 60.20) & 43.48 (39.80, 47.15) & 36.03 (32.85, 39.21) & 6.60 (5.35, 7.86) & 0.84 (-0.15, 1.84)\\
\bottomrule
\end{tabular}
\begin{tablenotes}
\item IDE represents the interventional direct effect; IIE represents the total interventional indirect effect; IIE - X represents interventional indirect effect via the indicated channel X; Ref represents the reference category for the indirect effects
\end{tablenotes}
\end{threeparttable}
\end{table}

%% file: tables/mediation-hte-table-0-xgb.tex
\begin{table}[!h]

\caption{Heterogeneous proportion mediated (Ref = 1): XGBoost estimates}
\centering
\begin{threeparttable}
\begin{tabular}[t]{llllll}
\toprule
Label & IDE & IIE & IIE - Isolated & IIE - Health & IIE - Covariant\\
\midrule
\addlinespace[0.3em]
\multicolumn{6}{l}{\textbf{Age}}\\
\hspace{1em}18-34 years & -2.18 (-2.84, -1.53) & -3.93 (-4.26, -3.59) & -3.86 (-4.17, -3.56) & -0.42 (-0.54, -0.30) & 0.36 (0.26, 0.45)\\
\hspace{1em}35-64 years & -1.75 (-2.03, -1.47) & -1.93 (-2.08, -1.77) & -1.58 (-1.72, -1.44) & -0.40 (-0.45, -0.36) & 0.06 (0.03, 0.09)\\
\hspace{1em}65 or more years & -1.79 (-2.03, -1.56) & -0.85 (-0.99, -0.72) & -0.65 (-0.76, -0.53) & -0.18 (-0.22, -0.13) & -0.03 (-0.07, 0.02)\\
\hspace{1em}No response & -2.73 (-4.19, -1.27) & -1.75 (-2.59, -0.91) & -1.48 (-2.27, -0.69) & -0.30 (-0.56, -0.04) & 0.03 (-0.19, 0.25)\\
\addlinespace[0.3em]
\multicolumn{6}{l}{\textbf{Ethnicity}}\\
\hspace{1em}Hispanic & -0.71 (-1.35, -0.07) & -1.81 (-2.17, -1.45) & -1.60 (-1.91, -1.29) & -0.28 (-0.38, -0.18) & 0.07 (-0.04, 0.18)\\
\hspace{1em}No response & -3.23 (-4.44, -2.02) & -1.64 (-2.35, -0.93) & -1.53 (-2.19, -0.87) & -0.16 (-0.38, 0.06) & 0.05 (-0.14, 0.23)\\
\hspace{1em}Not Hispanic & -1.96 (-2.16, -1.75) & -1.91 (-2.02, -1.80) & -1.64 (-1.74, -1.54) & -0.35 (-0.38, -0.31) & 0.08 (0.05, 0.11)\\
\addlinespace[0.3em]
\multicolumn{6}{l}{\textbf{Gender}}\\
\hspace{1em}Female & -1.96 (-2.21, -1.71) & -2.01 (-2.14, -1.88) & -1.70 (-1.82, -1.58) & -0.39 (-0.43, -0.34) & 0.07 (0.04, 0.10)\\
\hspace{1em}No response & -3.04 (-4.67, -1.41) & -1.12 (-2.27, 0.04) & -1.38 (-2.24, -0.51) & -0.39 (-0.70, -0.09) & 0.65 (0.20, 1.11)\\
\hspace{1em}Not female & -1.54 (-1.83, -1.24) & -1.71 (-1.87, -1.54) & -1.52 (-1.68, -1.37) & -0.23 (-0.27, -0.18) & 0.05 (0.01, 0.08)\\
\addlinespace[0.3em]
\multicolumn{6}{l}{\textbf{High-risk conditions}}\\
\hspace{1em}Missing & -1.84 (-4.02, 0.34) & -2.70 (-3.84, -1.56) & -2.68 (-3.79, -1.56) & -0.11 (-0.59, 0.37) & 0.09 (-0.23, 0.40)\\
\hspace{1em}None & -1.44 (-1.77, -1.12) & -2.03 (-2.20, -1.86) & -1.85 (-2.01, -1.70) & -0.30 (-0.36, -0.25) & 0.13 (0.09, 0.17)\\
\hspace{1em}One & -1.68 (-2.04, -1.33) & -1.85 (-2.04, -1.66) & -1.57 (-1.75, -1.40) & -0.32 (-0.38, -0.26) & 0.04 (0.00, 0.09)\\
\hspace{1em}Two or more & -2.38 (-2.71, -2.04) & -1.76 (-1.94, -1.57) & -1.44 (-1.60, -1.27) & -0.38 (-0.44, -0.32) & 0.06 (0.01, 0.11)\\
\addlinespace[0.3em]
\multicolumn{6}{l}{\textbf{Race}}\\
\hspace{1em}Asian or Pacific Islander & -0.81 (-1.57, -0.05) & -1.69 (-2.12, -1.25) & -1.32 (-1.71, -0.93) & -0.34 (-0.48, -0.21) & -0.02 (-0.14, 0.09)\\
\hspace{1em}Black & -0.53 (-1.65, 0.60) & -1.20 (-1.87, -0.52) & -1.23 (-1.86, -0.60) & -0.07 (-0.25, 0.11) & 0.11 (-0.05, 0.26)\\
\hspace{1em}No response & -1.19 (-2.32, -0.07) & -1.98 (-2.60, -1.35) & -1.70 (-2.29, -1.11) & -0.37 (-0.58, -0.16) & 0.10 (-0.06, 0.25)\\
\hspace{1em}Other & -1.11 (-1.93, -0.30) & -1.95 (-2.43, -1.48) & -1.91 (-2.32, -1.50) & -0.23 (-0.36, -0.10) & 0.19 (0.03, 0.34)\\
\hspace{1em}White & -2.06 (-2.27, -1.85) & -1.92 (-2.03, -1.80) & -1.64 (-1.74, -1.54) & -0.35 (-0.39, -0.31) & 0.07 (0.05, 0.10)\\
\bottomrule
\end{tabular}
\begin{tablenotes}
\item IDE represents the interventional direct effect; IIE represents the total interventional indirect effect; IIE - X represents interventional indirect effect via the indicated channel X; Ref represents the reference category for the indirect effects
\end{tablenotes}
\end{threeparttable}
\end{table}

%% file: tables/mediation-hte-table-1-xgb.tex
\begin{table}[!h]

\caption{Heterogeneous mediated effects (Ref = 1): XGBoost estimates}
\centering
\begin{threeparttable}
\begin{tabular}[t]{llllll}
\toprule
Label & IDE & IIE & IIE - Isolated & IIE - Health & IIE - Covariant\\
\midrule
\addlinespace[0.3em]
\multicolumn{6}{l}{\textbf{Age}}\\
\hspace{1em}18-34 years & -2.97 (-3.70, -2.23) & -3.14 (-3.37, -2.92) & -4.20 (-4.54, -3.87) & -0.40 (-0.55, -0.25) & 1.46 (1.06, 1.86)\\
\hspace{1em}35-64 years & -1.91 (-2.23, -1.59) & -1.77 (-1.84, -1.70) & -1.55 (-1.70, -1.40) & -0.36 (-0.42, -0.29) & 0.14 (-0.03, 0.32)\\
\hspace{1em}65 or more years & -1.94 (-2.19, -1.69) & -0.71 (-0.75, -0.67) & -0.58 (-0.68, -0.48) & -0.10 (-0.14, -0.07) & -0.02 (-0.14, 0.09)\\
\hspace{1em}No response & -2.99 (-4.75, -1.23) & -1.49 (-1.96, -1.02) & -1.21 (-2.06, -0.37) & -0.10 (-0.52, 0.32) & -0.17 (-1.19, 0.84)\\
\addlinespace[0.3em]
\multicolumn{6}{l}{\textbf{Ethnicity}}\\
\hspace{1em}Hispanic & -0.67 (-1.35, 0.02) & -1.85 (-2.04, -1.67) & -1.68 (-2.00, -1.36) & -0.22 (-0.37, -0.08) & 0.05 (-0.32, 0.43)\\
\hspace{1em}No response & -3.14 (-4.60, -1.69) & -1.73 (-2.12, -1.34) & -1.48 (-2.17, -0.79) & 0.03 (-0.29, 0.36) & -0.28 (-1.11, 0.55)\\
\hspace{1em}Not Hispanic & -2.26 (-2.49, -2.02) & -1.61 (-1.66, -1.55) & -1.64 (-1.75, -1.54) & -0.29 (-0.34, -0.25) & 0.33 (0.21, 0.45)\\
\addlinespace[0.3em]
\multicolumn{6}{l}{\textbf{Gender}}\\
\hspace{1em}Female & -2.26 (-2.54, -1.98) & -1.71 (-1.77, -1.65) & -1.73 (-1.85, -1.60) & -0.36 (-0.41, -0.30) & 0.37 (0.23, 0.51)\\
\hspace{1em}No response & -2.56 (-4.23, -0.90) & -1.59 (-2.21, -0.98) & -1.59 (-2.45, -0.74) & 0.49 (-0.04, 1.02) & -0.49 (-1.50, 0.52)\\
\hspace{1em}Not female & -1.75 (-2.09, -1.40) & -1.49 (-1.59, -1.40) & -1.48 (-1.65, -1.31) & -0.18 (-0.25, -0.11) & 0.17 (-0.03, 0.36)\\
\addlinespace[0.3em]
\multicolumn{6}{l}{\textbf{High-risk conditions}}\\
\hspace{1em}Missing & -2.46 (-4.94, 0.01) & -2.08 (-2.77, -1.38) & -3.08 (-4.09, -2.07) & -0.03 (-0.51, 0.44) & 1.04 (-0.24, 2.31)\\
\hspace{1em}None & -1.55 (-1.91, -1.18) & -1.93 (-2.02, -1.84) & -1.83 (-2.00, -1.66) & -0.25 (-0.33, -0.18) & 0.15 (-0.05, 0.35)\\
\hspace{1em}One & -2.06 (-2.46, -1.65) & -1.47 (-1.56, -1.38) & -1.54 (-1.73, -1.35) & -0.29 (-0.37, -0.22) & 0.36 (0.15, 0.57)\\
\hspace{1em}Two or more & -2.66 (-3.03, -2.29) & -1.48 (-1.56, -1.39) & -1.51 (-1.67, -1.34) & -0.30 (-0.37, -0.23) & 0.33 (0.14, 0.52)\\
\addlinespace[0.3em]
\multicolumn{6}{l}{\textbf{Race}}\\
\hspace{1em}Asian or Pacific Islander & -0.91 (-1.78, -0.03) & -1.59 (-1.77, -1.40) & -1.40 (-1.83, -0.96) & -0.38 (-0.57, -0.18) & 0.18 (-0.32, 0.69)\\
\hspace{1em}Black & -0.55 (-1.80, 0.71) & -1.18 (-1.46, -0.90) & -0.96 (-1.59, -0.33) & 0.00 (-0.24, 0.24) & -0.22 (-0.93, 0.48)\\
\hspace{1em}No response & -1.27 (-2.58, 0.04) & -1.90 (-2.26, -1.54) & -1.46 (-2.08, -0.84) & -0.22 (-0.53, 0.09) & -0.22 (-0.95, 0.51)\\
\hspace{1em}Other & -1.02 (-1.92, -0.12) & -2.04 (-2.32, -1.77) & -2.23 (-2.65, -1.80) & 0.11 (-0.09, 0.30) & 0.08 (-0.42, 0.57)\\
\hspace{1em}White & -2.37 (-2.61, -2.13) & -1.61 (-1.66, -1.55) & -1.63 (-1.74, -1.53) & -0.32 (-0.37, -0.27) & 0.34 (0.22, 0.47)\\
\bottomrule
\end{tabular}
\begin{tablenotes}
\item IDE represents the interventional direct effect; IIE represents the total interventional indirect effect; IIE - X represents interventional indirect effect via the indicated channel X; Ref represents the reference category for the indirect effects
\end{tablenotes}
\end{threeparttable}
\end{table}

%% file: tables/mediation-hte-prop-table-0-xgb.tex
\begin{table}[!h]

\caption{Heterogeneous proportion mediated (Ref = 0): XGBoost estimates}
\centering
\begin{threeparttable}
\begin{tabular}[t]{llllll}
\toprule
Label & IDE & IIE & IIE - Isolated & IIE - Health & IIE - Covariant\\
\midrule
\addlinespace[0.3em]
\multicolumn{6}{l}{\textbf{Age}}\\
\hspace{1em}18-34 years & 35.76 (28.39, 43.12) & 64.24 (56.88, 71.61) & 63.22 (56.19, 70.24) & 6.88 (4.77, 8.99) & -5.85 (-7.55, -4.15)\\
\hspace{1em}35-64 years & 47.61 (42.98, 52.24) & 52.39 (47.76, 57.02) & 43.06 (39.00, 47.12) & 11.00 (9.47, 12.54) & -1.67 (-2.52, -0.83)\\
\hspace{1em}65 or more years & 67.79 (63.05, 72.53) & 32.21 (27.47, 36.95) & 24.46 (20.50, 28.42) & 6.67 (4.92, 8.41) & 1.09 (-0.61, 2.80)\\
\hspace{1em}No response & 60.97 (43.09, 78.85) & 39.03 (21.15, 56.91) & 33.06 (16.73, 49.39) & 6.73 (0.55, 12.92) & -0.76 (-5.69, 4.17)\\
\addlinespace[0.3em]
\multicolumn{6}{l}{\textbf{Ethnicity}}\\
\hspace{1em}Hispanic & 28.17 (8.83, 47.51) & 71.83 (52.49, 91.17) & 63.51 (47.05, 79.97) & 11.10 (6.30, 15.91) & -2.78 (-7.17, 1.60)\\
\hspace{1em}No response & 66.30 (52.96, 79.64) & 33.70 (20.36, 47.04) & 31.36 (19.12, 43.60) & 3.32 (-1.31, 7.96) & -0.99 (-4.77, 2.80)\\
\hspace{1em}Not Hispanic & 50.67 (47.60, 53.74) & 49.33 (46.26, 52.40) & 42.44 (39.70, 45.18) & 8.94 (7.89, 9.99) & -2.05 (-2.74, -1.37)\\
\addlinespace[0.3em]
\multicolumn{6}{l}{\textbf{Gender}}\\
\hspace{1em}Female & 49.41 (45.71, 53.11) & 50.59 (46.89, 54.29) & 42.72 (39.47, 45.97) & 9.72 (8.44, 11.00) & -1.85 (-2.65, -1.05)\\
\hspace{1em}No response & 73.15 (46.59, 99.70) & 26.85 (0.30, 53.41) & 33.10 (13.71, 52.49) & 9.47 (1.38, 17.57) & -15.72 (-28.00, -3.45)\\
\hspace{1em}Not female & 47.38 (41.89, 52.87) & 52.62 (47.13, 58.11) & 47.03 (42.00, 52.07) & 6.98 (5.39, 8.57) & -1.40 (-2.59, -0.21)\\
\addlinespace[0.3em]
\multicolumn{6}{l}{\textbf{High-risk conditions}}\\
\hspace{1em}Missing & 40.47 (9.72, 71.22) & 59.53 (28.78, 90.28) & 58.99 (28.69, 89.28) & 2.43 (-8.21, 13.07) & -1.88 (-8.95, 5.19)\\
\hspace{1em}None & 41.50 (35.57, 47.44) & 58.50 (52.56, 64.43) & 53.38 (47.95, 58.81) & 8.77 (6.95, 10.59) & -3.65 (-4.90, -2.40)\\
\hspace{1em}One & 47.62 (41.61, 53.63) & 52.38 (46.37, 58.39) & 44.59 (39.28, 49.90) & 9.01 (7.08, 10.95) & -1.22 (-2.50, 0.05)\\
\hspace{1em}Two or more & 57.46 (52.96, 61.96) & 42.54 (38.04, 47.04) & 34.73 (30.86, 38.60) & 9.19 (7.61, 10.76) & -1.37 (-2.55, -0.20)\\
\addlinespace[0.3em]
\multicolumn{6}{l}{\textbf{Race}}\\
\hspace{1em}Asian or Pacific Islander & 32.32 (10.18, 54.47) & 67.68 (45.53, 89.82) & 52.99 (35.26, 70.72) & 13.80 (7.10, 20.51) & 0.88 (-3.69, 5.45)\\
\hspace{1em}Black & 30.65 (-17.42, 78.72) & 69.35 (21.28, 117.42) & 71.51 (23.35, 119.67) & 4.11 (-6.31, 14.53) & -6.27 (-16.27, 3.73)\\
\hspace{1em}No response & 37.61 (13.62, 61.60) & 62.39 (38.40, 86.38) & 53.77 (32.56, 74.98) & 11.64 (3.61, 19.67) & -3.02 (-8.06, 2.02)\\
\hspace{1em}Other & 36.27 (17.42, 55.12) & 63.73 (44.88, 82.58) & 62.30 (45.29, 79.31) & 7.51 (2.82, 12.19) & -6.08 (-11.40, -0.76)\\
\hspace{1em}White & 51.78 (48.73, 54.82) & 48.22 (45.18, 51.27) & 41.26 (38.54, 43.97) & 8.80 (7.75, 9.85) & -1.83 (-2.51, -1.16)\\
\bottomrule
\end{tabular}
\begin{tablenotes}
\item IDE represents the interventional direct effect; IIE represents the total interventional indirect effect; IIE - X represents interventional indirect effect via the indicated channel X; Ref represents the reference category for the indirect effects
\end{tablenotes}
\end{threeparttable}
\end{table}

%% file: tables/mediation-hte-prop-table-1-xgb.tex
\begin{table}[!h]

\caption{Heterogeneous proportion mediated (Ref = 1): XGBoost estimates}
\centering
\begin{threeparttable}
\begin{tabular}[t]{llllll}
\toprule
Label & IDE & IIE & IIE - Isolated & IIE - Health & IIE - Covariant\\
\midrule
\addlinespace[0.3em]
\multicolumn{6}{l}{\textbf{Age}}\\
\hspace{1em}18-34 years & 48.57 (40.75, 56.39) & 51.43 (43.61, 59.25) & 68.78 (60.31, 77.25) & 6.51 (3.99, 9.03) & -23.86 (-28.51, -19.22)\\
\hspace{1em}35-64 years & 51.94 (46.37, 57.51) & 48.06 (42.49, 53.63) & 42.19 (37.41, 46.98) & 9.72 (7.73, 11.72) & -3.86 (-5.81, -1.92)\\
\hspace{1em}65 or more years & 73.29 (68.88, 77.71) & 26.71 (22.29, 31.12) & 21.94 (18.26, 25.62) & 3.86 (2.44, 5.27) & 0.91 (-0.65, 2.46)\\
\hspace{1em}No response & 66.81 (43.00, 90.61) & 33.19 (9.39, 57.00) & 27.12 (7.96, 46.27) & 2.18 (-7.26, 11.62) & 3.90 (-6.80, 14.60)\\
\addlinespace[0.3em]
\multicolumn{6}{l}{\textbf{Ethnicity}}\\
\hspace{1em}Hispanic & 26.43 (4.73, 48.12) & 73.57 (51.88, 95.27) & 66.73 (48.03, 85.43) & 8.90 (2.53, 15.26) & -2.05 (-9.44, 5.33)\\
\hspace{1em}No response & 64.50 (46.33, 82.67) & 35.50 (17.33, 53.67) & 30.40 (15.80, 45.00) & -0.69 (-7.40, 6.01) & 5.79 (-2.49, 14.07)\\
\hspace{1em}Not Hispanic & 58.42 (54.86, 61.97) & 41.58 (38.03, 45.14) & 42.49 (39.26, 45.71) & 7.63 (6.37, 8.88) & -8.53 (-10.02, -7.04)\\
\addlinespace[0.3em]
\multicolumn{6}{l}{\textbf{Gender}}\\
\hspace{1em}Female & 56.89 (52.73, 61.05) & 43.11 (38.95, 47.27) & 43.47 (39.72, 47.23) & 8.95 (7.45, 10.44) & -9.31 (-11.01, -7.61)\\
\hspace{1em}No response & 61.68 (36.12, 87.24) & 38.32 (12.76, 63.88) & 38.35 (16.75, 59.96) & -11.81 (-25.61, 1.98) & 11.78 (-4.07, 27.63)\\
\hspace{1em}Not female & 53.90 (47.09, 60.72) & 46.10 (39.28, 52.91) & 45.65 (39.54, 51.76) & 5.62 (3.42, 7.82) & -5.17 (-8.04, -2.30)\\
\addlinespace[0.3em]
\multicolumn{6}{l}{\textbf{High-risk conditions}}\\
\hspace{1em}Missing & 54.28 (21.32, 87.23) & 45.72 (12.77, 78.68) & 67.77 (31.10, 104.45) & 0.77 (-9.75, 11.28) & -22.82 (-42.27, -3.37)\\
\hspace{1em}None & 44.51 (37.38, 51.63) & 55.49 (48.37, 62.62) & 52.60 (46.31, 58.90) & 7.28 (4.95, 9.60) & -4.39 (-7.07, -1.70)\\
\hspace{1em}One & 58.29 (51.49, 65.10) & 41.71 (34.90, 48.51) & 43.61 (37.30, 49.91) & 8.33 (5.98, 10.68) & -10.23 (-12.97, -7.49)\\
\hspace{1em}Two or more & 64.30 (59.31, 69.30) & 35.70 (30.70, 40.69) & 36.46 (32.05, 40.87) & 7.21 (5.36, 9.07) & -7.98 (-10.26, -5.69)\\
\addlinespace[0.3em]
\multicolumn{6}{l}{\textbf{Race}}\\
\hspace{1em}Asian or Pacific Islander & 36.42 (10.28, 62.56) & 63.58 (37.44, 89.72) & 55.93 (33.78, 78.07) & 15.04 (5.78, 24.31) & -7.39 (-15.23, 0.45)\\
\hspace{1em}Black & 31.65 (-23.31, 86.60) & 68.35 (13.40, 123.31) & 55.59 (9.93, 101.25) & -0.21 (-14.11, 13.69) & 12.98 (-6.03, 31.98)\\
\hspace{1em}No response & 40.03 (10.50, 69.56) & 59.97 (30.44, 89.50) & 46.12 (23.33, 68.90) & 6.92 (-3.24, 17.08) & 6.93 (-4.89, 18.76)\\
\hspace{1em}Other & 33.32 (11.37, 55.27) & 66.68 (44.73, 88.63) & 72.67 (51.08, 94.27) & -3.49 (-9.93, 2.94) & -2.50 (-11.53, 6.53)\\
\hspace{1em}White & 59.53 (56.02, 63.04) & 40.47 (36.96, 43.98) & 41.10 (37.94, 44.26) & 8.04 (6.78, 9.29) & -8.67 (-10.16, -7.19)\\
\bottomrule
\end{tabular}
\begin{tablenotes}
\item IDE represents the interventional direct effect; IIE represents the total interventional indirect effect; IIE - X represents interventional indirect effect via the indicated channel X; Ref represents the reference category for the indirect effects
\end{tablenotes}
\end{threeparttable}
\end{table}

%% file: tables/outcomes-inc-table-xgb-cc.tex
% latex table generated in R 4.2.0 by xtable 1.8-4 package
% Wed Nov 16 17:42:44 2022
\begin{table}[ht]
\centering
\caption{Incremental effects, XGBoost estimates (95\% CI)} 
\label{tab:inctabxgbcc}
\begin{tabular}{llll}
  \hline
Increment & Depression & Isolation & Health \\ 
  \hline
0 & 19.2 (19.1, 19.3) & 20.3 (20.2, 20.4) & 28.4 (28.3, 28.6) \\ 
  0.01 & 19.1 (19.0, 19.3) & 20.3 (20.1, 20.4) & 28.4 (28.2, 28.5) \\ 
  0.25 & 18.6 (18.5, 18.7) & 19.8 (19.7, 19.9) & 27.7 (27.6, 27.9) \\ 
  0.5 & 18.3 (18.2, 18.4) & 19.5 (19.4, 19.6) & 27.4 (27.3, 27.5) \\ 
  0.75 & 18.1 (18.0, 18.2) & 19.4 (19.3, 19.5) & 27.2 (27.1, 27.3) \\ 
  1 & 18.0 (17.9, 18.1) & 19.2 (19.1, 19.3) & 27.0 (26.9, 27.1) \\ 
  1.5 & 17.8 (17.7, 17.8) & 19.0 (18.9, 19.1) & 26.8 (26.6, 26.9) \\ 
  2 & 17.6 (17.5, 17.7) & 18.9 (18.8, 19.0) & 26.6 (26.5, 26.7) \\ 
  2.5 & 17.4 (17.4, 17.5) & 18.8 (18.7, 18.9) & 26.4 (26.3, 26.5) \\ 
  3 & 17.3 (17.2, 17.4) & 18.7 (18.6, 18.8) & 26.3 (26.2, 26.4) \\ 
  3.5 & 17.2 (17.1, 17.3) & 18.6 (18.5, 18.7) & 26.2 (26.0, 26.3) \\ 
  4 & 17.1 (17.0, 17.2) & 18.5 (18.4, 18.6) & 26.1 (25.9, 26.2) \\ 
  4.5 & 17.1 (17.0, 17.2) & 18.4 (18.3, 18.6) & 26.0 (25.9, 26.1) \\ 
  5 & 17.0 (16.9, 17.1) & 18.4 (18.3, 18.5) & 25.9 (25.8, 26.0) \\ 
   \hline
Observed & -1.2 (-1.3, -1.1) & -1.1 (-1.2, -1.0) & -1.4 (-1.5, -1.3) \\ 
  Proportion Observed & 32.7 (31.2, 34.2) & 33.5 (31.6, 35.3) & 34.5 (32.9, 36.1) \\ 
  Total & -3.7 (-3.9, -3.5) & -3.2 (-3.4, -3.0) & -4.1 (-4.4, -3.9) \\ 
   \hline
\end{tabular}
\end{table}

%% file: tables/outcomes-inc-table-glm-cc.tex
% latex table generated in R 4.2.0 by xtable 1.8-4 package
% Wed Nov 16 17:42:20 2022
\begin{table}[ht]
\centering
\caption{Incremental effects, GLM estimates (95\% CI)} 
\label{tab:inctabglmcc}
\begin{tabular}{llll}
  \hline
Increment & Depression & Isolation & Health \\ 
  \hline
0 & 19.2 (19.0, 19.3) & 20.3 (20.2, 20.4) & 28.4 (28.3, 28.6) \\ 
  0.01 & 19.1 (19.0, 19.3) & 20.3 (20.1, 20.4) & 28.4 (28.2, 28.6) \\ 
  0.25 & 18.6 (18.5, 18.7) & 19.8 (19.7, 19.9) & 27.8 (27.7, 27.9) \\ 
  0.5 & 18.3 (18.2, 18.4) & 19.5 (19.4, 19.7) & 27.4 (27.3, 27.6) \\ 
  0.75 & 18.1 (18.0, 18.2) & 19.4 (19.3, 19.5) & 27.2 (27.1, 27.3) \\ 
  1 & 18.0 (17.9, 18.1) & 19.2 (19.1, 19.3) & 27.0 (26.9, 27.1) \\ 
  1.5 & 17.7 (17.7, 17.8) & 19.0 (18.9, 19.1) & 26.7 (26.6, 26.9) \\ 
  2 & 17.6 (17.5, 17.7) & 18.9 (18.8, 19.0) & 26.5 (26.4, 26.6) \\ 
  2.5 & 17.4 (17.3, 17.5) & 18.7 (18.6, 18.8) & 26.4 (26.2, 26.5) \\ 
  3 & 17.3 (17.2, 17.4) & 18.6 (18.5, 18.7) & 26.2 (26.1, 26.3) \\ 
  3.5 & 17.2 (17.1, 17.3) & 18.5 (18.4, 18.6) & 26.1 (26.0, 26.2) \\ 
  4 & 17.1 (17.0, 17.2) & 18.5 (18.3, 18.6) & 26.0 (25.9, 26.1) \\ 
  4.5 & 17.0 (16.9, 17.1) & 18.4 (18.3, 18.5) & 25.9 (25.8, 26.1) \\ 
  5 & 17.0 (16.9, 17.1) & 18.3 (18.2, 18.4) & 25.8 (25.7, 26.0) \\ 
   \hline
Observed & -1.2 (-1.3, -1.1) & -1.1 (-1.2, -1.0) & -1.4 (-1.6, -1.3) \\ 
  Proportion Observed & 31.8 (29.9, 33.8) & 33.0 (30.5, 35.4) & 33.5 (31.3, 35.7) \\ 
  Total & -3.7 (-4.0, -3.5) & -3.3 (-3.5, -3.0) & -4.3 (-4.5, -4.0) \\ 
   \hline
\end{tabular}
\end{table}

%% file: tables/age-depression-inc-table-xgb.tex
% latex table generated in R 4.2.0 by xtable 1.8-4 package
% Wed Nov 16 17:42:48 2022
\begin{table}[ht]
\centering
\caption{Incremental effects by age: depression (95\% CI)} 
\label{tab:ageincdepxgb}
\begin{tabular}{lllll}
  \hline
Increment & Outcome & 18-24 years & 25+ years: Hispanic & 25+ years: not Hispanic \\ 
  \hline
0 & Depression &  50.1 ( 49.2,  50.9) &  18.9 ( 18.5,  19.2) &  17.9 ( 17.8,  18.0) \\ 
  0.01 & Depression &  50.1 ( 49.1,  51.1) &  18.9 ( 18.5,  19.3) &  17.8 ( 17.7,  18.0) \\ 
  0.25 & Depression &  49.8 ( 49.0,  50.6) &  18.7 ( 18.3,  19.0) &  17.2 ( 17.1,  17.3) \\ 
  0.5 & Depression &  49.4 ( 48.7,  50.1) &  18.5 ( 18.2,  18.8) &  16.9 ( 16.8,  17.1) \\ 
  0.75 & Depression &  49.1 ( 48.3,  49.8) &  18.4 ( 18.1,  18.7) &  16.8 ( 16.6,  16.9) \\ 
  1 & Depression &  48.8 ( 48.1,  49.5) &  18.3 ( 18.0,  18.6) &  16.6 ( 16.5,  16.7) \\ 
  1.5 & Depression &  48.3 ( 47.6,  49.0) &  18.2 ( 17.9,  18.5) &  16.4 ( 16.3,  16.5) \\ 
  2 & Depression &  47.8 ( 47.1,  48.5) &  18.1 ( 17.7,  18.4) &  16.2 ( 16.1,  16.3) \\ 
  2.5 & Depression &  47.4 ( 46.7,  48.2) &  18.0 ( 17.6,  18.3) &  16.1 ( 16.0,  16.2) \\ 
  3 & Depression &  47.1 ( 46.4,  47.8) &  17.9 ( 17.6,  18.2) &  16.0 ( 15.9,  16.1) \\ 
  3.5 & Depression &  46.8 ( 46.0,  47.5) &  17.8 ( 17.5,  18.2) &  15.9 ( 15.8,  16.0) \\ 
  4 & Depression &  46.5 ( 45.7,  47.2) &  17.8 ( 17.4,  18.1) &  15.8 ( 15.7,  15.9) \\ 
  4.5 & Depression &  46.2 ( 45.4,  47.0) &  17.7 ( 17.4,  18.1) &  15.8 ( 15.6,  15.9) \\ 
  5 & Depression &  45.9 ( 45.1,  46.7) &  17.7 ( 17.3,  18.0) &  15.7 ( 15.6,  15.8) \\ 
   \hline
Observed & Depression &  -1.3 ( -1.9,  -0.7) &  -0.6 ( -0.8,  -0.4) &  -1.3 ( -1.4,  -1.2) \\ 
  Proportion Observed & Depression &  10.6 (6.6,  14.6) &  30.4 ( 21.5,  39.3) &  36.0 ( 34.1,  37.9) \\ 
  Total & Depression & -12.2 (-13.7, -10.7) &  -1.9 ( -2.6,  -1.2) &  -3.6 ( -3.8,  -3.3) \\ 
   \hline
\end{tabular}
\end{table}

%% file: tables/age-isolation-inc-table-xgb.tex
% latex table generated in R 4.2.0 by xtable 1.8-4 package
% Wed Nov 16 17:42:48 2022
\begin{table}[ht]
\centering
\caption{Incremental effects by age: isolation (95\% CI)} 
\label{tab:ageincisoxgb}
\begin{tabular}{lllll}
  \hline
Increment & Outcome & 18-24 years & 25+ years: Hispanic & 25+ years: not Hispanic \\ 
  \hline
0 & Isolation &  39.0 ( 38.2, 39.8) &  20.2 ( 19.9, 20.6) &  19.5 ( 19.4, 19.7) \\ 
  0.01 & Isolation &  39.0 ( 38.0, 39.9) &  20.2 ( 19.8, 20.6) &  19.5 ( 19.3, 19.6) \\ 
  0.25 & Isolation &  38.5 ( 37.8, 39.2) &  20.0 ( 19.7, 20.4) &  19.0 ( 18.8, 19.1) \\ 
  0.5 & Isolation &  38.1 ( 37.4, 38.8) &  19.9 ( 19.5, 20.2) &  18.7 ( 18.6, 18.8) \\ 
  0.75 & Isolation &  37.8 ( 37.1, 38.5) &  19.8 ( 19.4, 20.1) &  18.5 ( 18.4, 18.7) \\ 
  1 & Isolation &  37.5 ( 36.8, 38.2) &  19.7 ( 19.3, 20.0) &  18.4 ( 18.3, 18.5) \\ 
  1.5 & Isolation &  37.0 ( 36.3, 37.7) &  19.5 ( 19.2, 19.8) &  18.2 ( 18.1, 18.3) \\ 
  2 & Isolation &  36.6 ( 35.9, 37.3) &  19.4 ( 19.1, 19.7) &  18.1 ( 18.0, 18.2) \\ 
  2.5 & Isolation &  36.2 ( 35.6, 36.9) &  19.3 ( 19.0, 19.7) &  18.0 ( 17.9, 18.1) \\ 
  3 & Isolation &  35.9 ( 35.2, 36.6) &  19.2 ( 18.9, 19.6) &  17.9 ( 17.8, 18.0) \\ 
  3.5 & Isolation &  35.6 ( 34.9, 36.3) &  19.2 ( 18.8, 19.5) &  17.8 ( 17.7, 17.9) \\ 
  4 & Isolation &  35.3 ( 34.6, 36.1) &  19.1 ( 18.8, 19.5) &  17.7 ( 17.6, 17.9) \\ 
  4.5 & Isolation &  35.1 ( 34.4, 35.8) &  19.1 ( 18.7, 19.5) &  17.7 ( 17.6, 17.8) \\ 
  5 & Isolation &  34.9 ( 34.1, 35.6) &  19.0 ( 18.6, 19.4) &  17.6 ( 17.5, 17.8) \\ 
   \hline
Observed & Isolation &  -1.5 ( -2.0, -0.9) &  -0.6 ( -0.8, -0.4) &  -1.1 ( -1.2, -1.0) \\ 
  Proportion Observed & Isolation &  13.8 (9.9, 17.7) &  31.0 ( 21.7, 40.2) &  36.6 ( 34.3, 38.9) \\ 
  Total & Isolation & -10.7 (-12.1, -9.3) &  -1.9 ( -2.6, -1.2) &  -3.0 ( -3.3, -2.8) \\ 
   \hline
\end{tabular}
\end{table}

%% file: tables/age-health-inc-table-xgb.tex
% latex table generated in R 4.2.0 by xtable 1.8-4 package
% Wed Nov 16 17:42:51 2022
\begin{table}[ht]
\centering
\caption{Incremental effects by age: health (95\% CI)} 
\label{tab:ageinchealthxgb}
\begin{tabular}{lllll}
  \hline
Increment & Outcome & 18-24 years & 25+ years: Hispanic & 25+ years: not Hispanic \\ 
  \hline
0 & Health & 35.6 (34.8, 36.4) & 48.0 (47.6, 48.5) & 25.6 (25.4, 25.8) \\ 
  0.01 & Health & 35.6 (34.6, 36.5) & 48.0 (47.4, 48.5) & 25.5 (25.4, 25.7) \\ 
  0.25 & Health & 35.1 (34.4, 35.8) & 47.5 (47.0, 47.9) & 24.9 (24.7, 25.0) \\ 
  0.5 & Health & 34.8 (34.1, 35.5) & 47.2 (46.8, 47.6) & 24.5 (24.4, 24.7) \\ 
  0.75 & Health & 34.6 (33.9, 35.3) & 46.9 (46.5, 47.4) & 24.3 (24.2, 24.4) \\ 
  1 & Health & 34.4 (33.7, 35.1) & 46.8 (46.3, 47.2) & 24.2 (24.0, 24.3) \\ 
  1.5 & Health & 34.1 (33.4, 34.8) & 46.4 (46.0, 46.9) & 23.9 (23.8, 24.0) \\ 
  2 & Health & 33.8 (33.2, 34.5) & 46.2 (45.8, 46.6) & 23.7 (23.6, 23.9) \\ 
  2.5 & Health & 33.6 (32.9, 34.3) & 46.0 (45.6, 46.4) & 23.6 (23.5, 23.7) \\ 
  3 & Health & 33.4 (32.7, 34.1) & 45.8 (45.4, 46.3) & 23.5 (23.3, 23.6) \\ 
  3.5 & Health & 33.2 (32.5, 33.9) & 45.7 (45.2, 46.2) & 23.4 (23.2, 23.5) \\ 
  4 & Health & 33.0 (32.3, 33.7) & 45.6 (45.1, 46.0) & 23.3 (23.1, 23.4) \\ 
  4.5 & Health & 32.9 (32.1, 33.6) & 45.5 (45.0, 45.9) & 23.2 (23.1, 23.3) \\ 
  5 & Health & 32.7 (32.0, 33.5) & 45.4 (44.9, 45.9) & 23.1 (23.0, 23.3) \\ 
   \hline
Observed & Health & -1.2 (-1.7, -0.7) & -1.3 (-1.6, -1.0) & -1.4 (-1.6, -1.3) \\ 
  Proportion Observed & Health & 17.4 (11.7, 23.1) & 26.9 (22.4, 31.4) & 36.9 (34.9, 38.9) \\ 
  Total & Health & -6.9 (-8.3, -5.6) & -4.7 (-5.6, -3.8) & -3.9 (-4.2, -3.7) \\ 
   \hline
\end{tabular}
\end{table}

%% file: appendices/06-sensitivity.tex
\section{Sensitivity analysis}\label{app:sensitivity}

This section contains two propositions: first, one containing the derivation of our sensitivity analysis used in Section~\ref{sec:limitations}. This result follows almost immediately from results previously shown in \cite{luedtke2015statistics}. Our second proposition provides a simple extension to generate bounds on an estimand that does not assume random sample selection.

Throughout we let $\bar{\mu}_a(x, a') = \mathbb{E}[Y^a \mid X = x, A = a', Z = 1]$ and again let $\mu_a(x) = \mathbb{E}[Y^a \mid X = x, Z = 1]$. We define the true (conditional) target of inference as $\psi(x) = \mu_1(x) - \mu_0(x)$ and the biased target of the observed data functionals as $\bar{\psi}(x) = \bar{\mu}_1(x, 1) - \bar{\mu}_0(x, 0)$. Our goal is to bound $\mathbb{P}_n[\psi(X)]$ using only functions of the observed data.

\subsubsection{Sensitivity analysis 1}

In this subsection we propose a sensitivity analysis for the primary estimand $\mathbb{P}_n[\mu_1(X) - \mu_0(X)]$. Recall the assumption in \eqref{eqn:sensitivity1} that there is some known $\tau$ satisfying:

\begin{align*}
    \left(\frac{1}{1-\tau}\right) \ge \frac{\bar{\mu}_a(x, a')}{\bar{\mu}_a(x, a)} \ge 1-\tau \qquad \forall (x, a, a') 
\end{align*}

\begin{proposition}\label{prop1}
For a fixed $\tau$ satisfying \eqref{eqn:sensitivity1}, we obtain the bounds:

\begin{align*}
\psi(x) &\le \bar{\psi}(x) + \left(\frac{\tau}{1 - \tau}\right)\bar{\mu}_1(x, 1) P(A = 0 \mid X = x, Z = 1) \\
&+ \tau \bar{\mu}_0(x, 0)P(A = 1 \mid X = x, Z = 1) \\
\psi(x) &\ge \bar{\psi}(x) - \tau \bar{\mu}_1(x, 1)P(A = 0 \mid X = x, Z = 1) \\
&- \left(\frac{\tau}{1 - \tau}\right)\bar{\mu}_0(x, 0) P(A = 1 \mid X = x, Z = 1)
\end{align*}
\end{proposition}

\begin{proof}[Proof of Proposition~\ref{prop1}]
\begin{align}
    \nonumber&\mathbb{E}[Y^a \mid X = x, Z = 1] \\
    \nonumber&= \sum_{a^\star \in \{a, a'\}} \bar{\mu}_a(x, a^\star)P(A = a^\star \mid X = x, Z = 1) \\
    \label{eqn:propproof1}&= \bar{\mu}_a(x, a) + P(A = a' \mid X = x, Z = 1)[\bar{\mu}_a(x, a') - \bar{\mu}_a(x, a)] 
\end{align}
where the first equality holds by the law of iterated expectations, and the second by rearranging terms. For a fixed $\tau$, \eqref{eqn:sensitivity1} implies that:

\begin{align*}
\left(\frac{\tau}{1 - \tau}\right)\bar{\mu}_a(x, a) \ge \bar{\mu}_a(x, a') - \bar{\mu}_a(x, a) \ge - \bar{\mu}_a(x, a)\tau
\end{align*}

Plugging this into \eqref{eqn:propproof1}, we obtain the following inequalities:

\begin{align}
    \nonumber&\left(\frac{\tau}{1 - \tau}\right) \bar{\mu}_a(x, a) P(A = a' \mid X = x, Z = 1) \\
    \nonumber&\ge \mu_a(x) - \bar{\mu}_a(x, a) \\
     \nonumber&\ge -\tau \bar{\mu}_a(x, a)P(A = a' \mid X = x, Z = 1) 
\end{align}

The result follows by taking the upper bound of $\mu_1(x)$ minus the lower bound for $\mu_0(x)$ to an upper bound on $\psi(x)$, and vice versa for the lower bound.
\end{proof}

\begin{remark}
To obtain bounds on the average effects, we can simply choose a $\tau$ valid for all values of $(x, a)$ and average the expression over the relevant covariate distribution.
\end{remark}

\begin{remark}
These bounds are expressed in terms of the observed data distribution. In order to estimate them we need to account for the uncertainty in our estimates of the nuisance functions. We use influence-function based estimators of these quantities (averaged over the empirical covariate distribution). We refer to \cite{luedtke2015statistics} for more details.
\end{remark}

\subsection{Sensitivity analysis 2}

Our target estimand is defined with respect to the observed sample and the above sensitivity analysis does not account for the fact that the prevalence of potential outcomes in our sample may not be representative. For instance, consider the following two estimands:

\begin{align}
    &\label{eqn:estimand1}\mathbb{P}_n(\mathbb{E}[Y^1 - Y^0 \mid X, V = 1]) \\
    &\label{eqn:estimand2}\mathbb{P}_n(\mathbb{E}[Y^1 - Y^0 \mid X, Z = 1])
\end{align}

Random sample selection implies that equations are equal (see Appendix~\ref{app:identification}). This equality implies that while the CTIS may have a non-representative covariate distribution, our statistical target of inference generalizes across covariate distributions of CTIS non-respondents, and that we could reweight our estimates to generalize to populations with different covariate distributions.

We now instead consider the case where we these quantities are not equal due to non-random sample selection. For example, if the selection process were a function of observed or potential outcomes our sample estimand be large in magnitude while the effects among non-respondents were zero. To be clear, this would imply that our effect estimates cannot generalize to a comparable covariate distribution of CTIS non-respondents. 

We therefore consider targeting \eqref{eqn:estimand1} instead of \eqref{eqn:estimand2}. To obtain a bound on this estimand, we only assume the following:

\begin{align}\label{eqn:sensitivity2}
    \frac{1}{1-\tau} \ge \frac{\mathbb{E}[Y^a \mid X = x, V = 1, Q_a = a']}{\mathbb{E}[Y \mid X = x, V = 1, Q_a = a]} \ge 1-\tau \qquad \forall (x, a, a') 
\end{align}
where $Q_a = \mathds{1}(A = a)RS$. We next derive bounds for equation~(\ref{eqn:estimand1}) while relaxing the assumption of Random Sample Selection (equation~(\ref{eqn:amarx})).

\begin{proposition}\label{prop2}
Define $\tilde{\mu}_a(x, a') = \mathbb{E}[Y^a \mid X = x, V = 1, Q_a = a']$, $\tilde{\psi} = \tilde{\mu}_1(x, 1) - \tilde{\mu}_0(x, 0)$, and let $\dot{\mu}_a(x) = \mathbb{E}[Y^a \mid X, V = 1]$. Assume that \eqref{eqn:sensitivity2} holds for a known $\tau$ for all $(x, a, a')$. Then for any $X = x$:

\begin{align*}
    &\dot{\mu}_1(x) - \dot{\mu}_0(x) \le \tilde{\psi}(x) + \tilde{\mu}_1(x, 1)\left(\frac{1}{1 - \tau}\right) + \tau \tilde{\mu}_0(x, 0) \\
    &\dot{\mu}_1(x) - \dot{\mu}_0(x) \ge \tilde{\psi}(x) - \tau \tilde{\mu}_1(x, 1) - \tilde{\mu}_0(x, 0)\left(\frac{\tau}{1 - \tau}\right) 
\end{align*}
\end{proposition}

\begin{proof}[Proof of Proposition~\ref{prop2}]\label{prop2}
The result follows by following the same steps in the proof of Proposition~\ref{prop1}, but iterating expectation on the variable $Q_a$. The challenge is that we neither know nor can estimate $P(Q_a = 1 \mid X = x, V = 1)$, since we do not observe any characteristics about survey non-respondents. On the other hand, this value must always be less than 1. This then yields the following inequality:

\begin{align*}
    &\left(\frac{\tau}{1 - \tau}\right) \tilde{\mu}_a(x, a) \ge \left(\frac{\tau}{1 - \tau}\right)\tilde{\mu}_a(x, a)P(Q_a = a' \mid X = x, V = 1) \\
    \nonumber&\ge \dot{\mu}_a(x) - \tilde{\mu}_a(x, a) \\
     \nonumber&\ge -\tau \tilde{\mu}_a(x, a)P(Q_a = a' \mid X = x, V = 1) \ge -\tau \bar{\mu}_a(x, a)
\end{align*}

The result follows by taking the upper bound for $\dot{\mu}_1(x)$ and subtracting the lower bound for $\dot{\mu}_0(x)$ for the upper bound of $\dot{\psi}(x)$, and vice versa for the lower bound.
\end{proof}

\begin{remark}
We can again average this expression over the empirical covariate distribution to obtain bounds on $\dot{\psi}$, and can estimate these bounds using influence-function based estimators to account for the uncertainty in the estimates of the nuisance functions.
\end{remark}

\begin{remark}
For our application we are only interested in an upper bound, and, in particular, calculating a value of $\tau^\star$ that would explain away the estimated treatment effect. Letting $\tilde{\psi}_a = \mathbb{E}[\tilde{\mu}_a(X, a)]$, this would occur where:

\begin{align*}
    \tau^\star = 1 - \sqrt{\frac{\tilde{\psi}_1}{\tilde{\psi}_0}}
\end{align*}
While a simple formula, this ignores the uncertainty in the estimates.
\end{remark}

%% file: appendices/07-robustness.tex
\section{Robustness checks}\label{app:robustness}

This section presents several robustness checks. These primarily involve running our analyses on different subsets of the data, and / or including / excluding different covariates. We describe each check in detail below. We rerun our mediation analyses across all of the specified checks. The results, when not provided, were similar to the primary analyses and are available on request.

\subsection{Vaccine-accepting sample}

Consider the target estimand:

\begin{align*}
\psi^{va} = \frac{1}{\sum_{i=1}^nV}\sum_{i: V = 1}^n\mathbb{E}[Y^1 - Y^0 \mid X, VS = 1]
\end{align*}

In contrast to our primary estimand, this estimand includes the incomplete-cases. As a first approach we directly derive point estimates for this quantity. These estimates require no additional assumptions beyond those previously invoked, but do require using inverse-weighting with respect to the probability of being a complete-case (see Appendix~\ref{app:estimation}). As a second approach we instead derive bounds on this quantity.

\subsubsection{Point estimates}

Our estimates are qualitatively similar to our primary analysis; however, the variance estimates increase substantially. Table~\ref{tab:medtab1a} displays the outcome analysis using logistic regression while Table~\ref{tab:medtab7} displays the results using XGBoost. For our mediation analysis, Table~\ref{tab:medtab1a} displays the results using logistic regression, and Table~\ref{tab:medtab7} displays the results using XGBoost. The XGBoost estimates are all statistically insignificant. This appears primarily driven by the more extreme complete-case weights that we estimate using these models, suggesting positivity violations. We present plots displaying the full data heterogeneity analysis below in Figures~\ref{fig:alloutcomes-full-glm}, ~\ref{fig:alloutcomes-full-xgb}, ~\ref{fig:htemedprops-glm-full}, ~\ref{fig:htemedprops-xgb-full}. 

  \begingroup%
  \renewcommand\normalsize{\footnotesize}% Specify your font modification
  \input{tables/risk-difference-table-full}%
  \endgroup%

  \begingroup%
  \renewcommand\normalsize{\footnotesize}% Specify your font modification
  \input{tables/risk-difference-table-xgb-full}%
  \endgroup%

  \begingroup%
  \renewcommand\normalsize{\footnotesize}% Specify your font modification
  \input{tables/mediation-table-one-full}  \endgroup%

  \begingroup%
  \renewcommand\normalsize{\footnotesize}% Specify your font modification
  \input{tables/mediation-table-seven}  \endgroup%

\begin{figure}[H]
\begin{center}
    \includegraphics[scale=0.5]{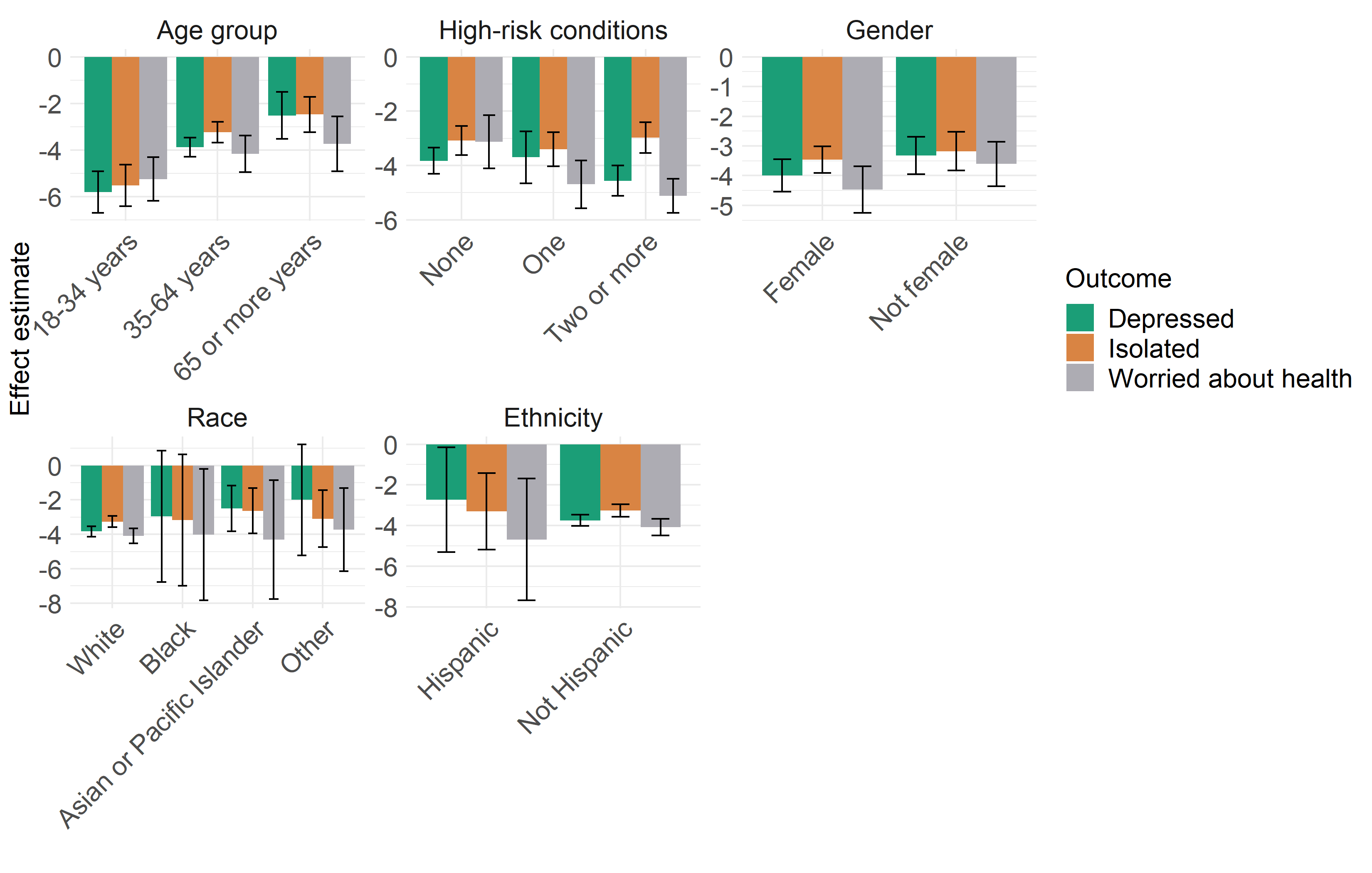}
    \caption{Outcome heterogeneity, full data estimates (GLM models)}
    \label{fig:alloutcomes-full-glm}
\end{center}
\end{figure}

\begin{figure}[H]
\begin{center}
    \includegraphics[scale=0.5]{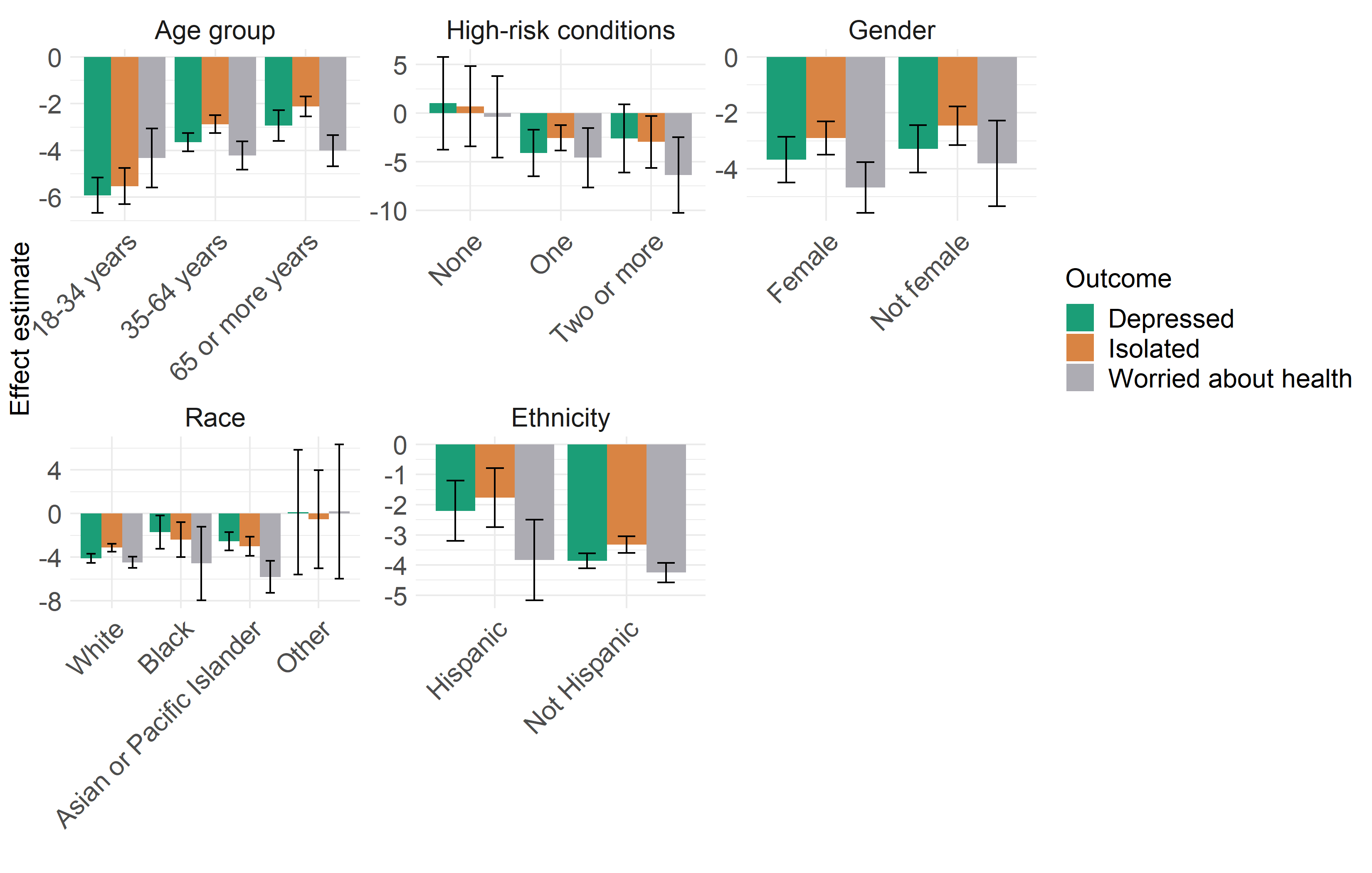}
    \caption{Outcome heterogeneity, full data estimates (XGBoost models)}
    \label{fig:alloutcomes-full-xgb}
\end{center}
\end{figure}

\begin{figure}[H]
\begin{center}
    \includegraphics[scale=0.5]{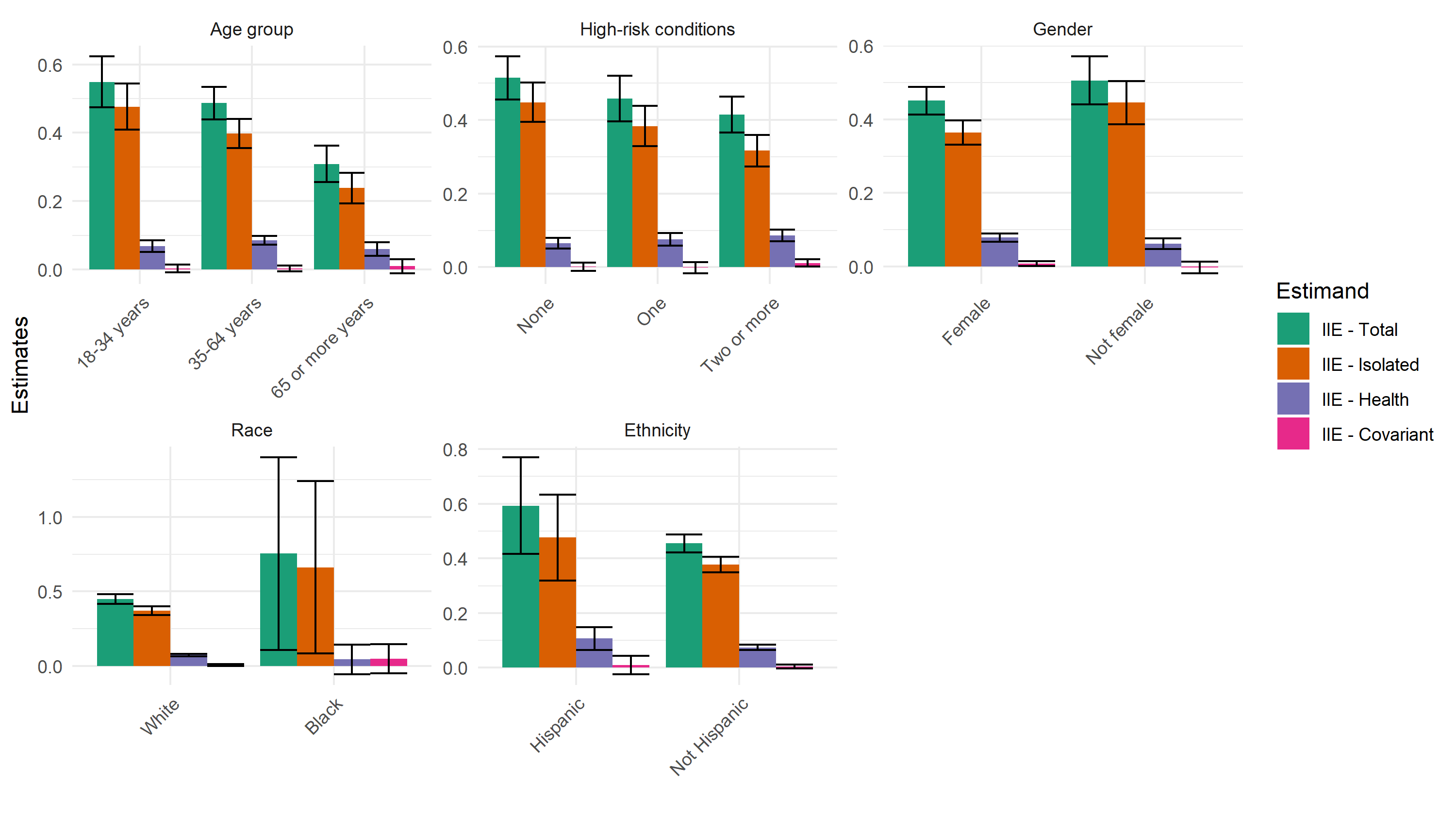}
    \caption{Mediation heterogeneity: full data estimates (GLM models)}
    \label{fig:htemedprops-glm-full}
\end{center}
\end{figure}

\begin{figure}[H]
\begin{center}
    \includegraphics[scale=0.5]{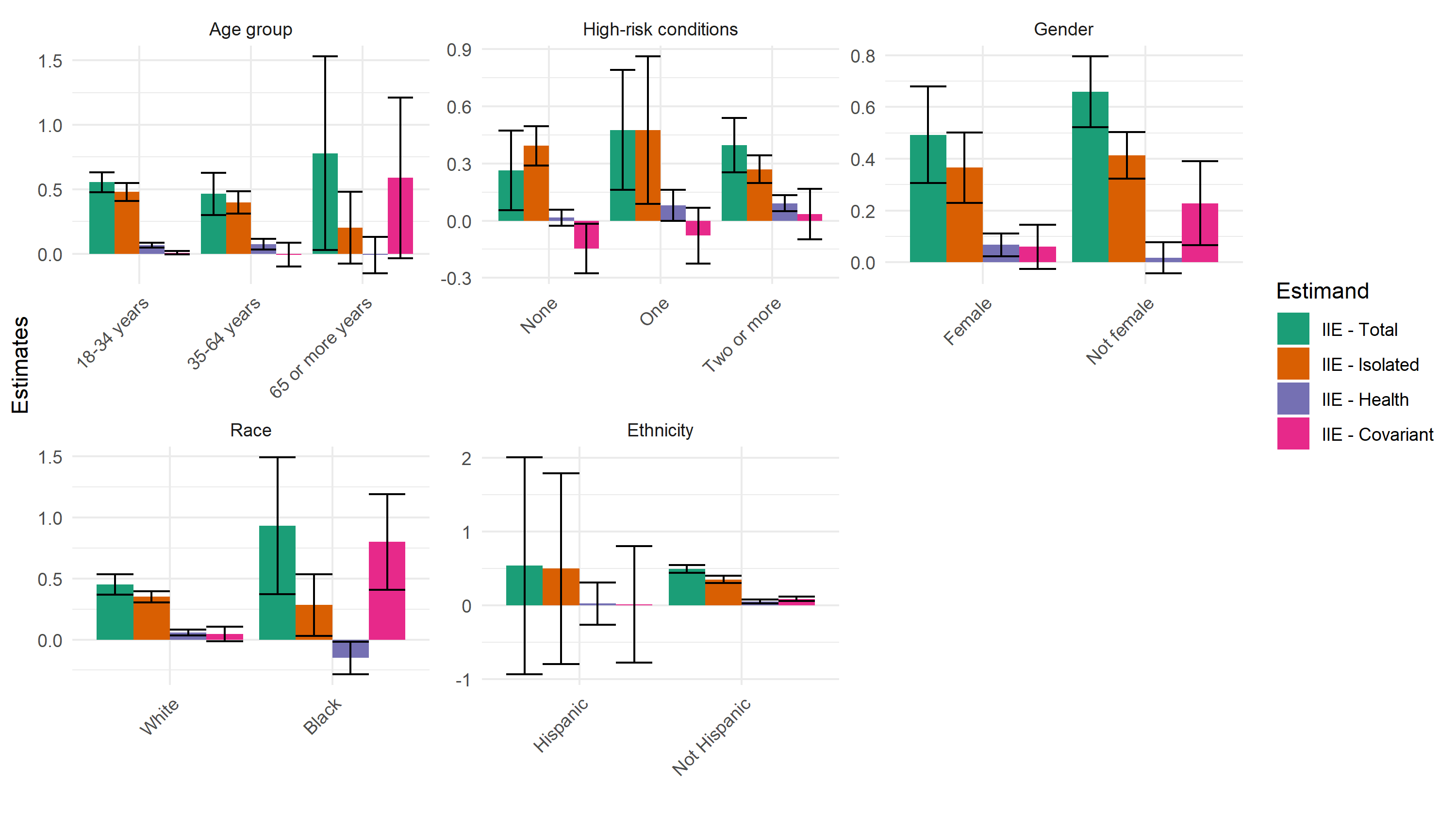}
    \caption{Mediation heterogeneity: full data estimates (XGBoost models)}
    \label{fig:htemedprops-xgb-full}
\end{center}
\end{figure}

\subsubsection{Bounds}

We can obtain upper bounds on $\psi^{va}$ by assuming that effects on the $R = 0$ stratum are not greater than zero on average. This removes the positivity requirement required to estimate $\psi^{va}$ -- \eqref{eqn:positivity3} in Appendix~\ref{app:identification}. 

Letting $p_r$ indicate the proportion of complete-cases over the total number of vaccine-accepting respondents, we can scale our estimates of $\psi$ by $p_r$ (approximately 86 percent) to obtain a valid upper bound on $\psi^{va}$.\footnote{This percent excludes non-respondents to the vaccine-hesitancy question. Technically, $p_r$ is not point identified.} We can generalize this approach to cover any other portion of the data that we do not include in our primary sample (e.g. those with missing ZIP code information, who live in Alaska, or who are vaccine-hesitant). Notably, we can never fully explain away our effect estimates under this assumption: we would scale them by at most 64\%, the complete-case percent of the initial sample.

We may also wish to make inferences while allowing our estimates to be biased, either due to violations of no unmeasured confounding or random sample-selection. Fortunately, the sensitivity parameters we estimated previously in Section~\ref{sec:limitations} remain valid for $\psi^{va}$. This follows since, by assumption, the effect on the unobserved portion of the data is not greater than zero. This logic generalizes to any portion of the data excluded from our primary analysis.

\subsection{Hesitant included}

Table~\ref{tab:senstab1} presents our primary estimates when including vaccine-hesitant respondents. Consistent with our expectations detailed in Section~\ref{sec:limitations}, the effect estimates for the outcomes anlaysis move closer to zero. 

  \begingroup%
  \renewcommand\normalsize{\footnotesize}% Specify your font modification
  \input{tables/senstab1}%
  \endgroup%

%\myinput{tables/mediation-table-two}

\subsection{Bad controls}

Table~\ref{tab:senstab5} presents our results when we do not control for occupation, employment status, or worries about finances. Consistent with our expectations detailed in Section~\ref{sec:limitations}, these results move further away from zero. 

  \begingroup%
  \renewcommand\normalsize{\footnotesize}% Specify your font modification
  \input{tables/senstab5}  \endgroup%

\subsection{Including hesitancy non-response}

Our primary analysis excludes all individuals who report vaccine hesitancy, or who do not respond to the question. Implicitly, this analysis assumes that all non-responses were vaccine-hesitant. We run our analysis when including the 75,568 respondents who did not respond to the vaccine hesitancy question (and who either did not respond to the vaccination question or indicated that they had not previously received the COVID-19 vaccine). However, this only adds 1,053 respondents to our analytic sample. Table~\ref{tab:senstab4} compares the primary outcome analysis to this analysis, and shows that the estimates are consistent. 

  \begingroup%
  \renewcommand\normalsize{\footnotesize}% Specify your font modification
  \input{tables/senstab4}  \endgroup%

%\myinput{tables/mediation-table-four}

\subsection{Omitting suspicious responses}

We also run our analysis excluding individuals who give suspicious responses. To identify these observations, we roughly follow a script used in \cite{bilinski2021better}, and exclude people who report any values that fall outside the range of the responses:

\begin{itemize}
    \item Household size between 0 and 30
    \item Number of sick in household between 0 and 30
    \item Number of work contacts between 0 and 100
    \item Number of shopping contacts between 0 and 100
    \item Number of work contacts between 0 and 100
    \item Number of social contacts between 0 and 100
    \item Number of other contacts between 0 and 100
    \item Number of high-risk health conditions less than all options minus 1
    \item Number of symptoms less than all options minus 1
\end{itemize}

These omissions account for approximately 2.5 percent of responses on our primary dataset, and approximately 2.1 percent of the complete-case subset. Table~\ref{tab:senstab2} displays the results when excluding these responses. The results are quite close to our primary estimates. Of course, our classifications of unreasonable answers might be incorrect, and any other forms of measurement error are also possible. However, our results demonstrate some robustness to the possibility of measurement error. 

  \begingroup%
  \renewcommand\normalsize{\footnotesize}% Specify your font modification
  \input{tables/senstab2}  \endgroup%

%\myinput{tables/mediation-table-three} 

%% file: tables/risk-difference-table-full.tex
\begin{table}[!h]

\caption{Effect estimates, GLM estimates, full data (95\% CI) \label{tab:full-glm}}
\centering
\begin{tabular}[t]{lllll}
\toprule
Outcome & Adjusted RD & Adjusted RR & Unadjusted RD & Unadjusted RR\\
\midrule
Depressed & -4.17 (-4.55, -3.78) & 0.78 (0.76, 0.80) & -9.82 (-9.98, -9.66) & 0.55 (0.54, 0.55)\\
Isolated & -3.09 (-3.43, -2.76) & 0.84 (0.83, 0.86) & -8.39 (-8.57, -8.22) & 0.62 (0.62, 0.63)\\
Worried about health & -4.59 (-5.12, -4.06) & 0.84 (0.82, 0.86) & -11.35 (-11.55, -11.16) & 0.64 (0.63, 0.64)\\
\bottomrule
\multicolumn{5}{l}{\rule{0pt}{1em}RD indicates risk-differences; RR indicates risk-ratio}\\
\end{tabular}
\end{table}

%% file: tables/risk-difference-table-xgb-full.tex
\begin{table}[!h]

\caption{Effect estimates, XGBoost models, full data (95\% CI) \label{tab:full-xgboost}}
\centering
\begin{tabular}[t]{lllll}
\toprule
Outcome & Adjusted RD & Adjusted RR & Unadjusted RD & Unadjusted RR\\
\midrule
Depressed & -2.06 (-4.20, 0.07) & 0.88 (0.77, 1.00) & -9.82 (-9.98, -9.66) & 0.55 (0.54, 0.55)\\
Isolated & -1.65 (-3.37, 0.08) & 0.91 (0.82, 1.00) & -8.39 (-8.57, -8.22) & 0.62 (0.62, 0.63)\\
Worried about health & -3.76 (-5.90, -1.61) & 0.86 (0.79, 0.93) & -11.35 (-11.55, -11.16) & 0.64 (0.63, 0.64)\\
\bottomrule
\multicolumn{5}{l}{\rule{0pt}{1em}RD indicates risk-differences; RR indicates risk-ratio}\\
\end{tabular}
\end{table}

%% file: tables/mediation-table-one-full.tex
\begin{table}[!h]

\caption{Mediation analysis: GLM estimates, full data (95\% CI) (95\% CI) \label{tab:medtab1a}}
\centering
\begin{threeparttable}
\begin{tabular}[t]{lllll}
\toprule
Effect & Ref = 0 & Proportion (Ref = 0) & Ref = 1 & Proportion (Ref = 1)\\
\midrule
IDE & -2.91 (-3.25, -2.57) & 69.77 (65.31, 74.24) & -1.90 (-2.29, -1.50) & 45.50 (39.38, 51.62)\\
IIE & -1.26 (-1.47, -1.05) & 30.23 (25.76, 34.69) & -2.27 (-2.37, -2.17) & 54.50 (48.38, 60.62)\\
IIE - Isolated & -1.53 (-1.70, -1.37) & 36.77 (32.89, 40.64) & -1.53 (-1.68, -1.38) & 36.64 (32.51, 40.78)\\
IIE - Health & -0.41 (-0.48, -0.34) & 9.90 (7.92, 11.88) & -0.25 (-0.31, -0.19) & 5.92 (4.25, 7.59)\\
IIE - Covariant & 0.69 (0.58, 0.79) & -16.44 (-19.75, -13.13) & -0.50 (-0.69, -0.30) & 11.94 (9.45, 14.43)\\
\bottomrule
\end{tabular}
\begin{tablenotes}
\item IDE represents the interventional direct effect; IIE represents the total interventional indirect effect; IIE - X represents interventional indirect effect via the indicated channel X; Ref represents the reference category for the indirect effects
\end{tablenotes}
\end{threeparttable}
\end{table}

%% file: tables/mediation-table-seven.tex
\begin{table}[!h]

\caption{Mediation analysis: XGBoost estimates, full data (95\% CI) \label{tab:medtab7}}
\centering
\begin{threeparttable}
\begin{tabular}[t]{lllll}
\toprule
Effect & Ref = 0 & Proportion (Ref = 0) & Ref = 1 & Proportion (Ref = 1)\\
\midrule
IDE & -1.58 (-3.10, -0.06) & 76.50 (34.24, 118.76) & -1.17 (-2.73, 0.39) & 56.73 (23.78, 89.69)\\
IIE & -0.48 (-1.65, 0.68) & 23.50 (-18.76, 65.76) & -0.89 (-1.19, -0.59) & 43.27 (10.31, 76.22)\\
IIE - Isolated & -0.58 (-1.67, 0.51) & 28.23 (-9.91, 66.37) & -0.88 (-1.73, -0.03) & 42.51 (4.10, 80.92)\\
IIE - Health & -0.47 (-0.79, -0.15) & 22.60 (-7.40, 52.60) & -0.27 (-0.48, -0.06) & 13.18 (0.98, 25.37)\\
IIE - Covariant & 0.56 (0.12, 1.01) & -27.33 (-67.20, 12.53) & 0.26 (-0.71, 1.22) & -12.42 (-35.01, 10.16)\\
\bottomrule
\end{tabular}
\begin{tablenotes}
\item IDE represents the interventional direct effect; IIE represents the total interventional indirect effect; IIE - X represents interventional indirect effect via the indicated channel X; Ref represents the reference category for the indirect effects
\end{tablenotes}
\end{threeparttable}
\end{table}

%% file: tables/senstab1.tex
% latex table generated in R 4.2.0 by xtable 1.8-4 package
% Fri Jan 27 13:43:32 2023
\begin{table}[ht]
\centering
\caption{Estimate sensitivity: hesitant included} 
\label{tab:senstab1}
\begin{tabular}{lll}
  \hline
Outcome & Analytic data & Hesitant included \\ 
  \hline
Depressed & -3.72 (-3.98, -3.47) & -3.01 (-3.29, -2.74) \\ 
  Isolated & -3.26 (-3.53, -2.99) & -2.55 (-2.83, -2.26) \\ 
  Worried about health & -4.26 (-4.55, -3.96) & -0.53 (-0.84, -0.22) \\ 
   \hline
\end{tabular}
\end{table}

%% file: tables/senstab5.tex
% latex table generated in R 4.2.0 by xtable 1.8-4 package
% Wed Feb  1 12:24:30 2023
\begin{table}[ht]
\centering
\caption{Estimate sensitivity: employment status, financial worry, 
 occupation not controlled} 
\label{tab:senstab5}
\begin{tabular}{ll}
  \hline
Outcome & Effects \\ 
  \hline
Depression & -5.53 (-5.72, -5.33) \\ 
  Isolation & -6.05 (-6.26, -5.85) \\ 
  Health & -6.90 (-7.14, -6.67) \\ 
   \hline
\end{tabular}
\end{table}

%% file: tables/senstab4.tex
% latex table generated in R 4.2.0 by xtable 1.8-4 package
% Fri Jan 27 13:43:33 2023
\begin{table}[ht]
\centering
\caption{Estimate sensitivity: missing hesitancy included} 
\label{tab:senstab4}
\begin{tabular}{lll}
  \hline
Outcome & Analytic data & Missing hesitancy included \\ 
  \hline
Depressed & -3.72 (-3.98, -3.47) & -3.72 (-3.97, -3.47) \\ 
  Isolated & -3.26 (-3.53, -2.99) & -3.26 (-3.53, -2.99) \\ 
  Worried about health & -4.26 (-4.55, -3.96) & -4.24 (-4.54, -3.94) \\ 
   \hline
\end{tabular}
\end{table}

%% file: tables/senstab2.tex
% latex table generated in R 4.2.0 by xtable 1.8-4 package
% Fri Jan 27 13:43:32 2023
\begin{table}[ht]
\centering
\caption{Estimate sensitivity: suspicious responses excluded} 
\label{tab:senstab2}
\begin{tabular}{lll}
  \hline
Outcome & Analytic data & Suspicious responses excluded \\ 
  \hline
Depressed & -3.72 (-3.98, -3.47) & -3.81 (-4.06, -3.55) \\ 
  Isolated & -3.26 (-3.53, -2.99) & -3.34 (-3.61, -3.07) \\ 
  Worried about health & -4.26 (-4.55, -3.96) & -4.16 (-4.46, -3.86) \\ 
   \hline
\end{tabular}
\end{table}